\documentclass[12pt]{article}
\usepackage{amsfonts}
\usepackage{amsmath}
\usepackage{bbm}
\usepackage{natbib}
\usepackage{geometry}
\usepackage[onehalfspacing]{setspace}
\usepackage{tikz}
\usetikzlibrary{shapes,arrows,positioning }
\usepackage{pgfplots}
\pgfplotsset{width=10cm,compat=1.9}
\usepackage{caption}
\usepackage{subcaption}
\usepackage{hyperref}
\usepackage{enumitem}
\newtheorem{innercustomthm}{Theorem}
\newtheorem{theorem}{Theorem}

\newtheorem{innercustomassump}{Assumption}
\newtheorem{innercustomremark}{Remark}

\newtheorem{claim}{Claim}

\newtheorem{corollary}{Corollary}

\newtheorem{definition}{Definition}

\newtheorem{lemma}{Lemma}

\newtheorem{proposition}{Proposition}

\newcommand{\dd}[1]{\mathrm{d}#1}
\newenvironment{proof}[1][Proof]{\textbf{#1.} }{\ \rule{0.5em}{0.5em}}

\usepackage{url}

\newcommand{\nc}{\textnormal{NC}}

\newcommand{\OPT}{\mathrm{OPT}}
\newenvironment{customthm}[1]
  {\renewcommand\theinnercustomthm{#1}\innercustomthm}
  {\endinnercustomthm}

\newenvironment{customassump}[1]
  {\renewcommand\theinnercustomassump{#1}\innercustomassump}
  {\endinnercustomassump}

\newenvironment{customremark}[1]
  {\renewcommand\theinnercustomremark{#1}\innercustomremark}
  {\endinnercustomremark}
  
\newcommand{\ubar}[1]{\underline{#1}}
\begin{document}

\title{Solving Problems of Unknown Difficulty

}

\author{Nicholas Wu\thanks{Microsoft Research New England, Cambridge MA 02139, nicktwu@stanford.edu \\ I am deeply indebted to my advisors Dirk Bergemann, Marina Halac, and Johannes H\"orner for their invaluable feedback. I also thank Yeon-Koo Che, Roberto Corrao,  Mira Frick,  Tan Gan, Ying Gao, Ryota Iijima, Ravi Jagadeesan, Navin Kartik, Jan Knoepfle, Elliot Lipnowski, Larry Samuelson, Philipp Strack, Juuso V\"alim\"aki,  and participants at seminars at Stanford GSB, University of Chicago, Northwestern Kellogg, Arizona State University, Ohio State University, Carnegie Mellon University, Stony Brook University, Purdue Daniels, Vanderbilt University, the Yale microeconomic theory lunch, the Harvard/MIT theory seminar, and Stony Brook International Game Theory Festival for helpful comments and discussions. I also especially thank Yifan Dai, Florian Mudekereza, and Tianyi Zhang for their close reading and feedback. A previous version of this paper was titled ``Strategic Idea Management.'' All errors are my own. 
}
}
\date{This version: November 3, 2026 \\
\vspace{0.5cm}
\textcolor{purple}{\href{https://drive.google.com/file/d/1FpMdTOS8OHkkr4jwCkm9a_3hajLMq13H/view?usp=sharing}{Click here for the most recent version.}}
}
\maketitle

\begin{abstract}
This paper studies how uncertainty about problem difficulty shapes problem-solving strategies. I develop a dynamic model where an agent solves a problem by brainstorming approaches of unknown quality and allocating a fixed effort budget among them. Success arrives from spending effort pursuing good approaches, at a rate determined by the unknown problem difficulty. The agent balances costly exploration (expanding the set of approaches) with exploitation (pursuing existing approaches). Failures could signal either a bad idea or a hard problem, and this uncertainty generates novel dynamics: optimal search alternates between trying new approaches and revisiting previously abandoned ones. I then examine a principal–agent environment, where moral hazard arises on the intensive margin: how the agent explores. Dynamic commitment creates a frontloading force, which can be counteracted by the presence of learning. The framework reflects scientific discovery, product development, and other creative work, providing insights into innovation and organizational design.
\bigskip

\end{abstract}

\newpage

\section{Introduction}
\begin{quote}
   \textit{ I haven't failed. I've just found 10,000 ways that don't work.} 
        \vskip1mm
        \hspace*\fill{\small--- Thomas Edison}
\end{quote}

According to the \textit{Job Outlook 2025} survey by the National Association of Colleges and Employers, the single most sought-after attribute employers seek on a candidate's resume is ``problem-solving'' ability.\footnote{See \hyperlink{https://www.training.nih.gov/oite-careers-blog/top-skills-employers-are-looking-for-in-2025-problem-solving-teamwork-and-communication/}{https://www.training.nih.gov/oite-careers-blog/top-skills-employers-are-looking-for-in-2025-problem-solving-teamwork-and-communication/}.} This emphasis reflects the fact that successes in science, innovation, and business depend critically on how individuals and organizations tackle uncertain problems; fundamentally, solving a problem is the smallest increment of innovation and growth. Understanding this process is essential both for predicting behavior in environments such as innovation, entrepreneurship, and research, and for designing contracts and institutions that guide problem-solving effectively. 

Not all challenges are created equally. Some are \textit{exercises}: tasks where the methods to solve them are clear. Other challenges are \textit{problems}: questions that one does not know how to answer, at least initially. By definition, solving problems requires searching for new approaches and strategies.\footnote{This distinction is articulated nicely by mathematician Paul Zeitz: \hyperlink{https://jrmf.org/blog/a-life-changing-revelation-problems-vs-exercises/}{https://jrmf.org/blog/a-life-changing-revelation-problems-vs-exercises/}.} In facing a problem, a solver does not know in advance which approaches will or will not work, and will have to be creative in inventing new solutions.

Consider an inventor experimenting with developing a new technology, who faces uncertainty not only over which prototype might succeed, but also over the tractability of the underlying idea. A computer scientist working on developing an AI model with an emergent property (e.g. artificial general intelligence) may not know necessarily which model architecture to use, nor at what scale the property might arise. A materials scientist attempting to synthesize a material with a desired property (e.g. superconductivity) similarly may not know which type of crystalline structure to synthesize, nor what fraction of possible configurations might exhibit the macroscopic property.
In each example, the uncertainty lies not just in which approach might work, but in how hard the problem is overall. How does uncertainty about problem difficulty affect the way we pursue possible solutions?

I study this question through a dynamic model of problem solving. An agent aims to achieve a one-time breakthrough as quickly as possible. At each moment, they choose how to allocate a unit of effort to pursuing existing approaches. Additionally, they decide when to undertake costly exploration, which expands the set of possible approaches. Approaches are valid or not; effort on a valid approach stochastically produces a breakthrough. Each approach's validity is drawn independently of the validity of other approaches; there could be multiple valid ways to solve the problem. However, the agent does not observe which approaches are valid and which are not. Additionally, the difficulty of the problem is unknown: if the problem is easy, the rate at which a valid approach yields success is high, while if it is hard, valid approaches yield success more slowly. As the agent experiments and observes failures, they update their beliefs about the problem’s difficulty and about approach validity, and this evolving belief shapes the balance between creating and exploring new possible solutions and exploiting existing approaches.

A key feature in solving problems of unknown difficulty is that a lack of success is ambiguous. When an attempted approach does not succeed, the solver cannot tell whether the failure reflects a poor choice of approach or whether the problem itself is inherently difficult. This ambiguity creates a tradeoff. On the one hand, if the solver attributes the lack of success to a bad approach, the solver would like to move on to a different approach. On the other hand, the lack of success may simply have been because the problem was difficult, so pursuing the approach further may still be worthwhile. I show that the optimal policy in such environments involves balancing learning about the viability of potential solutions and the difficulty of the problem. In the optimal policy, the solver alternates between exploring new approaches and revisiting old approaches as new information accumulates. This pattern contrasts with a benchmark in which problem difficulty is known, where abandoned approaches are never revisited.

The decision problem departs from standard multi-armed bandit models in a crucial way: the difficulty of the problem is a common factor across all approaches. This correlation means that effort on one approach shifts beliefs about every other approach as well. As a result, the decision maker’s problem is not separable across arms, but a correlated bandit with dynamic interdependence. This structure complicates both the characterization of optimal policies and the derivation of the solution, and standard index results do not apply. To analyze the decision problem under uncertain difficulty, I derive the optimal policy using interchange arguments and two key insights; that (1) the optimal policy ranks existing approaches relative to each other based on the historical effort spent pursuing that approach, and (2) that the amount of effort that must be spent on pursuing an existing approach before the agent is willing to explore a new approach increases over time. 

Having established in the decision problem this exploration-exploitation balance between exploring new approaches and further pursuing older approaches, I then move to consider how economic forces shape \textit{how} agents solve problems. In order to do so, I will characterize a more tractable limit model that preserves the key tradeoffs from the baseline. While the baseline decision problem captures the essential single-agent tradeoffs, the discreteness of the problem in the number of approaches means that the distribution of success times is not differentiable at the times where the agent generates new approaches, which makes the baseline model difficult to embed into strategic settings. 
To circumvent this, I develop a version of the decision problem with a continuum of approaches that still retains the core breadth-depth tradeoff, where the agent instead maintains a breadth of their search (corresponding to the number of approaches tried) and balances breadth with the depth of their search (corresponding to how much effort is spent on any given approach). This continuum model ``smooths'' out the baseline decision problem, and I microfound the model by showing that this continuum model arises as a limit of the baseline decision problem.

I apply the continuum model to study the effect of agency frictions on how agents solve problems. I focus on a principal-agent model where the agent is tasked with solving a problem but requires support from the principal; for example, consider a venture capitalist funding an entrepreneur working to develop a new technology. The contracting instrument available to the investor is a possibly time-varying share of the value of success. In contrast to standard contracting environments, I suppose moral hazard arises because the principal and agent disagree on an intensive margin of effort rather than an extensive margin. That is, the investor and entrepreneur are aligned in that the entrepreneur wants to work; however, they disagree about \textit{how} the entrepreneur should be tackling the problem. I show that the agency friction induces the entrepreneur to search less broadly than the first-best; the entrepreneur spends too long per approach and tries new approaches (i.e., pivots) less than a planner would like. The fundamental tradeoff in the principal–agent model is between providing the entrepreneur with strong enough incentives to broaden search, and avoiding excessive dilution of the principal’s payoff from success.

I then characterize the optimal time-varying share under interior-solution conditions, and highlight two countervailing forces that arise. 
First, I illustrate how the nature of the moral hazard problem introduces an incentive for the principal to frontload incentives. To isolate this effect, I show that without learning, the principal decreases the entrepreneur's equity over time. Intuitively, this occurs because the investor offers the entrepreneur a larger stake in earlier success to induce the entrepreneur to try more new approaches early, which not only accelerates success in the short-term, but also generates a persistent increase to the rate of success in the future even if the new approaches do not immediately yield success. This contrasts with classic results in dynamic moral hazard, which find that when the moral hazard appears in an extensive margin (i.e. the agent faces work/shirk decisions), the principal backloads incentives. This force is also distinct from the dynamic agency costs studied in the experimentation literature, where frontloading arises because the agent wants to prolong the project while the principal wants the project to succeed as rapidly as possible; in this paper, both the agent and the principal are aligned in their preferences over the duration of the project. 

I finally illustrate how learning produces an opposing force for the principal to want to backload contracts; here, because the agent and the principal both become more pessimistic about difficulty when time passes without success, the principal becomes willing to relinquish more of the share to encourage the agent to continue exploring. I demonstrate how these counteracting forces can yield non-monotonicities in the optimal interior contract, which sometimes implies the contract has a tolerance (or even reward) for early failure. 

The paper is structured as follows. The next section surveys the related literature. Section \ref{sec:model} provides the formalism for the baseline model. Section \ref{sec:analysis} analyzes the baseline decision problem, contrasting a benchmark setting with known difficulty in Section \ref{sec:benchmark} with the unknown difficulty case in Section \ref{sec:learning}. Section \ref{sec:continuum} then presents the limiting continuum model, which then enables analysis of the principal-agent problem in Section \ref{sec:agency}. Section \ref{sec:robustness} discusses robustness, and Section \ref{sec:conclusion} concludes.
\subsection{Related Literature}
This paper contributes to an emerging literature connecting bandit models of experimentation to models of exploration over a complex landscape. The former focuses on studying the question of what experimentation dynamics look like over time, and typically feature an exogenously given risky option. The latter focuses on modeling what alternatives an agent chooses to explore. In this paper, I study their interaction; how the dynamics of what is learned from experimentation influence how an agent explores new approaches.

The first strand of the literature on strategic experimentation, dating back to \cite{bh1999} and \cite{krc2005}, focuses on studying dynamic behavior of strategic agents faced with a given risky endeavor. This line of the literature typically focuses on the bandit problem where agents face a choice between a risky arm and a safe arm, characterizing how this choice evolves over time and in response to strategic environments. Most of these papers focus on an exogenously given choice set of arms, with a few notable exceptions. \cite{cs2023} allow (at most) one new arm to be generated, but instead focus on a time-constrained agent who manages the risk of not finding a new approach with the time pressure induced by a deadline. They provide a rationale for ``false starts'' (where the agent revisits the initial arm), but this stems from the exogenously given time pressure, and not as a result of dynamic learning as is the case in my model. \cite{ko2003}, \cite{s2021}, \cite{d2024} and \cite{fp2024} allow arbitrarily many new arms to be generated; however, the models in \cite{s2021}, \cite{d2024}, and \cite{fp2024} rely on independence assumptions between the arms that do not apply to my model. \cite{ko2003} derives a condition for a Gittins index to hold even when there is some correlation in branching arms, but this condition does not apply either. In my model, uncertainty about problem difficulty results in information spillovers between arms, where learning about one arm \textit{does} influence the beliefs and payoff assessments about other arms. In particular, as a result of the uncertainty about difficulty, the belief an agent has about the state of an arm (validity of an approach) ceases to be independent of the beliefs over the states of other arms after the arm is pulled. This distinction results in qualitative differences in the optimal strategy, including featuring ``false starts'' and revisiting old approaches, which is not present in the benchmark or the related model of \cite{s2021}. 

The second strand of relevant literature incorporates correlation among outcomes. This line started from the seminal \cite{c2011} model, which models exploration as observing points on a Brownian motion, so nearby observations are spatially correlated. These models typically answer questions related to where agents explore; notably, \cite{cs2025} applies this framework to think about what kinds of questions researchers investigate. The closest related experimentation models with a Brownian form of correlation are \cite{l2021} and \cite{uy2025}. In their environment, the exploration technology is exogenously given, but features \textit{temporal} correlation; the payoff of a new alternative is observed and correlated with recent alternatives. This literature either supposes the agents observe outcomes of exploration deterministically, or supposes agents are short-lived, as in \cite{cs2025}. This paper contributes to this literature by offering a model of exploration with imperfect observation and learning dynamics; in my model, the outcomes from exerting effort on the agent's approaches are correlated, but at each instant, a lack of success only generates \textit{imperfect} information about the value of the approach, which influences the long-lived agent's future exploration patterns.

Since the paper considers a dynamic moral hazard problem that arises when a principal contracts with an agent who faces the baseline decision problem modeled, it additionally contributes to the literature on dynamic moral hazard. The canonical dynamic moral hazard models \citep{r85,ss87,s08} suppose the agent faces a work/shirk choice every period, so the moral hazard arises in the extensive margin of how much effort the agent exerts. This paper offers a different perspective by studying an intensive margin of \textit{what} the agent works on (namely, how the agent explores new approaches), which dates back to a notion of moral hazard from \cite{hm1991}. This paper shows that the distinction in the type of moral hazard yields different implications for optimal contracts; while optimal contracts that focus on the extensive margin of effort are backloaded \citep{r85, ss87, s08}, I show that contracts focusing on this intensive margin of effort are frontloaded. This paper highlights an economic mechanism for this feature distinct from that of other papers on contracting over experimentation \citep{bh2005, hs2013}. In those papers, frontloading arises because of the agent has an incentive to prolong the project; by contrast, in this paper, the agent and principal both agree that they would like success to arrive as soon as possible. Complementary to \cite{hkl2016}, which is focused on the interaction of adverse selection and moral hazard, this paper instead focuses on a distinct notion of moral hazard, not its interaction with adverse selection. Finally, \cite{m2011} examines the design of incentives in a two period model to induce an agent to explore a single risky bandit arm, finding that inducing exploration may involve paying the agent for failure. While \cite{m2011} focuses on the tension between encouraging exploration and sustaining the agent's incentives to work, I instead focus not on the tension with sustaining incentives to not shirk, but on the dynamics relating exploration to commitment and learning.

Lastly, there are two other related papers featuring sufficiently distinct models that complement the analysis in this paper. \cite{lsy2024} focus on separating the information arrival process from the payoff process; their agent is allowed to allocate attention to learning about one arm despite getting payoffs from pulling another arm. In this paper, I instead maintain the standard coupling of payoffs to information, where the agents learn by observing their payoffs; instead, the information spillover in my paper stems not from the ability to allocate attention to other arms, but because of the correlation induced by unknown difficulty. 
\cite{lw2024} features a continuum of arms for the agents to choose from, and study a competition model where agents search over a fixed, but shared, continuum of arms. In complement to their work, this paper instead focuses on disjoint, but \textit{expanding} search spaces, where not all potential alternatives are available at the outset but instead agents must forgo costs to expand the set of arms available to them.
\section{Baseline Model}
\label{sec:model}
\paragraph{Exposition}  
An agent dynamically allocates a unit of effort with the end goal of solving a problem. The agent initially has no idea how to find a solution, but can spend a cost to brainstorm a new approach at any time. The new approaches are either valid or flawed, but the agent a priori does not know whether any given approach is valid or not. The agent succeeds by exerting on a valid approach; on a valid approach, a breakthrough (success) arrives at a Poisson rate, which ends the game and yields a payoff (normalized to 1). Pursuing a flawed approach never yields success. Additionally, the agent is uncertain whether the problem they face is easy or hard; if the problem is easy, they generate breakthroughs faster on \textit{all} approaches than if the problem is hard.

More precisely, the agent faces a continuous-time effort allocation problem over a multi-armed bandit, with a special action that creates a new approach. Time is continuous, $t \in [0, \infty)$. At time 0, the agent has no approaches (arms). Whenever the agent chooses, they can pay a cost $c$ to undertake the special action (``brainstorm''), which instantaneously creates a new approach.\footnote{In the analysis, I will suppose the cost is low enough that the agent is ``willing to undertake the problem'' which implicitly supposes the agent can choose not to solve the problem at outside option zero. Since it is not important to the analysis, I have abstracted away from an explicit description of the outside option.}
When an approach is brainstormed, it is independently drawn to be valid with probability $\nu_0$ for (v)alidity, and flawed with complementary probability; this is unobserved by the agent, so the agent does not know whether approaches are valid. The unknown state of the approaches can be summarized by a sequence of independent Bernoulli random variables with parameter $\nu_0$; denote $\{ \omega_i \}_{i \in \mathbb{N}_+}$ as the realization of such a sequence. The agent is additionally uncertain about whether the problem is easy or hard: there is an unobserved state $\theta \in \{E, H\}$, drawn $H$ with probability $\delta_0$ for (d)ifficulty. Difficulty influences the Poisson rate at which breakthroughs arrive $\lambda_\theta$; assume $\lambda_E \ge \lambda_H \ge 0$.

At some $t$, if the agent has $n$ approaches, the agent's action choice is a vector $k = (k_1, k_2, ..., k_n)$ such that $\sum_{i=1}^n k_i = 1$, where I interpret the action vector as allocating a fraction $k_i$ of their effort to pursuing approach $i$. By pursuing approach $i$ with effort $k_i$, the agent solves the problem at Poisson rate $\lambda_\theta \omega_i k_i$; note that the agent can only solve the problem using approach $i$ iff approach $i$ is valid ($\omega_i = 1$). 
As shorthand, I will use the notation $[N]$ to denote the set $\{1, 2, ..., N\}$.

\paragraph{Payoffs} The agent discounts the future at rate $r > 0$, and only earns a payoff of size 1 when success arrives; if success occurs at time $\tau$, and the agent brainstormed new approaches at times $\{ \tau_1, \tau_2, ..., \tau_N \}$, then the agent's payoff is 
\[e^{-r\tau} - \sum_{i=1}^N e^{-r\tau_i} c. \] Note that normalizing payoff size to 1 is without loss of generality.

\paragraph{States and Strategies} The state of the problem can be summarized by the number of approaches brainstormed and the total effort exerted on each approach.\footnote{Unlike the precedent in \cite{bh1999} and \cite{krc2005}, I am not using beliefs as the state object. This is partly because the correlation structure induced by a given history of effort makes the laws of motion for beliefs quite difficult to handle. However, the historical effort vector is a sufficient statistic to recover all beliefs; given a vector of historical efforts, it is straightforward to compute the beliefs at that history from Bayes' rule. } That is, with a slight abuse of notation, I refer to the initial state with no approaches as state $\emptyset$, and let the vector $s = (K_1, K_2, ..., K_N)$ denote the state where $N$ approaches have been brainstormed, and approach $i$ has had total effort $K_i \ge 0$ exerted on it in the past. The state space is the set of all such possible vectors, of any possible dimension $N \in \mathbb{N}$: 
\[ \Sigma := \{ \emptyset \} \cup \left( \bigcup_{i=1}^\infty \mathbb{R}^i_+ \right) . \]
In words, the state space consists of all effort vectors of arbitrary length, and the initial state.\footnote{The state space characterization implicitly assumes the vector of efforts is ordered according to the times when the corresponding approaches were brainstormed. The fact that the state space is ordered does not matter for the results.}
At some state $s = (K_1, K_2, ..., K_N) \in \Sigma$, the allowable actions are $A(s) := \{ k \in \mathbb{R}^{N}_+ \mid \sum_{m=1}^{N} k_m = 1 \} \cup \{ \textnormal{brainstorm} \}$, so the only initially feasible action is  just $A(\emptyset) :=\{ \textnormal{brainstorm} \}$. Denote the set of all actions as $A = \cup_{s \in \Sigma}A(s)$. 

The state evolution is as follows: if the agent takes $a = (k_1, k_2, ..., k_N)$ at state $s = (K_1, K_2, ..., K_N)$ for a time interval of length $\Delta$, the subsequent state is $(K_1 + k_1\Delta, K_2 + k_2\Delta, ..., K_N + k_n\Delta)$. If the agent chooses to brainstorm at state $s = (K_1, K_2, ..., K_N)$, then the state immediately transitions to the state with a new approach brainstormed with zero previous effort: $(K_1, K_2, ..., K_N, 0)$.

A Markov strategy for the agent is then a function from states to allowable action vectors; that is, a strategy $\sigma: \Sigma \to A$  such that each state $s \in \Sigma$ maps to a feasible action: $\sigma(s) \in A(s)$.\footnote{Focusing on Markovian strategies is without loss here, since the agent faces a Markov decision problem (MDP). I will focus on strategies that induce  action paths which are c\`adl\`ag in time.} 
Borrowing terminology from the bandit literature, I will use the phrase ``pulling arm $i$'' to mean exerting positive effort on approach $i$.

\paragraph{Discussion} 
I modeled the cost as a fixed cost for obtaining a new idea; this applies to mental costs that might arise with creativity or to startup costs that might arise from starting a new project. The effort allocation was modeled under a unit constraint; this is to highlight that I focus on environments where there are many possible approaches which the agent can discover, so the scarcity arises in the agent's effort and not in the approaches available. The conclusions of the main section of the analysis would still hold under a time cost for new approaches, if one were to interpret the cost as arising from the time required to get ``up to speed'' on a new approach. 

I have also modeled the approaches as being ex-ante identical. I defer a discussion of heterogeneity to Section \ref{sec:heterogeneity}, but remark that there exist research endeavors where the approaches are not ex-ante differentiated; for example, high-throughput screening in drug development tests thousands of compounds against biological targets in a ``spray-and-pray'' search for drug candidates, and AI model architecture design involves picking layer arrangements, neural network configurations, and hyperparameters without much ex-ante information about which might work best. 

Finally, the assumption that the validity of approaches is independent can be relaxed to allow some correlation; a discussion of model generalizability is deferred to Section \ref{sec:generalization}.

\section{Decision Problem Analysis}\label{sec:analysis}
In this section, I characterize the solution to the single-agent decision problem described by the baseline model. I will first provide a benchmark case where there is no learning about difficulty before analyzing the learning case, and show that there are qualitative differences in the structure of the optimal strategy that arise as a result of learning about problem difficulty.

\subsection{Benchmark: Known Difficulty}\label{sec:benchmark}
First, I analyze the case with no aggregate uncertainty: $\lambda_E = \lambda_H = \lambda > 0$. Assume the cost of brainstorming approaches is sufficiently low: 
\[ c < \frac{\nu_0 \lambda}{r + \lambda}. \]
Intuitively, the assumption supposes that the expected payoff from brainstorming a single approach (and working on it forever) exceeds the cost of brainstorming the approach. If this condition is violated, the agent does not attempt the problem (never decides to start the first approach), so I focus on the interesting case where the agent is at least willing to start.

In this benchmark, the agent's only source of uncertainty is whether the approaches are valid or not. In particular, since the validity of an approach is drawn independently, the agent's belief about whether an approach is valid or not is only a function of their own historical effort on that approach. That is, define
\begin{equation} \label{eqn:interim_belief} \nu(K) := \frac{\nu_0 e^{-\lambda K} }{ \nu_0e^{-\lambda K} + 1 - \nu_0}, \end{equation} as the posterior belief that an approach is good when $K$ historical effort has been exerted on the approach and no breakthrough has arrived. The belief $\nu(K)$ is decreasing in $K$; intuitively, the more historical effort exerted without success, the lower the agent's belief on the approach's validity.

Let $B(s) = \arg \min_{i \in [N]} K_i$ be the approaches with least total effort when the state is $s = (K_1, K_2, ..., K_N)$; i.e., the (b)est approaches available.
\begin{proposition}\label{prop:benchmark}
    Suppose at some state, $N$ approaches have been discovered. Let $K^*$ denote the unique solution to 
    \begin{equation} \sum_{n=0}^\infty \left( e^{-rK^*} (1-\nu_0 + \nu_0e^{-\lambda K^*})\right)^n \left[ \nu_0 \int_0^{K^*} \lambda e^{-\lambda t} [e^{-rt}] \dd{t}  - c \right]  = \frac{\lambda \nu (K^*)}{\lambda \nu (K^*) + r} . \label{defn:taustar} \end{equation}
    Then the optimal strategy brainstorms a new approach if $\min_{i \in [N]} K_i \ge K^*$, else exerts effort on all the approaches in $B(s)$ with effort $1/|B(s)|$.
\end{proposition}
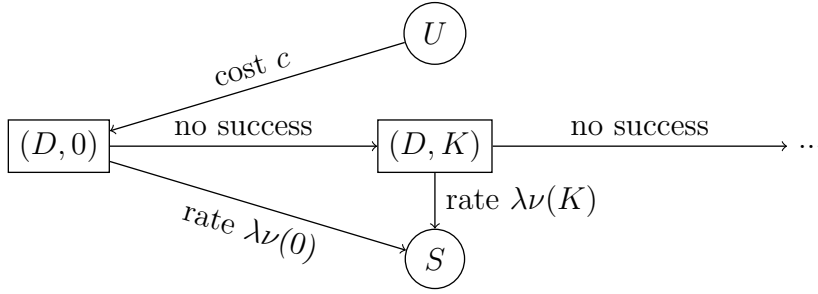
\begin{figure}
\centering
    \begin{tikzpicture}[->, line width=0.5 pt,level distance=1.5cm, sibling distance=5cm, edge from parent/.style={draw,->}]
\node [circle, draw ] (U) {$U$}
    child {node[rectangle,draw] (D1) {$(D,0)$}  edge from parent[draw=none]}
    child {node[rectangle,draw] (D2) {$(D,K)$} edge from parent[draw=none]
        child {node[circle,draw] (S) {$S$}  edge from parent[draw=none]}
    }
    child {node (D3) {...} edge from parent[draw=none]};
\draw[->] (D1) --  node[above] {no success} (D2) ;
\draw[->] (D2) --  node[above] {no success} (D3) ;
\draw[->] (U) --  node[sloped, above] {cost $c$} (D1) ;
\draw[->] (D1) -- node[sloped,below] {rate $\lambda \nu(0)$} (S);
\draw[->] (D2) -- node[right] {rate $\lambda \nu(K)$} (S);
\end{tikzpicture}
\caption{ Markov chain representation of the individual state of an arm. Each arm starts in an ``undiscovered'' state $U$, which transitions to a discovered state $D$ immediately with cost $c$; effort on a discovered state increases the historical effort associated with the arm, and an arm with historical effort $K$ transitions to a solved state at rate $\lambda \nu(K)$, where $\nu(K)$ is the posterior belief that the arm is valid given $K$ historical effort defined in \eqref{eqn:interim_belief}. }
\label{fig:gittins_rep}
\end{figure}
The proof is a straightforward application of the Gittins' index theorem on the state representation of an arm given by Figure \ref{fig:gittins_rep}. The formal details are left to the appendix.

To understand the threshold condition \eqref{defn:taustar}, observe that the right-hand side is the marginal time-weighted return to effort on an arm after exerting $K$ effort on it. The left-hand side is the value of sequentially pulling each arm for exactly $K^*$, which is also the initial value of the problem. Intuitively, when the agent is ready to brainstorm a new approach, (i.e., the value is equal to the left-hand-side), the agent is exactly indifferent between brainstorming a new approach and pursuing the previous approach an instant longer.

The benchmark result should not be surprising; Proposition \ref{prop:benchmark} indicates that the path of research effort sequentially tackles approaches. Additionally, the optimal policy never splits effort between multiple approaches (``task-juggling'').

\begin{corollary}\label{corr:seq}
    The optimal effort path induced by the optimal strategy is to sequentially brainstorm a new approach, work only on the new approach for total effort $K^*$, and then move on. The agent never simultaneously pursues approaches on-path, and once the agent stops working on an approach, it is never revisited. 
\end{corollary}

These observations echo the findings in the related model of \cite{s2021}, which solves a model where instead of a cost of approaches $c$, the approaches are implicitly costly through time delay. In the next section, I will show that with learning about difficulty, the features described in Corollary \ref{corr:seq} are no longer true of the optimal path; that is, the agent will both work on approaches simultaneously and revisit approaches.

\paragraph{Comparative Statics}
To illustrate how the effort threshold varies in the primitives of the model, I prove the following result:
\begin{proposition}\label{prop:baseline_comp}
    The threshold $K^*$ is increasing in $c$ and decreasing in $\lambda$.
\end{proposition}
The proposition follows from \eqref{defn:taustar} and the implicit function theorem, and the proof is left to the appendix.

Intuitively, if the agent experiences a higher brainstorming cost $c$, the agent values already discovered approaches relatively more than brainstorming new approaches. Conversely, if approaches should yield results quickly (higher $\lambda$), the agent becomes more willing to move on to new approaches. 

\paragraph{Remark on Time Pressure} Because I model the cost as a fixed mental cost of brainstorming, the threshold $K^*$ is non-monotone in the discount factor $r$. See Figure \ref{fig:r_comp_static} for an example. When the agent is very patient, increasing the time pressure on the agent can actually induce the agent to brainstorm approaches faster (i.e., be \textit{more} creative). Intuitively, when the agent is extremely patient ($r$ very close to zero), the agent spends a lot of time working on every approach because the benefit of a faster success from a new approach is small; thus, by increasing the time pressure slightly, the agent moves on to new approaches faster. However, if $r$ is large, increasing the degree of impatience makes the agent more willing to procrastinate on brainstorming new approaches because the benefit of delaying the cost of brainstorming becomes larger.

This is in contrast to Proposition 2 in \cite{s2021}, who shows that when the cost of a new approach is given by a delay, the time spent on an approach is \textit{always} increasing in $r$. This non-monotonicity is consistent with findings in the empirical psychology literature, that creativity is highest under intermediate time pressure  \citep{bo2006}.

\begin{figure}
    \centering
    \includegraphics[width=0.55\linewidth]{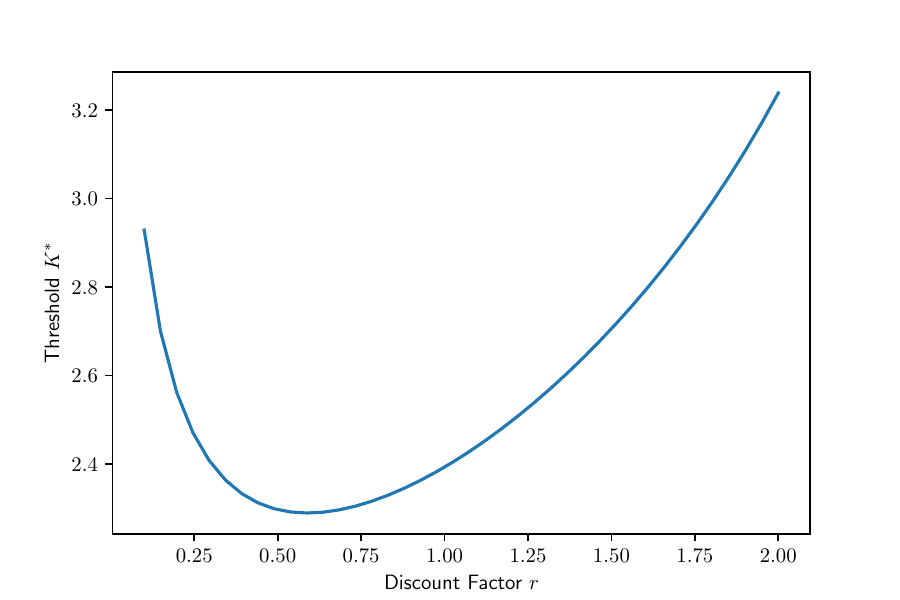}
    \caption{Non-monotone effect of time pressure on the threshold $K^*$. Example plotted for $\lambda = 1$, $c = 0.2$, $\nu_0 = 0.75$.}
    \label{fig:r_comp_static}
\end{figure}

\subsection{Learning About Difficulty}\label{sec:learning}
Now, consider the learning case, where $\lambda_E > \lambda_H > 0$. In this case, assume similarly that the cost is sufficiently low that the agent is willing to brainstorm at least a single approach:
\begin{equation}
    c < \nu_0 (1 - \delta_0)\frac{\lambda_E }{r + \lambda_E } + \nu_0 \delta_0 \frac{\lambda_H}{r + \lambda_H}. \label{eqn:learning_c_assum}
\end{equation} 
The restriction above ensures that the problem is interesting; if \eqref{eqn:learning_c_assum} is violated, then the agent would never brainstorm even at time zero, and the solution is trivial. The main result in this section characterizes the path of research effort under the presence of learning both about problem difficulty and about individual approaches.

Critically, learning about problem difficulty breaks the independence assumption necessary to characterize the optimal strategy using the Gittins index; effort on one approach produces information about the difficulty of the problem, which influences the agent's evaluation of other approaches as well. To illustrate this and understand the nature of the correlation, consider state $s = (K_1, K_2)$ with two approaches. Let the survival probability $S_\theta(K) = 1- \nu_0 + \nu_0e^{-\lambda_\theta K}$ be the probability that $K$ effort is exerted on an arm without success when the state is $\theta$. Consider the posterior belief that approach $1$ is valid ($\omega_1=1$); by Bayes' rule, that belief is 
\[ \frac{\delta_0 \nu_0 e^{-\lambda_H K_1} S_H(K_2) + (1-\delta_0)\nu_0 e^{-\lambda_E K_1} S_E(K_2) }{\delta_0 S_H(K_1)S_H(K_2) + (1-\delta_0)S_E(K_1)S_E(K_2) } = \frac{\delta_0 \nu_0 e^{-\lambda_H K_1} \frac{S_H(K_2)}{S_E(K_2)} + (1-\delta_0)\nu_0 e^{-\lambda_E K_1}  }{\delta_0 S_H(K_1)\frac{S_H(K_2)}{S_E(K_2)} + (1-\delta_0)S_E(K_1) }.\]
In particular, the belief over whether approach 1 is valid depends on the effort exerted on approach 2 as well; further, note that the ratio $S_H(K_2)/S_E(K_2)$ is initially increasing in $K_2$ when $K_2$ is close to zero, since a hard problem is less likely to yield a solution on any given approach than an easy problem. Since approach 1 with $K_1$ historical effort is more likely to be valid if the problem is hard than if the problem is easy, then as $K_2$ initially increases from zero, the posterior belief over whether approach 1 is valid \textit{increases}. Figure \ref{fig:belief_illustration} plots this belief movement. \footnote{The illustration highlights that the bandit problem features \textit{restless} arms, meaning the beliefs about one arm move while pulling another arm. The ``restlessness'' introduces permanent \textit{endogenous} correlation between the arms; one can easily check that the beliefs over the validity of different arms becomes correlated, with covariance depending on the history of effort over the arms.}

\begin{figure}
    \centering
    \begin{tikzpicture}[scale=0.75,
    declare function={
    func1(\x)= (0.5*0.5*exp(-1)*(1-0.5+0.5*exp(-\x))+0.5*0.5*exp(-2)*(1-0.5+0.5*exp(-2*\x)))/(0.5*(1-0.5+0.5*exp(-1))*(1-0.5+0.5*exp(-\x))+0.5*(1-0.5+0.5*exp(-2))*(1-0.5+0.5*exp(-2*\x)));
  }, declare function={
    func2(\x)= (\x<=1.05995) * 0.75 + and(\x>1.05995, \x<=2.1199)  * (0.5*(0.25+0.75*exp(-1.05995))*0.75*exp(-(\x-1.05995)) + 0.5*(0.25+0.75*exp(-2*1.05995))*0.75*exp(-2*(\x-1.05995)))/(0.5*(0.25+0.75*exp(-(\x-1.05995)))*(0.25+0.75*exp(-1.05995)) + 0.5*(0.25+0.75*exp(-2*(\x-1.05995)))*(0.25+0.75*exp(-2*1.05995)))   +
     and(\x>2.1199,  \x<=2.249062) * (0.5*0.75*exp(-(1.05995+0.5*(\x-2.1199)) )*(0.25+0.75*exp(-(1.05995+0.5*(\x-2.1199)))) + 0.5*0.75*exp(-2*(1.05995+0.5*(\x-2.1199)))*(0.25+0.75*exp(-2*(1.05995+0.5*(\x-2.1199)))))/( 0.5*(0.25+0.75*exp(-(1.05995+0.5*(\x-2.1199))))*(0.25+0.75*exp(-(1.05995+0.5*(\x-2.1199)))) + 0.5*(0.25+0.75*exp(-2*(1.05995+0.5*(\x-2.1199))))*(0.25+0.75*exp(-2*(1.05995+0.5*(\x-2.1199))))) +
                (\x>2.249062) * (0.5*0.75*exp(-1.124531 )*(0.25+0.75*exp(-1.124531))*(0.25+0.75*exp(-(\x-2.249062))) + 0.5*0.75*exp(-2*1.124531)*(0.25+0.75*exp(-2*1.124531))*(0.25+0.75*exp(-2*(\x-2.249062))))/( 0.5*(0.25+0.75*exp(-1.124531))*(0.25+0.75*exp(-1.124531))*(0.25+0.75*exp(-(\x-2.249062))) + 0.5*(0.25+0.75*exp(-2*1.124531))*(0.25+0.75*exp(-2*1.124531))*(0.25+0.75*exp(-2*(\x-2.249062))));
  }, declare function={
    func3(\x)= (\x<=1.05995) * (0.5*(0.25+0.75*exp(-2*\x)))/(0.5*(0.25+0.75*exp(-\x)) + 0.5*(0.25+0.75*exp(-2*\x))) + and(\x>1.05995, \x<=2.1199)  * (0.5*(0.25+0.75*exp(-2*1.05995))*(0.25+0.75*exp(-2*(\x-1.05995))))/(0.5*(0.25+0.75*exp(-(\x-1.05995)))*(0.25+0.75*exp(-1.05995)) + 0.5*(0.25+0.75*exp(-2*(\x-1.05995)))*(0.25+0.75*exp(-2*1.05995)))   +
     and(\x>2.1199,  \x<=2.249062) * ( 0.5*(0.25+0.75*exp(-2*(1.05995+0.5*(\x-2.1199))))*(0.25+0.75*exp(-2*(1.05995+0.5*(\x-2.1199)))))/( 0.5*(0.25+0.75*exp(-(1.05995+0.5*(\x-2.1199))))*(0.25+0.75*exp(-(1.05995+0.5*(\x-2.1199)))) + 0.5*(0.25+0.75*exp(-2*(1.05995+0.5*(\x-2.1199))))*(0.25+0.75*exp(-2*(1.05995+0.5*(\x-2.1199))))) +
                (\x>2.249062) * (0.5*(0.25+0.75*exp(-2*1.124531))*(0.25+0.75*exp(-2*1.124531))*(0.25+0.75*exp(-2*(\x-2.249062))))/( 0.5*(0.25+0.75*exp(-1.124531))*(0.25+0.75*exp(-1.124531))*(0.25+0.75*exp(-(\x-2.249062))) + 0.5*(0.25+0.75*exp(-2*1.124531))*(0.25+0.75*exp(-2*1.124531))*(0.25+0.75*exp(-2*(\x-2.249062))));
  }
    ]
\begin{axis}[ymin=0.18,ymax=0.24,ytick={0.18,0.2,0.22,0.24},
    ymajorgrids=true, xmin=0, xmax=2.5,
    grid style=dashed, xlabel=\(t\)]
\addplot[color=red,thick, domain=0:2.5,samples=30]{func1(x)};
\addlegendentry{\(\mathbb{P}[\omega_1 = 1| t]\)}
\end{axis}
\end{tikzpicture}
    \caption{The belief over whether approach 1 is valid at state $(K_1, K_2)$ as $K_2$ varies. Example plotted for $\lambda_H = 1$, $\lambda_E = 2$, $\delta_0 = 0.5$, $\nu_0 =0.5$, and $K_1$ fixed to 1. }
    \label{fig:belief_illustration}
\end{figure}
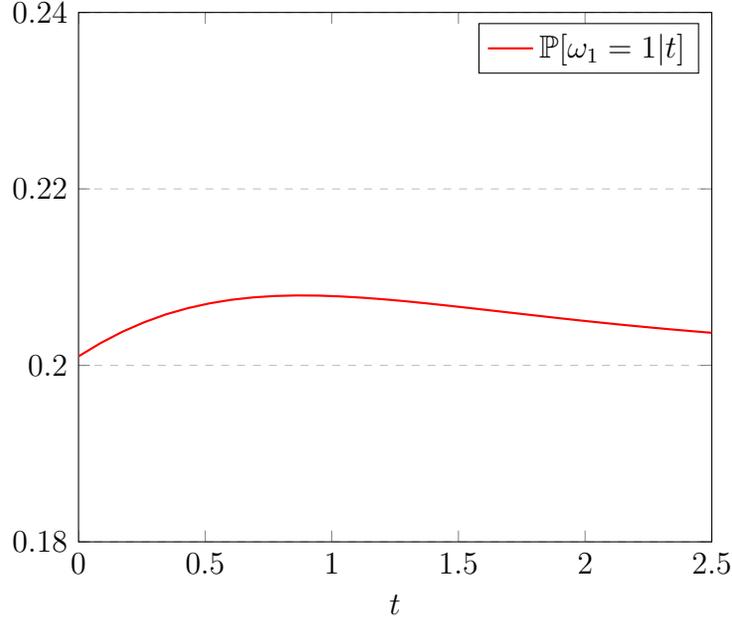

To start building intuition for the coming characterization result, I first highlight a form of the first-order condition derived from \eqref{defn:taustar} that will play a role in the characterization of the optimal policy. Denote the posterior belief over validity conditional on $\theta$ as $\nu_\theta(K) = \nu_0 e^{-\lambda_\theta K}/S_\theta(K)$. Consider a particular first-order effect: 
\begin{equation} \phi_\theta(K) := \lambda_\theta \nu_\theta(K)  - ( r + \lambda_\theta \nu_\theta(K))  \left[ -c + \frac{\nu_0\lambda_\theta}{\lambda_\theta + r}(1 - e^{-(r+\lambda_\theta) K}) \right] - e^{-rK}S_\theta(K)\lambda_\theta \nu_\theta(K). 
\label{eqn:phi}
\end{equation}
The first term is the instantaneous hazard; the marginal value of working on an approach with $K$ effort for an instant longer. The second term is a delay cost; the present-value effect of delaying exploring a new approach. The third term adjusts for continued survival. Intuitively, this marginal condition will need to be zero, averaged across the belief over $\theta$, when the optimal policy moves on to new approaches. Mathematically, $\phi_\theta(K)$ is the Gateaux derivative of the agent's value along a particular policy perturbation corresponding to an interchange.\footnote{See Figure \ref{fig:int_appendix} for a depiction of the interchanges that generate this Gateaux derivative.}
Note that $\phi_\theta(K) = 0$ is exactly a rearrangement of the first-order condition in the benchmark \eqref{defn:taustar}; therefore, the static threshold solution to \eqref{defn:taustar} if the difficulty were known, $K^*_\theta$, is a root of $\phi_\theta$.

I now present the characterization result. Let $B(s) = \arg \min_{i \in [N]} K_i$ be the approaches with least total effort when $N$ approaches have been discovered and the state is $s = (K_1, K_2, ... K_N)$.

\begin{theorem}\label{thm:learning}
    Suppose $\lambda_H > 0$. Define the sequence of thresholds $\{ K^*_n \}_{n\in \mathbb{N}_+}$ implicitly by the smallest solution to the equation
    \begin{equation}\label{eqn:def_tn}
        (1-\delta_0) S_E(K^*_n) ^n\phi_E(K^*_n)  + \delta_0 S_H(K^*_n)^n \phi_H(K^*_n) = 0.
    \end{equation}
    The smallest solution $K^*_n$ to \eqref{eqn:def_tn} is finite for all $n \in \mathbb{N}_+$. The sequence $K^*_n$ is increasing in $n$, and each threshold $K^*_n \in (K^*_E, K^*_H)$.\footnote{$K^*_H$ is understood to be $\infty$ if the agent would not attempt the problem if it were known to be hard.}  On path, the optimal research policy generates a new approach if $\min_{i \in [N]} K_i \ge K^*_N$, else exerts effort on all approaches in $B$ with effort $1/|B|$.
\end{theorem}

Observe that \eqref{eqn:def_tn} provides an implicit characterization of $K^*_n$ for each $n$; notably, the characterization is \textit{not} recursive. This is because the optimal path equalizes effort among approaches, and since $K^*_n$ increases in $n$, the state where the $n+1$st approach is brainstormed is the $n$-vector $(K^*_n, K^*_n, ... K^*_n)$, which has no more dependence on the prior thresholds. 

Theorem \ref{thm:learning} focuses on the case where $\lambda_H > 0$ because the solution $K^*_n$ exists for all $n$; in the case where $\lambda_H = 0$, the result is nearly identical, except that \eqref{eqn:def_tn} no longer admits a solution for all $n$.
\begin{customthm}{1b}\label{thm:learning_impossible}
    If $\lambda_H =0$, there exists a maximum $\bar{N}$ such that the smallest $K^*_n$ solving \eqref{eqn:def_tn} exists for all $n < \bar{N}$. The finite sequence $\{ K^*_n \}_{n \in [\bar{N}-1]}$ is increasing, and the optimal research path generates a new approach if $\min_{i \in [N]} K_i \ge K^*_N$, else exerts effort on all approaches in $B$ with effort $1/|B|$. The agent generates at most $\bar{N}$ approaches (possibly fewer if success arrives sooner), and since working is free, once all $\bar{N}$ approaches are generated the policy works on all approaches in $B$ with effort $1/|B|$. 
\end{customthm}
Intuitively, the distinction between Theorems \ref{thm:learning} and \ref{thm:learning_impossible} is that the agent eventually stops brainstorming new approaches when $\lambda_H = 0$ but does not when $\lambda_H > 0$. I will discuss the intuition for this persistence in the next section, and first provide a sketch of the proofs of these results.

While the formal proof of Theorems \ref{thm:learning} and \ref{thm:learning_impossible} is left to the appendix, I provide the key intuitions that lead to the result. The proofs rely on two key steps; first, I show that if the strategy is not taking the ``brainstorm'' action, the optimal strategy must be pursuing the ``best'' approaches, or the approaches with the lowest history of effort:
\begin{lemma}\label{lem:lowest_history_of_effort}
Suppose that in state $s = (K_1, K_2, ... K_N)$, an optimal strategy does not generate approaches. Then if $\sigma^*(s)_i > 0$, $i$ must be in $B(s)$.
\end{lemma}
Intuitively, if a strategy is exerting effort on some approach, it must be exerting effort on the approaches which are most likely to be good, regardless of what the realization of the aggregate state $\theta$ is. Since historical effort is a negative signal about approach quality, the approaches with least historical effort are most likely to be good irrespective of the agent's belief over $\theta$.

Given a set of times when approaches are brainstormed, Lemma \ref{lem:lowest_history_of_effort} determines a unique on-path strategy, and so it is sufficient to optimize over the sequence of  brainstorming times. Recall from the definition in Theorem \ref{thm:learning} that the effort on approach $i$ when approach $i+1$ is brainstormed is given by $K^*_i$. The following lemma establishes that the optimal strategy must work longer on approaches the more approaches have been brainstormed:
\begin{lemma}\label{lem:increasing_thresholds}
    The effort threshold sequence $\{ K^*_n \}_{n \in \mathbb{N}}$ is weakly increasing.  
\end{lemma}

Given Lemmas \ref{lem:lowest_history_of_effort} and \ref{lem:increasing_thresholds}, the remaining part of the problem is the optimization over the sequence $\{ K_n \}$. The thresholds are pinned down by the optimality of the \textit{calendar} time that the $(n+1)$st approach is brainstormed, when the rest of the continuation policy adjusts correspondingly. So the condition for $K^*_n$ turns into a one-dimensional optimization with respect to the Gateaux derivative given by \eqref{eqn:phi}, and the first-order condition weights the derivative across states according to the belief. Figure \ref{fig:int_appendix} intuitively describes the perturbation associated with this first-order condition. 
I will defer the details of the proof of Theorem \ref{thm:learning} to the appendix, and discuss the implications and insights of the theorem.

    \begin{figure}
     \centering
     \begin{subfigure}[b]{0.3\textwidth}
         \centering
         \begin{tikzpicture}
             \draw[->] (0,0) -- (5,0);
             \filldraw (1,0) circle[radius=2pt] node[anchor=north]{$t_{n+1}$};
             \filldraw (4,0) circle[radius=2pt] node[anchor=north]{$t_{n+1} + K^*_{n}$};
             \fill [orange,opacity=0.5] (0,0) rectangle (1,0.33333);
             \fill [green,opacity=0.5] (0,0.33333) rectangle (1, 0.66667);
             \fill [red,opacity=0.5] (0,0.66667) rectangle (1, 1);
             \node at (0.5, 1.4) {\vdots};
             \fill [olive,opacity=0.5] (0,1.66667) rectangle (1, 2);
             \fill [blue,opacity=0.5] (1,0) rectangle (4,2);
             \fill [blue,opacity=0.5] (4,0) rectangle (5,0.285714);
             \fill [orange,opacity=0.5] (4,0.285714) rectangle (5,0.571428);
             \fill [green,opacity=0.5] (4,0.571428) rectangle (5, 0.857142);
             \fill [red,opacity=0.5] (4,0.857142) rectangle (5, 1.142857);
             \node at (4.5, 1.55) {\vdots};
             \fill [olive,opacity=0.5] (4,1.714285) rectangle (5, 2);
             \draw (1,0) -- (1,2);
             \draw (4,0) -- (4,2);
         \end{tikzpicture}
         \caption{Original action path $k$}
         \label{fig:int_appendix_0}
     \end{subfigure}
     \hfill
     \begin{subfigure}[b]{0.3\textwidth}
         \centering         
         \begin{tikzpicture}
             \draw[->] (0,0) -- (5,0);
             \filldraw (1,0) circle[radius=2pt] node[anchor=north]{$t_{n+1}$};
             \filldraw (4,0) circle[radius=2pt] node[anchor=north]{$t_{n+1} + K^*_{n}$};
             \fill [orange,opacity=0.5] (0,0) rectangle (0.75,0.33333);
             \fill [green,opacity=0.5] (0,0.33333) rectangle (0.75, 0.66667);
             \fill [red,opacity=0.5] (0,0.66667) rectangle (0.75, 1);
             \node at (0.375, 1.4) {\vdots};
             \fill [olive,opacity=0.5] (0,1.66667) rectangle (0.75, 2);
             \fill [blue,opacity=0.5] (0.75,0) rectangle (3.708333,2);
             \fill [blue,opacity=0.5] (3.708333,0) rectangle (5,0.285714);
             \fill [orange,opacity=0.5] (3.708333,0.285714) rectangle (5,0.571428);
             \fill [green,opacity=0.5] (3.708333,0.571428) rectangle (5, 0.857142);
             \fill [red,opacity=0.5] (3.708333,0.857142) rectangle (5, 1.142857);
             \node at (4.25, 1.55) {\vdots};
             \fill [olive,opacity=0.5] (3.708333,1.714285) rectangle (5, 2);
             \draw[thick] (0.75,0) rectangle (1,2);
             \draw[thick] (3.708333,0.285714) rectangle (4,2);
             \draw[<->, very thick] (1.1,1) -- (3.4,1.5);
             
         \end{tikzpicture}
         \caption{Interchange 1}
         \label{fig:int_appendix_1}
     \end{subfigure}
     \hfill
     \begin{subfigure}[b]{0.3\textwidth}
         \centering         
         \begin{tikzpicture}
             \draw[->] (0,0) -- (5,0);
             \filldraw (1,0) circle[radius=2pt] node[anchor=north]{$t_{n+1}$};
             \filldraw (4,0) circle[radius=2pt] node[anchor=north]{$t_{n+1} + K^*_{n}$};
             \fill [orange,opacity=0.5] (0,0) rectangle (1.25,0.33333);
             \fill [green,opacity=0.5] (0,0.33333) rectangle (1.25, 0.66667);
             \fill [red,opacity=0.5] (0,0.66667) rectangle (1.25, 1);
             \node at (0.5, 1.4) {\vdots};
             \fill [olive,opacity=0.5] (0,1.66667) rectangle (1.25, 2);
             \fill [blue,opacity=0.5] (1.25,0) rectangle (4.291667,2);
             \fill [blue,opacity=0.5] (4.291667,0) rectangle (5,0.285714);
             \fill [orange,opacity=0.5] (4.291667,0.285714) rectangle (5,0.571428);
             \fill [green,opacity=0.5] (4.291667,0.571428) rectangle (5, 0.857142);
             \fill [red,opacity=0.5] (4.291667,0.857142) rectangle (5, 1.142857);
             \node at (4.61667, 1.55) {\vdots};
             \fill [olive,opacity=0.5] (4.291667,1.714285) rectangle (5, 2);
             \draw[thick] (1,0) rectangle (1.25,2);
             \draw[thick] (4,0.285714) rectangle (4.291667,2);
             \draw[<->, very thick] (1.35,1) -- (3.9,1.5);
         \end{tikzpicture}
         \caption{Interchange 2}
         \label{fig:int_appendix_2}
     \end{subfigure}
        \caption{The original action path, and the interchanges which generate the first-order condition that determines $K^*_n$. Each color denotes effort on a distinct approach. Blue denotes effort on the $(n+1)$st approach, and $t_{n+1}$ denotes the time that the $(n+1)$st approach is brainstormed. }
        \label{fig:int_appendix}
\end{figure}
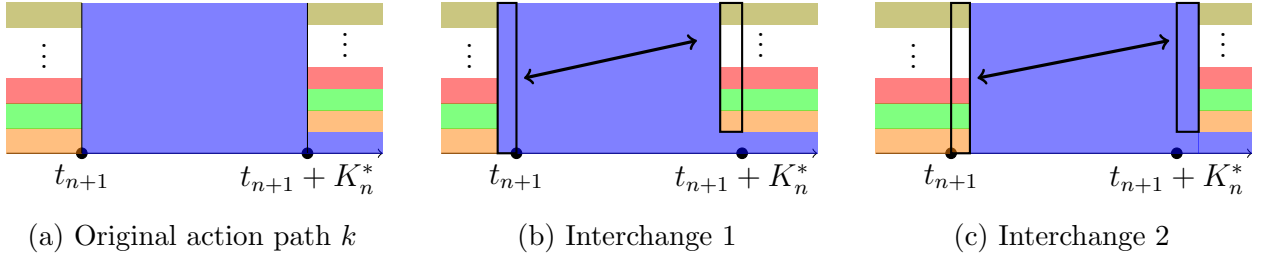

\subsubsection{Discussion and Intuition}

\paragraph{Task Juggling and Recall} The optimal research pattern proceeds as follows: the agent works on the first approach up until effort $K^*_1$, after which the agent brainstorms the second approach. The agent works on the second approach only until $K^*_1$ has been exerted on approach 2, after which the agent \textit{jointly} works on both the first and second approach until both approaches reach effort $K^*_2$, before brainstorming the next approach. In contrast to the benchmark result in Proposition \ref{prop:benchmark}, task juggling does occur on the optimal path, and the agent optimally revisits previous approaches.

\begin{corollary}\label{corr:juggling_recall}
    The optimal policy alternates between exclusively exerting effort on a new approach, and splitting effort between all approaches (including previous approaches).
\end{corollary}

To understand the difference in optimal research behavior, consider the beliefs of the agent over $\lambda_\theta$. Suppose the agent first stops working on approach $1$ at time $K^*_1$; let the agent have belief $q_1$ that the problem is easy, $\lambda = \lambda_E$, and belief $p_1$ that approach 1 is good. Later, after the agent has exerted effort $K^*_1$ on the \textit{second} approach without success, the agent's belief over $\theta$ has fallen to some $q < q_1$. In consequence, the agent's perceived continuation value of a new approach is now \textit{lower} than it was at $K^*_1$. Further, over the interval of time where the agent worked on the second approach, the agent believes it is more likely that the problem is difficult, and hence their belief that approach 1 is valid increases to some $p > p_1$; that is, the agent attributes the lack of success on approach 1 less to the approach being flawed and more to the problem being difficult.\footnote{Note that even though the agent's perceived probability that approach 1 is good ($\mathbb{E}[\omega_1]$) increases as the agent exerts effort on approach 2 without success, the expected arrival rate of a breakthrough on approach 1 ($\mathbb{E}[\omega_1 \lambda_\theta]$) always decreases in effort on approach 2.} As a result, the agent prefers revisiting approach $1$ relative to brainstorming, and thus spends more time working on both 1 and 2 for some time before continuing. Figure \ref{fig:learning_beliefs} plots the relevant beliefs on this time interval.

\paragraph{Perseverance} When hard endeavors are never impossible $(\lambda_H > 0)$, one consequence is that the sequence given by \eqref{eqn:def_tn} is infinite; in particular, the agent never stops brainstorming approaches. This feature arises even if $\lambda_H$ is so low that the agent would not be willing to start brainstorming approaches if it were known that $\theta = H$. 
To see this, suppose that the agent has brainstormed $N$ approaches. Then as the agent continues to work on the $N$ approaches without breakthrough, the agent becomes more convinced that all $N$ approaches must be flawed. However, if the approaches were all flawed, then the agent's belief in the problem difficulty drifts back towards the prior (because effort on a flawed approach is completely uninformative). More precisely, the belief that an agent has over the aggregate state is non-monotone in the effort on their approaches; if $\delta(K,N)$ is the belief that $\theta = H$ when there has been $K$ effort on $N$ approaches, then 
\[ \delta(K,N) = \frac{\delta_0 S_H(K)^N}{\delta_0 S_H(K)^N + (1-\delta_0) S_E(K)^N}. \]
Note that the ratio of $S_E(K) / S_H(K) \to 1$ as $K \to \infty$; so $\delta(K,N) \to \delta_0$ as $K \to \infty$. 
Since the agent was willing to start working on an approach at the prior belief, it follows that the agent will eventually become sufficiently optimistic about $\theta$ again and brainstorm a new approach.

\begin{figure}
    \centering
    \begin{tikzpicture}[
    declare function={
    func1(\x)= (\x<=1.05995) * (0.5*0.75*exp(-\x) + 0.5*0.75*exp(-2*\x))/(0.25+0.5*0.75*exp(-\x) + 0.5*0.75*exp(-2*\x))   +
     and(\x>1.05995, \x<=2.1199) * (0.5*0.75*exp(-1.05995)*(0.25+0.75*exp(-(\x-1.05995))) + 0.5*0.75*exp(-2*1.05995)*(0.25+0.75*exp(-2*(\x-1.05995))))/( 0.5*(0.25+0.75*exp(-1.05995))*(0.25+0.75*exp(-(\x-1.05995))) + 0.5*(0.25+0.75*exp(-2*(1.05995)))*(0.25+0.75*exp(-2*(\x-1.05995))))    +
     and(\x>2.1199,  \x<=2.249062) * (0.5*0.75*exp(-(1.05995+0.5*(\x-2.1199)) )*(0.25+0.75*exp(-(1.05995+0.5*(\x-2.1199)))) + 0.5*0.75*exp(-2*(1.05995+0.5*(\x-2.1199)))*(0.25+0.75*exp(-2*(1.05995+0.5*(\x-2.1199)))))/( 0.5*(0.25+0.75*exp(-(1.05995+0.5*(\x-2.1199))))*(0.25+0.75*exp(-(1.05995+0.5*(\x-2.1199)))) + 0.5*(0.25+0.75*exp(-2*(1.05995+0.5*(\x-2.1199))))*(0.25+0.75*exp(-2*(1.05995+0.5*(\x-2.1199))))) +
                (\x>2.249062) * (0.5*0.75*exp(-1.124531 )*(0.25+0.75*exp(-1.124531))*(0.25+0.75*exp(-(\x-2.249062))) + 0.5*0.75*exp(-2*1.124531)*(0.25+0.75*exp(-2*1.124531))*(0.25+0.75*exp(-2*(\x-2.249062))))/( 0.5*(0.25+0.75*exp(-1.124531))*(0.25+0.75*exp(-1.124531))*(0.25+0.75*exp(-(\x-2.249062))) + 0.5*(0.25+0.75*exp(-2*1.124531))*(0.25+0.75*exp(-2*1.124531))*(0.25+0.75*exp(-2*(\x-2.249062))));
  }, declare function={
    func2(\x)= (\x<=1.05995) * 0.75 + and(\x>1.05995, \x<=2.1199)  * (0.5*(0.25+0.75*exp(-1.05995))*0.75*exp(-(\x-1.05995)) + 0.5*(0.25+0.75*exp(-2*1.05995))*0.75*exp(-2*(\x-1.05995)))/(0.5*(0.25+0.75*exp(-(\x-1.05995)))*(0.25+0.75*exp(-1.05995)) + 0.5*(0.25+0.75*exp(-2*(\x-1.05995)))*(0.25+0.75*exp(-2*1.05995)))   +
     and(\x>2.1199,  \x<=2.249062) * (0.5*0.75*exp(-(1.05995+0.5*(\x-2.1199)) )*(0.25+0.75*exp(-(1.05995+0.5*(\x-2.1199)))) + 0.5*0.75*exp(-2*(1.05995+0.5*(\x-2.1199)))*(0.25+0.75*exp(-2*(1.05995+0.5*(\x-2.1199)))))/( 0.5*(0.25+0.75*exp(-(1.05995+0.5*(\x-2.1199))))*(0.25+0.75*exp(-(1.05995+0.5*(\x-2.1199)))) + 0.5*(0.25+0.75*exp(-2*(1.05995+0.5*(\x-2.1199))))*(0.25+0.75*exp(-2*(1.05995+0.5*(\x-2.1199))))) +
                (\x>2.249062) * (0.5*0.75*exp(-1.124531 )*(0.25+0.75*exp(-1.124531))*(0.25+0.75*exp(-(\x-2.249062))) + 0.5*0.75*exp(-2*1.124531)*(0.25+0.75*exp(-2*1.124531))*(0.25+0.75*exp(-2*(\x-2.249062))))/( 0.5*(0.25+0.75*exp(-1.124531))*(0.25+0.75*exp(-1.124531))*(0.25+0.75*exp(-(\x-2.249062))) + 0.5*(0.25+0.75*exp(-2*1.124531))*(0.25+0.75*exp(-2*1.124531))*(0.25+0.75*exp(-2*(\x-2.249062))));
  }, declare function={
    func3(\x)= (\x<=1.05995) * (0.5*(0.25+0.75*exp(-2*\x)))/(0.5*(0.25+0.75*exp(-\x)) + 0.5*(0.25+0.75*exp(-2*\x))) + and(\x>1.05995, \x<=2.1199)  * (0.5*(0.25+0.75*exp(-2*1.05995))*(0.25+0.75*exp(-2*(\x-1.05995))))/(0.5*(0.25+0.75*exp(-(\x-1.05995)))*(0.25+0.75*exp(-1.05995)) + 0.5*(0.25+0.75*exp(-2*(\x-1.05995)))*(0.25+0.75*exp(-2*1.05995)))   +
     and(\x>2.1199,  \x<=2.249062) * ( 0.5*(0.25+0.75*exp(-2*(1.05995+0.5*(\x-2.1199))))*(0.25+0.75*exp(-2*(1.05995+0.5*(\x-2.1199)))))/( 0.5*(0.25+0.75*exp(-(1.05995+0.5*(\x-2.1199))))*(0.25+0.75*exp(-(1.05995+0.5*(\x-2.1199)))) + 0.5*(0.25+0.75*exp(-2*(1.05995+0.5*(\x-2.1199))))*(0.25+0.75*exp(-2*(1.05995+0.5*(\x-2.1199))))) +
                (\x>2.249062) * (0.5*(0.25+0.75*exp(-2*1.124531))*(0.25+0.75*exp(-2*1.124531))*(0.25+0.75*exp(-2*(\x-2.249062))))/( 0.5*(0.25+0.75*exp(-1.124531))*(0.25+0.75*exp(-1.124531))*(0.25+0.75*exp(-(\x-2.249062))) + 0.5*(0.25+0.75*exp(-2*1.124531))*(0.25+0.75*exp(-2*1.124531))*(0.25+0.75*exp(-2*(\x-2.249062))));
  }
    ]
\begin{axis}[ymin=0,ymax=1,ytick={0.2,0.4,0.6,0.8},
    ymajorgrids=true, xmin=0, xmax=2.5,
    grid style=dashed, xlabel=\(t\)]
\addplot[color=red, thick, domain=0:2.5,samples=30]{func1(x)};
\addlegendentry{\(\mathbb{P}[\omega_1 = 1| t]\)}
\addplot[color=blue, thick, domain=0:2.5,samples=30]{func2(x)};
\addlegendentry{\(\mathbb{P}[\omega_2 =1 | t]\)}
\addplot[color=olive, thick, domain=0:2.5,samples=30]{func3(x)};
\addlegendentry{\(\mathbb{P}[\theta= E | t]\)}
\pgfplotsinvokeforeach{1.05995,2.1199,2.249062}{
  \draw[dashed] ({rel axis cs: 0,0} -| {axis cs: #1, 0}) -- ({rel axis cs: 0,1} -| {axis cs: #1, 0});}
\end{axis}
\end{tikzpicture}
    \caption{The beliefs over $\omega_1$, $\theta$, and $\omega_2$ on the optimal path, for an example where $\nu_0 = 0.75$, $\delta_0 = 0.5$, $\lambda_E = 2$, $\lambda_H = 1$, $r=1$, and $c = 0.1$. Up until $t = 1.06$, the optimal policy only works on approach 1; for $t \in [1.06,2.12)$, the policy works on approach 2; for $t \in [2.12, 2.25)$, the policy splits effort on 1 and 2, and after $t \ge 2.25$ the policy works on approach 3. Note that in the intermediate region of time from $[1.06, 2.12)$, the policy works only on approach 2, but the belief over whether approach 1 is valid drifts upwards.}
    \label{fig:learning_beliefs}
\end{figure}
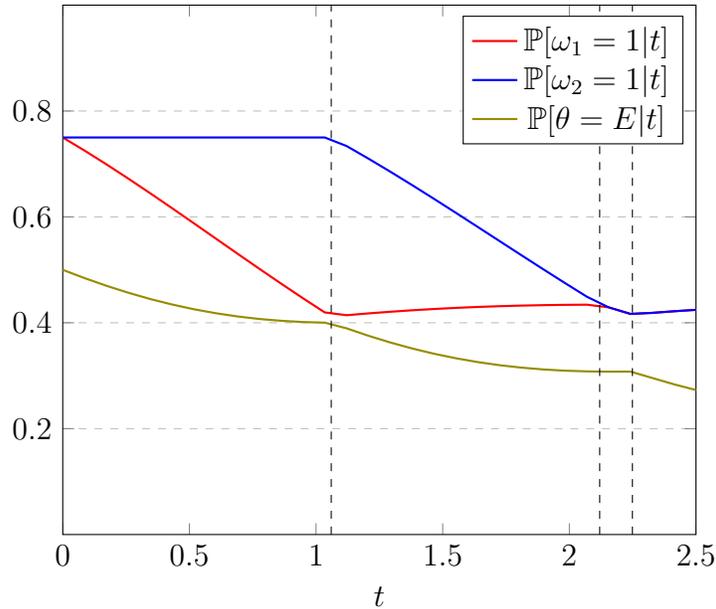

\begin{figure}
     \centering
     \begin{subfigure}[b]{0.48\textwidth}
         \centering
         \includegraphics[width=\textwidth]{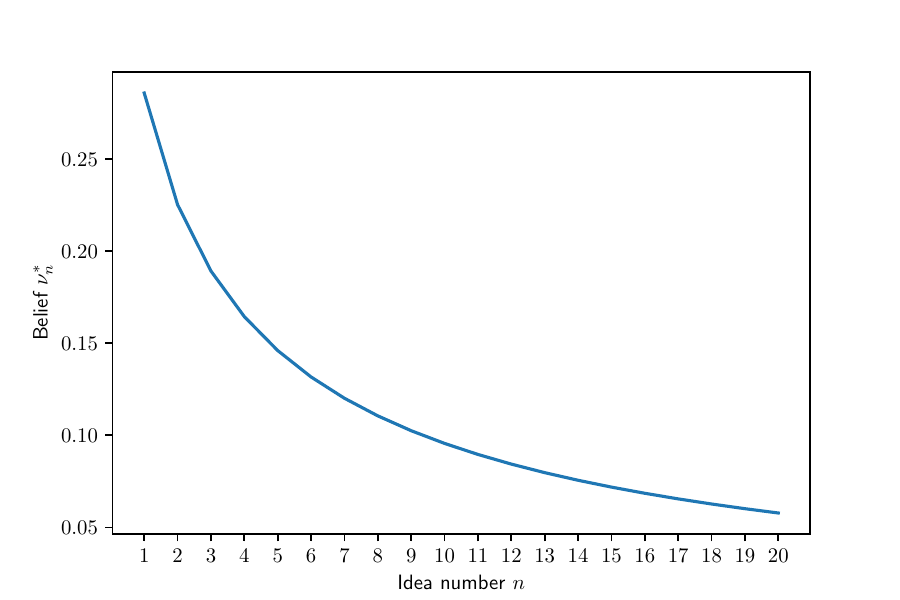}
         \caption{Belief over quality of old approaches when brainstorming new approaches.}
         \label{fig:pstar}
     \end{subfigure}
     \hfill
     \begin{subfigure}[b]{0.48\textwidth}
         \centering
         \includegraphics[width=\textwidth]{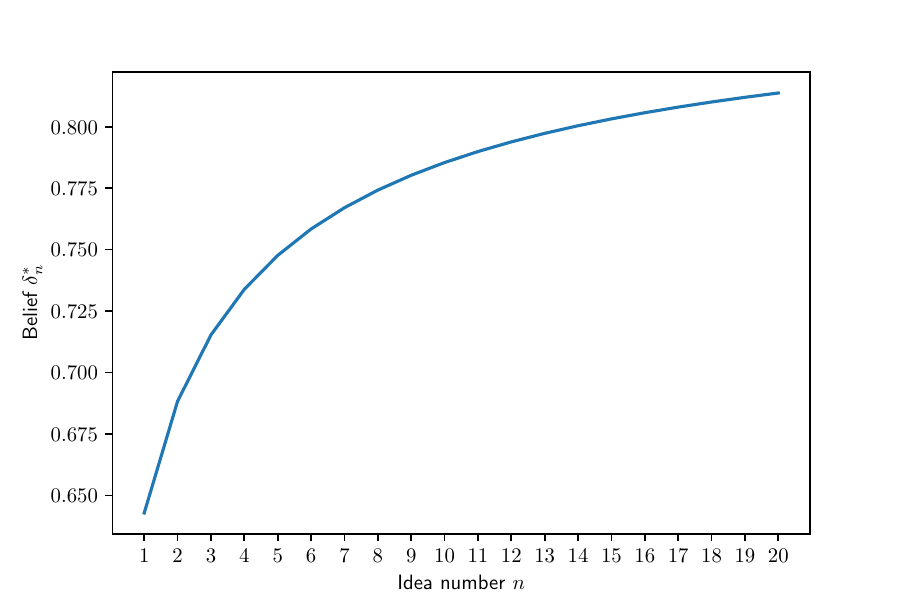}
         \caption{Belief over the aggregate state $\theta$ when brainstorming new approaches.}
         \label{fig:qstar}
     \end{subfigure}
        \caption{Beliefs at the times when new approaches are brainstormed. Example simulated for $p = 0.75$, $r=1$, $\delta_0 = 0.5$, $c = 0.2$, $\lambda_E = 2$ and $\lambda_H = 0.25$. Note that the belief thresholds decline as more approaches have been worked on, so beliefs alone are insufficient to characterize the timing decisions on when to brainstorm.}
        \label{fig:beliefs_over_time}
\end{figure}
\paragraph{Beliefs Over Time} Given Lemma \ref{lem:lowest_history_of_effort}, it is tempting to think that an index result might hold in this environment with the appropriate notion on beliefs. However, a standard index result does not apply because the value of a brainstorming a new approach evolves according to the full state of effort exerted on the other approaches, not just the belief over the aggregate state $\theta$. Let $\nu^*_n$ denote the belief that any individual approach is good when $n$ approaches have effort $K^*_n$, and let $\delta^*_n$ denote the belief that $\theta = H$ at the same state. Figure \ref{fig:beliefs_over_time} plots the evolution of beliefs when new approaches are brainstormed. Note that $\nu^*_n$ declines in $n$, and $\delta^*_n$ increases in $n$, so the value of a new approach evolves with the history of effort on the other approaches.
Thus, unlike much of the literature in strategic experimentation, the optimal behavior cannot be characterized easily in belief space.

\section{Limit Model}\label{sec:continuum}

The discrete-arm model highlights the fundamental tradeoff between exploring new approaches and exploiting existing ones.
Unfortunately, the discrete structure of the baseline model makes it difficult to analyze richer strategic environments. In particular, the distribution of success arrival exhibits kinks at the moments when new approaches are introduced, which complicates the analysis of dynamic contracts and may preclude the existence of pure-strategy equilibria in multi-agent settings. 

Since one of the objectives of this paper is to develop a model of problem-solving amenable to studying economically interesting environments, this section uses insights gained from Theorem \ref{thm:learning} in order to develop a workhorse model. Specifically, in order to overcome these tractability issues, I develop a continuum-arm formulation in which the agent chooses, at each instant, the breadth (number of approaches started) and the depth (average effort allocated to each approach). This limit model preserves the qualitative insights of the baseline model while providing a continuous state space that is amenable to strategic analysis. 
After presenting the continuum model, the final section uses this formulation to examine how an investor designs dynamic incentives when experimentation is hidden.

\paragraph{Continuum Model} Consider an agent searching for a breakthrough. At some time $t$, instead of maintaining a discrete set of approaches, the agent has a \textit{measure} $b_t \in \mathbb{R}_+$ of approaches; I will call this the ``breadth'' of their search for a breakthrough. The agent also maintains a ``depth'' of their search, $d_t \in \mathbb{R}_+$; this corresponds to how much ``effort'' they have exerted per approach they have. The agent experiences a flow cost $c$ of expanding the \textit{breadth} of their search (i.e., new approaches require overhead cost). The agent only has a single unit of effort per time, so must decide between how broad versus how deep to search:
\[ \frac{\dd{}}{\dd{t}}[b_t d_t] = b_t \dot{d}_t + \dot{b}_t d_t = 1 .\footnote{The product rule expansion of the area law admits a singularity when $b, d$ are zero. However, in this case, the area of the rectangle is still well-defined, and the requirement that the area expands are rate $1$ determines the time path along any trajectory with $(b_0, d_0) = (0,0)$.} \]

The agent gets a payoff (normalized to 1) when breakthrough arrives. Breakthrough is stochastic, and its distribution depends on the state of search; that is, given a state trajectory $(b_t, d_t)$, the distribution function of breakthrough is defined by 
\[ \hat{F}(b_t, d_t) := 1 - (1-\delta_0) \exp\left(- \nu_0 b_t (1 - e^{-\lambda_E d_t}) \right) - \delta_0 \exp\left(- \nu_0 b_t (1 - e^{-\lambda_H d_t}) \right).  \]
The form of the breakthrough distribution is chosen as an appropriate limit of the discrete-arm model, which I will state more precisely in the coming Theorem \ref{thm:microfound}. 

To simplify from two paths to one, note that the effort constraint implies that $d_t = t / b_t$, so it suffices to reduce the state to the one-dimensional breadth (``number of ideas'').  That is, let $x$ denote the state variable (breadth), and re-parametrize the distribution function for breakthrough as:\footnote{As a technical formality, \eqref{def:F} is not necessarily well-defined for $F(0,0)$; so set $F(0,0) = F(\cdot,0) = F(0,\cdot) = 0$.}
\begin{equation} F(x,t) = 1 - (1-\delta_0) \exp\left(- \nu_0 x (1 - e^{-\lambda_Et/x}) \right) - \delta_0 \exp\left( - \nu_0 x( 1 - e^{-\lambda_Ht/x}) \right).\label{def:F}
\end{equation}
The parameter $\delta_0$ is relatively straightforward; it is the probability of a hard problem. To understand the double exponential form, one can interpret the $\nu_0$ parameter as a rate at which expanding breadth yields a good idea; so if the agent has searched $x$ wide, the number of good ideas is distributed as a Poisson distribution with parameter $\nu_0 x$. Fixing a number of good ideas $N$, the probability of no success is then the probability that no good approach has succeeded (after having explored each idea for $t/x$), which gives rise to the inner exponential:
$\mathbb{E}[(e^{-\lambda_\theta t/x})^N]$. The outer exponential arises from taking the expectation over $N$, which is distributed as $N \sim $ Pois$(\nu_0x)$. Intuitively, this helps make sense of the $t \to \infty$ asymptotic; when $t \to \infty$ fixing $x$, the CDF becomes 
\[ F(x,\infty) = 1 - \exp(-\nu_0 x), \]
which is precisely the Poisson probability that $N > 0$ for $N \sim $ Pois$(\nu_0x)$, or the probability that by exploring breadth $x$ that a good idea has been discovered.
As shorthand, I will use $F_x, F_t$ to denote the partial derivatives of $F$ in $x, t$ respectively, and similarly use $F_{xx}, F_{xt}, F_{tt}$ to denote the partial second derivatives.

The agent's problem is to choose a search strategy over breadth and depth, in order to maximize their discounted payoff net of costs:\footnote{The expression $\dd F / \dd t$ is a mild abuse of notation but kept for sake of cleanliness, since the objective technically means the probability measure over success times could have a mass point at $\infty$, in which case the objective should account for the lifetime accumulated costs. This is taken into account in the analysis when integrating the objective by parts.}
\begin{gather}
    \max_{x} \int_0^\infty  \left[ e^{-rt} - \left(\int_0^t e^{-rs} c \dot{x}  \dd{s}\right) \right]\frac{\dd F(x,t)}{\dd{t}}\dd{t}, \label{eqn:continuum_single} \\
    \textnormal{ subject to: } \dot{x} \in [0, x/t] \textnormal{ a.e.}, \quad x(0) = 0, \quad \dot{x}(0) < \infty. \notag
\end{gather}
Say a strategy $x$ is \textit{admissible} if it is continuous, piecewise differentiable, and satisfies the above constraints where applicable, and say the strategy is \textit{interior} if $\dot{x} \in (0, x/t)$.
As before, assume that costs are sufficiently small: $c < \nu_0$ (otherwise, the agent never starts). 

Note that Theorem \ref{thm:learning} plays a key role in simplifying the formulation of the continuum-arm version; in particular, the observation that all approaches must have equalized effort when new approaches are brainstormed implies that in the continuum, the optimal ``search'' area is a rectangle, and thus the state can be summarized by a two-dimensional sufficient statistic (i.e., breadth and depth of search). 

\paragraph{Interpretation} 

One way to interpret the ``continuum of arms'' is as an aggregation of human capital. In this interpretation, the cost $c$ corresponds to a flow mental effort cost of learning new things, and searching ``deep'' by increasing $d$ corresponds to thinking about applying learned concepts towards solving the problem.

More concretely, the continuum-arm framework naturally maps onto several empirical research domains where investigators face breadth-depth tradeoffs under uncertainty. In drug discovery, pharmaceutical researchers must decide how many chemical scaffolds to screen (breadth) versus how deeply to optimize each lead compound (depth), while not knowing ex ante whether a particular disease target is ``druggable.'' Similarly, a materials science laboratory balances the number of compositional variations to explore (breadth) against the thoroughness of characterization for each material (depth), without knowing whether their target properties are achievable within the chosen material class. In machine learning research, practitioners allocate compute resources between trying many neural architectures (breadth) versus extensive hyperparameter tuning of fewer candidates (depth). 
In these interpretations, the cost of expanding breadth should be interpreted as a creativity cost that the agent bears.

\subsection{Decision Problem Solution} 
First, I characterize the decision problem solution to \eqref{eqn:continuum_single}. 
\begin{proposition}
    There exists a unique optimal trajectory $x^*$ that solves \eqref{eqn:continuum_single}. The optimal solution  $x^*$ is implicitly characterized by 
    \begin{equation}
        \frac{F_x(x^*,t)}{1 - F(x^*,t)} - c = \frac{c}{r}\left( \frac{F_t(x^*,t)}{1-F(x^*,t)}\right),\label{eqn:single_EL}
    \end{equation}
    where $x^*(t)$ solves the above pointwise in $t$. 
    \label{prop:continuum_single}
\end{proposition}
The algebraic details are left to the appendix; the condition \eqref{eqn:single_EL} is a straightforward Euler-Lagrange condition that arises from the maximization problem in \eqref{eqn:continuum_single}. Note that the characterization \eqref{eqn:single_EL} is \textit{not} a partial differential equation; the notation $F_x$, $F_t$ serve as shorthand conveniences, since $F$ is given in \eqref{def:F}. 

To understand the trajectory $x^*$ better, it will be useful to think about the correspondingly induced ``depth'', $t/ x^*(t)$. Note that when depth is constant, $x^*(t)$ will be linear. In fact, as might be expected from the discrete-arm analysis, $x^*(t)$ is linear when $\lambda_E = \lambda_H$: 
\begin{corollary}\label{corr:continuum_properties_1}
    When $\lambda_E = \lambda_H = \lambda$, the optimal $x^*$ for the single agent problem is linear, with $x^*(t) = t / d^*$, where $d^*$ is the unique solution to the following implicit equation:
    \begin{equation}
        r \nu_0 \left(1 - e^{-\lambda d^*} - \lambda d^* e^{-\lambda d^*} \right) - rc - c \nu_0\lambda e^{-\lambda d^*} = 0. \label{eqn:constant_depth}
    \end{equation}
\end{corollary}
The next proposition shows that the same comparative statics in Proposition \ref{prop:baseline_comp} still hold.

\begin{proposition}\label{prop:continuum_single_comp}
    When $\lambda_H = \lambda_E$, the optimal single-agent depth $d^*$, which is solution to \eqref{eqn:constant_depth}, is increasing in $c$ and decreasing in $\lambda$. 
\end{proposition}
The proof is a straightforward application of the implicit function theorem and is left to the appendix. It follows from the same exact intuition as in the baseline model.

As might be expected from Theorem \ref{thm:learning}, when the agent is learning about difficulty over time, $x^*(t)$ will be increasing but concave, and the depth increases over time:
\begin{corollary}\label{corr:continuum_properties_2}
    If $\lambda_E > \lambda_H$ and $\delta_0 \in (0,1)$, then $x^*(t)$ is strictly concave. In the limit as $t \to \infty$, the depth $t / x^*(t) \to d_H$, where $d_H$ solves \eqref{eqn:constant_depth} for $\lambda = \lambda_H$. In the limit $t \to 0$, the depth $t/x^*(t)$ approaches $d_0$, which solves
    \begin{align*}& \delta_0 \left[r \nu_0 \left(1 - e^{-\lambda_H d_0} - \lambda_H d_0 e^{-\lambda_H d_0} \right) - c \nu_0\lambda_H e^{-\lambda_H d_0}\right] \\
    &+ (1-\delta_0)\left[r \nu_0 \left(1 - e^{-\lambda_E d_0} - \lambda_E d_0 e^{-\lambda_E d_0} \right) - c \nu_0\lambda_E e^{-\lambda_E d_0} \right] - rc= 0. \end{align*}
\end{corollary}

The continuum model allows us to capture the breadth/depth tradeoff between exploring more approaches and exploiting existing solution approaches.
Having demonstrated that the solution to this continuum model is qualitatively similar to the results presented for the baseline, I now formally microfound the continuum model and show that it arises as an appropriate limit of a sequence of baseline models.

\subsection{Connection to Baseline Model}
In this section, I show how the continuum-arm model arises as the limit of an appropriate sequence of discrete-arm models analyzed previously. 
The baseline decision problem can be described by the primitives $(r, \nu_0, \delta_0, \lambda_E, \lambda_H, c)$. Consider the sequence of decision problems, $\{ D^n := (r^n, \nu_0^n, \delta^n_0, \lambda_E^n, \lambda_H^n, c^n) \}_{n \ge 1}$, where $r^n$ is fixed to $r^*$, $\delta_0^n$ is fixed to $\delta_0^*$, and the other primitives vary as:
\begin{align*}
    \nu_0^n &= \nu^*_0 / n, & \lambda_\theta^n &= \lambda^*_\theta n \quad \forall \theta \in \{E, H\}, & c^n &= c^*/n.
\end{align*}
Let $D^\infty$ denote the continuum-arm model with primitives $(r^*, \nu_0^*, \delta_0^*, \lambda_E^*, \lambda_H^*, c^*)$. 
I will show that the continuum-arm model arises in the $n \to \infty$ limit of $\{D_n\}_{n\ge 1}$. To interpret the sequence, note that $n$ appears in scaling down $\nu^0_n$ and $c$ and scaling up $\lambda^*_\theta$. Intuitively, as $n$ becomes large, individual approaches are less likely to be good ($\nu^0_n \to 0$) but also become extremely cheap ($c \to 0$), and the decision-maker learns about them faster ($\lambda^n_\theta \to \infty$). An alternative equivalent interpretation of the $\lambda_\theta$ convergence is that the total measure of effort scales as $n$.  

To present the convergence result formally, I first define notions of convergence in strategies and outcomes. Suppose there is a sequence of strategies $\{ \hat{\sigma}^n \}_{n \ge 1}$, $\hat{\sigma}^n$ for $D^n$, such that $\hat{\sigma}^n$ satisfies Lemmas \ref{lem:lowest_history_of_effort} and \ref{lem:increasing_thresholds} for all $n$. Then $\hat{\sigma}^n$ can be identified by an increasing sequence of thresholds $\{K^n_j\}_{j \ge 1}$, which in turn induces the following ``normalized number of arms'' function in time: 
\begin{equation} \hat{N}^n(t) := \frac{1}{n}\left( 1 + \sum_{j=1}^\infty \mathbb{I}[j K^n_j < t] \right), \label{def:nna} \end{equation}
where $\mathbb{I}$ denotes the indicator function. Note that $j K^n_j$ is precisely the time that the $j+1$st arm is brainstormed, so the parenthesized part is exactly the counting function for the number of arms brainstormed by the strategy $\hat{\sigma}^n$.\footnote{Note that $\hat{N}^n(t)$ is well-defined even when the threshold sequence may terminate, as with $\lambda_H = 0$, and so I work with $\hat{N}^n(t)$ to avoid separately stating the main result for $\lambda_H = 0$.}
\begin{definition}
    A sequence of strategies $\{ \hat{\sigma}^n \}_{n \ge 1}$, $\hat{\sigma}^n$ for $D^n$, which satisfies Lemmas \ref{lem:lowest_history_of_effort} and \ref{lem:increasing_thresholds}, \textbf{converges in strategy} to $\hat{x}$ if the sequence of functions $ \{ \hat{N}^n \}_{n\ge 1} $ uniformly converges to $\hat{x}$.
\end{definition}
To define convergence in outcome, let $F^n(\sigma^n, t)$ denote the cumulative distribution function of the breakthrough time under strategy $\sigma^n$ in the discrete-arm model $D^n$, and let $\Pi^n(\sigma^n)$ similarly denote the decision-maker's payoff. Let $\Pi(x)$ denote the value of the objective \eqref{eqn:continuum_single} under strategy $x$. 
\begin{definition}
    A sequence of strategies $\{ \sigma^n \}_{n \ge 1}$ \textbf{converges in outcome} to the strategy $x$ if 
    \begin{align*} 
    \lim_{n \to \infty} F^n( \sigma^n, t) &= F(x(t), t) \quad \textnormal{uniformly in } t, \\
    \lim_{n \to \infty} \Pi^n(\sigma^n) &= \Pi(x).
    \end{align*}
\end{definition}
Finally, I present the convergence result. 

\begin{theorem}
    Consider the sequence of decision problems $\{D^n\}_{n\ge 1}$. 
    \begin{enumerate}
        \item Suppose the sequence of strategies $\{ \hat{\sigma}^n \}_{n \ge 1}$, $\hat{\sigma}^n$ for $D^n$, converges in strategy to $\hat{x}$. Then $\{ \hat{\sigma}^n \}_{n \ge 1}$ converges in outcome to $\hat{x}$.
        \item Fix any feasible strategy $x$ of the continuum-arm problem. Then there exists a sequence of strategies $\{ \sigma^n\}_{n \ge 1}$, $\sigma^n$ for $D^n$ satisfying Lemmas \ref{lem:lowest_history_of_effort} and \ref{lem:increasing_thresholds}, such that $\{\sigma^n\}_{n \ge 1}$ converges in strategy (and hence in outcome) to $x$. 
        \item The sequence of optimal strategies $\{ \sigma^{n}_{\OPT} \}_{n\ge 1}$, $\sigma^n_\OPT$ the optimal solution for $D^n$, converges in strategy to the optimal solution to the continuum problem, $x^*$ (and hence, by point 1, also converges in outcome).
    \end{enumerate}
   \label{thm:microfound}
\end{theorem}
Intuitively, the first point states that convergence in strategy implies convergence in outcomes. The second then implies that any feasible trajectory in the continuum model is indeed the limit in strategy of some sequence of strategies in the discrete-arm models. The proof of Theorem \ref{thm:microfound} involves careful but straightforward limit-taking, and is left to the appendix. 
The last point indicates that the optimal solution of the continuum-arm limit exactly corresponds not just of some sequence of strategies for the discrete-arm models, but precisely the limit of the optimal strategies. The formal proof is provided in the appendix.

Note that the conditions for the validity of the single-dimensional state variable (breadth $x$) are precisely that Lemmas \ref{lem:lowest_history_of_effort} and \ref{lem:increasing_thresholds} are satisfied. If these do not hold, then the vector of historical efforts will not be able to be summarized by a two-dimensional state (number of ideas and average effort per idea), and instead becomes generically countably-infinite-dimensional; hence, in the continuum limit, the problem would require a higher-dimensional state variable.

Having provided the microfoundation for the continuum model, the next section uses the continuum model to study how agency frictions affect the breadth-depth tradeoff in problem-solving.

\section{Agency Frictions}\label{sec:agency}
Having presented the workhorse continuum model in the previous section, this section provides a proof-of-concept of the tractability of the workhorse model. In particular, I show that the continuum model can be used to analyze a natural ``dynamic multitasking'' problem,\footnote{To my knowledge, this is the first model of a dynamic version of the multitasking problem from \cite{hm1991}.} where the principal cannot monitor \textit{what} the agent chooses to work on. This serves to demonstrate how the limit model allows the derivation of insights into the forces present in new types of moral hazard environments. 

Consider a principal (investor) contracting with an agent (entrepreneur), who is solving a problem that would benefit both parties. As before, suppose that the total value of success is normalized to 1, and that the agent (entrepreneur) privately bears the cost $c$ related to expanding the breadth of the search (i.e., the cost of creativity).
Suppose there is moral hazard in that the principal cannot observe/contract on exactly \textit{what} the agent is working on, but the incentives are aligned insofar as the agent wants to work. Since success is observable, the principal thus chooses a bonus contract, indicating a transfer to the agent upon success, and I will suppose utilities are quasilinear in transfers. I will also suppose the classical limited liability assumption, so the transfer must be nonnegative.\footnote{This is to ensure that the principal cannot ``sell the problem'' to the agent. If transfers can be negative, the agent can just purchase the value of the problem and first-best can be achieved, just as in the classic moral hazard problems.}

I first discuss the optimal static contract to illustrate the principal's fundamental tradeoff. I then discuss time-varying contracts and then illustrate the key dynamics in the shape of the contract.

\subsection{Static Contract}
For the first part of the analysis, let the contracting instrument for the principal be a time-invariant share of the breakthrough value (i.e., a fixed equity stake). That is, the principal chooses an $\alpha$, which denotes the fraction of the value of success granted to the agent. Then the first-best trajectory is exactly given by the solution \eqref{eqn:single_EL} above; denote that trajectory $x_{FB}$. The principal then offers $\alpha$ which solves the following design problem:\footnote{The form of the agent's constraint in problem \eqref{prblm:static_contract} is provided to show the agent solves the baseline problem. In the results that follow, I am not restricting the agent's strategy $x$ to be differentiable, just measurable; integrating the objective of the agent by parts yields the form of the cost without requiring differentiability. }
\begin{gather}
    \max_{\alpha \in [0,1]} \int_0^\infty (1-\alpha) e^{-rt} \dd{F}(x(t,\alpha),t), \label{prblm:static_contract} \\
    \textnormal{subject to } x(\cdot,\alpha) \in \arg\max_{x(\cdot)} \int_0^\infty \left[ \alpha e^{-rt} - \int_0^t  e^{-rs}c \dot{x}(s) \dd{s} \right]\dd{F(x(t),t)} \notag
\end{gather}
The following result shows that under the agency friction, the agent searches more narrowly than the first-best solution.
\begin{proposition} \label{prop:static_share}
    The agent's best response path $x_s(t;\alpha)$ to a static share $\alpha$ is increasing in $\alpha$ for any $t$. In consequence, for the principal's optimal choice of $\alpha^*$, the agent is insufficiently creative relative to the first-best: $x_s(t;\alpha^*) < x_{FB}(t)$ for all $t > 0$. 
\end{proposition}
The proof is a straightforward application of the implicit function theorem and is in the appendix. The intuition is straightforward; since the cost of expanding breadth is privately borne and not contractible, and because the agent receives a smaller fraction of the reward than the first-best, the agent expands breadth less than the first-best would like.

For the principal, increasing the share conceded to the agent mechanically lowers their own payoff, fixing success arrival time; however, conceding a larger share to the agent induces the agent to take a broader search (expand breadth faster), which benefits the principal via a stochastically earlier success time. 
 
Proposition \ref{prop:static_share} should not be surprising, but highlights the essential components of the friction that I focus on; namely, that the search technology is such that exploring new approaches is more productive, but the principal and the agent disagree over the rate of exploration.

\subsection{Dynamic Contract Characterization}
The previous section studied how the optimal static contract shapes an agent's exploration, and how the contract varies with the environment; now, I turn to discussing the principal's intertemporal tradeoffs when they can offer the agent a share that changes over time. Consider the optimal solution to the design problem where the principal now offers a function $\alpha(t)$:
\begin{gather}
    \max_{\alpha: \mathbb{R}_+ \to  [0,1]} \int_0^\infty (1-\alpha(t)) e^{-rt} \dd{F}(x(t,\alpha),t), \label{prblm:dynamic} \\
    \textnormal{subject to } x(\cdot,\alpha) \in \arg\max_{x(\cdot)} \int_0^\infty \left[ \alpha(t) e^{-rt} - \int_0^t  e^{-rs}c \dot{x}(s) \dd{s} \right]\dd{F(x(t),t)} \label{cst:agent-IC}
\end{gather}
To understand the contracting problem and its solution, define following quantity:
\[ I(x,t) := \frac{(1 - F(x,t))r + F_t(x,t) }{rF_x(x,t)}c , \]
which I will interpret as the present-value \textit{incentive cost} required to get the agent to be willing to experiment to state $(x,t)$. More precisely, $I(x,t)$ is the static contract share the principal would have to give the agent in order for the agent to reach state $(x,t)$ on their own. Intuitively, at any state along the first-best trajectory characterized by equation \eqref{eqn:single_EL}, the incentive cost is $1$ (i.e., to sustain the first-best trajectory, the agent must receive the entirety of the value of success). Note that $I$ is only a function of state $(x,t)$, and so the cost can be interpreted as a cost potential, and that $I$ is increasing in $x$ fixing $t$, and decreasing in $t$ fixing $x$. 

The following technical result characterizes the solution to the principal's calculus of variations problem by providing an implicit equation that can be solved pointwise in time to recover the induced agent exploration path. I first state the characterization, and then interpret the economically interesting forces and implications of the result. 
\begin{theorem}\label{thm:dynamic_contract} Suppose the agent trajectory $x^*(t)$ that pointwise maximizes
\begin{equation}
    x^*(t) \in \arg\max_x F(x,t) (1 - I(x,t)) \qquad \forall t, \label{eqn:dyn_a_xmax}
\end{equation}
is admissible and interior. Then the contract 
\begin{equation}
    \alpha^*(t) = e^{rt} \int_t^\infty   e^{-rs}\frac{(1-F(x^*(s),s))r + F_t (x^*(s),s)}{  F_x(x^*(s),s)} c \ \dd{s} = \int_t^\infty r e^{-r(s-t)} I(x^*(s),s) \dd{s} \label{eqn:dyn_a_alpha}
\end{equation}
implements $x^*$ as the agent's unique globally optimal best response and is optimal for the principal among measurable contracts inducing admissible interior agent responses.
\end{theorem}
The proof is left to the appendix; the key idea is to solve the principal's problem by instead imposing the agent's Euler-Lagrange optimality condition as a law on $\alpha$. This yields a variational calculus problem which can be solved using standard methods. 

The conditions \eqref{eqn:dyn_a_xmax} - \eqref{eqn:dyn_a_alpha} implicitly specify the induced path of exploration and the contract up to a measure zero set of times. That is, given the problem setup, equation \eqref{eqn:dyn_a_xmax} gives an implicit characterization for $x^*$, which, if admissible and interior, implies that given this induced path of exploration, the condition \eqref{eqn:dyn_a_alpha} characterizes the contract $\alpha$ that is optimal in the class of measurable contracts with admissible interior responses.

To derive some economic intuition for the characterization, I will use the cost potential interpretation of the object $I$. The condition \eqref{eqn:dyn_a_xmax} requires that at each point in time, the induced exploration pointwise maximizes a canonical notion of profit: the probability of success $F$, times the value of success net of the incentives provided to the agent, $(1-I)$.

Given the trajectory $x^*$, a reformulation of the share law \eqref{eqn:dyn_a_alpha} gives 
\[ \alpha - \frac{1}{r}\dot{\alpha} = I \iff \alpha(t) = \int_t^\infty re^{-r(s-t)} I(x^*(s),s) \dd{s}. \]
By examining the differential equation characterization, observe that incentives can be provided to the agent in two ways; either by a high share presently, or by manipulating the path of the share in the future. So the law \eqref{eqn:dyn_a_alpha} requires that $\alpha$ provides enough incentives for the agent to explore, via either present or future share; by solving the differential equation, the contract share is just the present discounted value of future incentive costs. 

The dynamic tradeoffs present are not easy to see in the mathematical formulation of the result; indeed, the bonus offered here may be nonmonotone. In the following sections, I will highlight the two key economic features at play in the dynamic environment and their influences on the contract using this characterization result. The first is the time-persistent effects of exploration, which arise due to the nature of the moral hazard friction (the dynamic ``multitasking'' problem). This means that dynamic commitment has value, and the principal uses it to exert endogenous time pressure on the agent by frontloading the contract. The second dynamic force arises because of the uncertain difficulty of the problem; I will show that the presence of learning can act as a countervailing force to the principal's desire to frontload agent incentives; here, the principal may wish to backload some incentives in order to encourage exploration despite growing pessimism from the agent.

\begin{customremark}{}
Theorem \ref{thm:dynamic_contract} relies on using first-order conditions in taking the agent's Euler-Lagrange condition, so its scope is restricted to the class of contracts where the agent's response is interior. Because there are rich dynamics for the optimal contract even in this restricted class, I will focus on this setting for this paper; henceforth, in the rest of the paper I will take ``optimal interior contract'' to mean optimal in this class. As a caveat, this implies that eminently plausible features like deadlines are not in the optimization class. In general, the agent's best response to arbitrary contracts (for example, contracts that switch back and forth between 1 and 0) can be noninterior. In such cases, the rectangularity assumption on the agent's state can be unreasonable as a limit of the discrete-arm model as equalizing may be suboptimal. To satisfactorily address the general environment and arbitrary contracts, one would need to enrich the agent's state space to allow a function $d(x)$ to capture the extent of experimentation at each index $x$ with a constraint on the rate of change of its integral. Full characterization of the general optimal contract is an important direction for future work.
\end{customremark}

\subsection{Frontloading and Dynamic Commitment}\label{sec:dynamic_commitment}
First, I show that there is a frontloading force that arises due to the ``solving a problem'' feature of the environment, which induces this dynamic multitasking problem. I will show that this frontloading is reliant on the principal's ability to dynamically commit.
To isolate this effect separately from the learning effect in this section, I will shut down the learning channel and focus on the known difficulty case: $\lambda_H =\lambda_E$. 

\subsubsection{Dynamic Commitment}
To understand the frontloading force, it will be first useful to consider the benchmark environment without dynamic commitment, where the principal can only offer spot contracts in every period. That is, the interaction takes the implicit heuristic timing at each instant $t$:
\begin{enumerate}
    \item The principal sets $\alpha$ for that period.
    \item The agent observes $\alpha$ and chooses a control $u = \dot{x}$, which depends on $\alpha$.
    \item Breakthrough realizes if it happens, otherwise play advances to the next period.
\end{enumerate}
The key difference in the no-dynamic commitment problem relative to the original problem \eqref{prblm:dynamic} is that the choice of $\alpha$ must satisfy Markov perfection. More precisely, the variational approach used in solving the dynamic contract used the agent's Euler-Lagrange equation; this is equivalent to maximizing the principal's choice of share, agent exploration recommendation, and agent \textit{costate} variable, where the costate variable captured the ability for the principal to commit to future utility promises; this commitment ability is not available once Markov perfection is imposed.

\begin{proposition}[No Dynamic Commitment]\label{prop:no_commit}
    When $\lambda_H =\lambda_E> 0$ and without dynamic commitment, there is a Markov perfect equilibrium where the principal offers the agent a time-invariant share $\alpha_\nc$ each period, and the agent on-path maintains a constant depth $d_\nc$, so the induced trajectory is $x_\nc(t) = t/d_\nc$.
\end{proposition}
Intuitively, in focusing on the role of commitment and shutting down the forces induced by learning, Proposition \ref{prop:no_commit} shows that the lack of dynamic commitment means there is stationarity in the problem; thus, in a no-commitment MPE, the principal offers the same share every period.

\subsubsection{Persistent Benefits of Exploration}
Having established that the benchmark when the principal has no dynamic commitment is a stationary solution, the following proposition shows that by contrast, the Theorem \ref{thm:dynamic_contract} solution offers a decreasing share over time.
\begin{proposition}[Time Pressure]\label{prop:dec_share}
    When $\lambda_H =\lambda_E> 0$, if an optimal interior dynamic contract exists, it induces a path $x_C$ with less exploration than the first best, grants the agent a share that is strictly decreasing over time and approaches $c/\nu_0$ in the limit as $t\to\infty$, and explores asymptotically slower than the no-commitment equilibrium: 
    \( \lim_{t \to \infty} \frac{x_C(t)}{x_\nc(t)} = 0 .\)
\end{proposition}
The proof of Proposition \ref{prop:dec_share} applies Theorem \ref{thm:dynamic_contract}, and is left to the appendix. The proposition implies that the agent receives a larger share of a potential success if it arrives early on, and this share declines over time; the principal \textit{frontloads} incentives. This implies that dynamic commitment has value to the principal by enabling the principal to exert time pressure on the agent.

Why should the principal want to frontload the contract, and why does dynamic commitment have value to the principal? In the lens of \cite{fhm1990}, the key feature allowing dynamic commitment to have value is the failure of the assumption on common knowledge of future technology. In particular, because the state of exploration also has a persistent impact on the future rates of success, the agent has better information about the success arrival technology because of their private observation of breadth; this introduces a dependency in the future rate of success on past agent choices which are not revealed and not contractible to the principal. Note that this dependency is not due to learning, but the persistent benefits of exploration. Alternatively, under an interpretation of the model as a dynamic ``multitasking'' problem \citep{hm1991}, the agent's multitasking choice in any period (how much breadth versus depth to pursue) impacts the success technology and the tradeoff between breadth and depth \textit{at all future times} as well.

Figure \ref{fig:dynamic_vs_static_contract} plots the comparison of the share of breakthrough granted to the agent under the dynamic contract and the no dynamic commitment benchmark. As Proposition \ref{prop:dec_share} suggests, the dynamic contract share is decreasing.
The asymptotic limit of the dynamic share approaches $c / \nu_0$; this is the infimum of all shares where the agent would be willing to explore, since expanding breadth costs the agent $c$ and yields a valid approach at Poisson rate $\nu_0$ in breadth. In consequence, in the large time horizon without success, the exploration induced by the optimal contract is asymptotically more narrow than that induced by the optimal static contract. Figure \ref{fig:dynamic_vs_static_breadth} plots the difference in breadth over time, illustrating how the dynamic contract induces more exploration early, but results in less exploration later.

\begin{figure}
    \centering
    \includegraphics[width=0.65\linewidth]{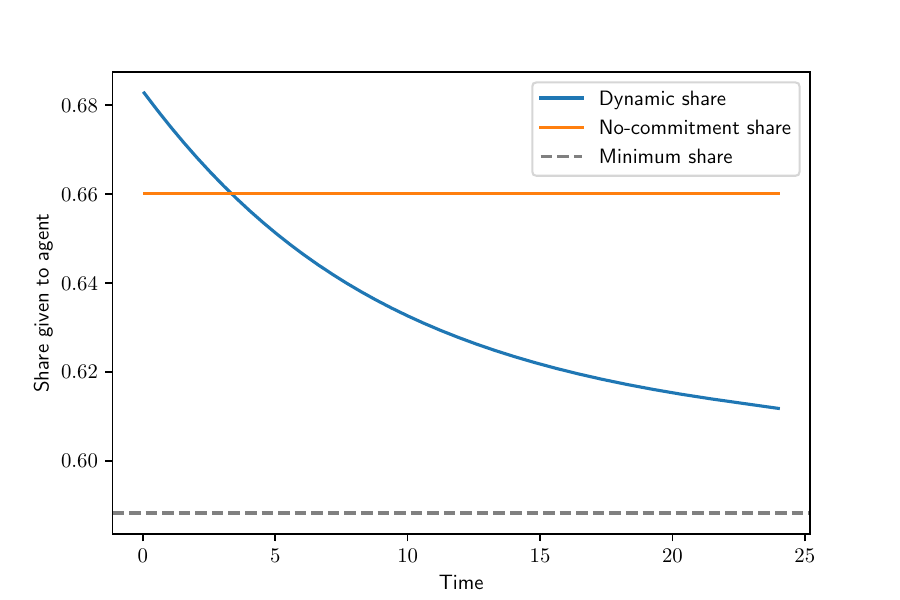}
    \caption{The dynamic contract share given to the agent over time, and the share granted to the agent in the no-commitment environment. Solution plotted for $\lambda = 1$, $c = 0.5$, $\nu_0 = 0.85$, and $r=1$. }
    \label{fig:dynamic_vs_static_contract}
\end{figure}

\begin{figure}
    \centering
    \includegraphics[width=0.65\linewidth]{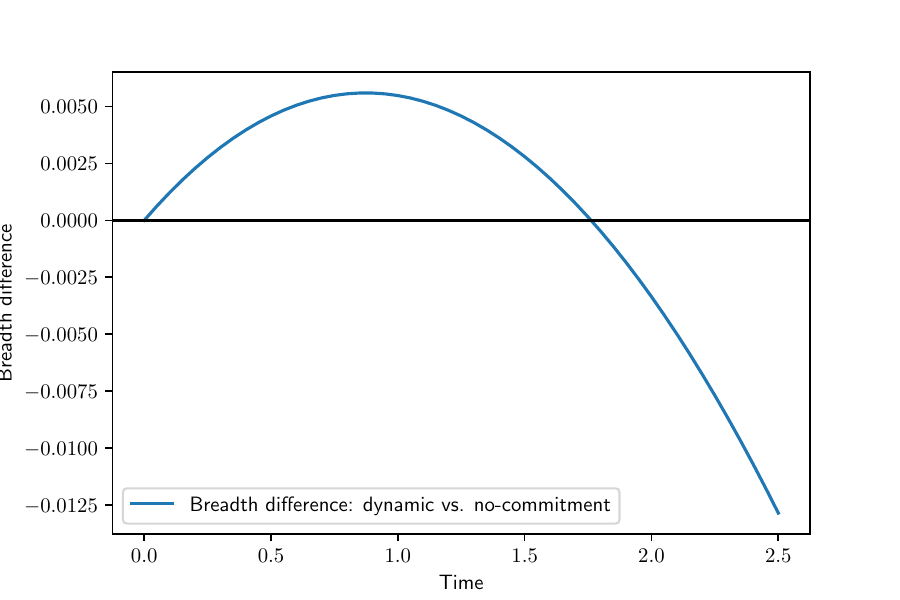}
    \caption{The path of exploration in the dynamic contract vs the no-commitment equilibrium. The difference in breadth over time is plotted.  Solution plotted for $\lambda = 1$, $c = 0.5$, $\nu_0 = 0.85$, and $r=1$. }
    \label{fig:dynamic_vs_static_breadth}
\end{figure}

\paragraph{Comparison to Literature} This contrasts with the canonical dynamic principal–agent models of \cite{r85}, \citet{ss87}, and \cite{s08}, where the principal uses dynamic commitment to backload incentives. In those settings, the agent’s effort in each period stochastically yields higher output only in that period; in that case, the principal uses rising continuation values to sustain incentives, leading to contracts in which rewards are pushed into the future. In this model, by contrast, the moral hazard arises over the agent's decision to open new approaches, taking a broader search. Because this effort not only increases the instantaneous probability of success but also permanently raises the likelihood of \textit{future} breakthrough, the principal has the strongest motive to reward the agent up front, when marginal incentives are most effective in accelerating progress. The fact that the agent manages a state over time thus results in the principal utilizing dynamic commitment in a qualitatively distinct way. 

This frontloading force is distinct from those identified by \cite{bh2005} and \cite{hs2013}, which feature a different type of moral hazard. In their models, the agent receives funds, which the agent could divert for their own personal gain rather than use towards eventual success. Since success ends the agent’s rent stream, higher continuation values create a disincentive to use funds appropriately in the short-term; as a result, the principal responds by making the future less attractive. In my setting, by contrast, the principal and agent are both aligned in wanting breakthroughs to arrive as soon as possible; the moral hazard here arises not because the agent wants to divert funds, but because of disagreement about how the agent pursues success.

\subsection{Learning About Difficulty}
Having illustrated how the nature of the moral hazard problem leads the principal to frontload incentives, I now show how learning about difficulty serves as a countervailing force. The first section focuses on isolating the learning effect, and the second section shows that this can result in non-monotonicities in the dynamic contract.

\subsubsection{Backloading and Isolating the Learning Effect}
In order to isolate the learning effect, I will shut down the dynamic ``multitasking'' problem. Consider an alternative benchmark, where the agent does not face the exploration-exploitation tradeoff described previously, but instead manages an \textit{extensive} margin of effort on a single approach. That is, suppose the principal contracts with the agent to solve an ``exercise''; i.e., the agent is endowed with one approach which is commonly known to be valid. The agent is still uncertain about difficulty, so $\theta$ is unknown. As such, the agent chooses the effort to exert on the given approach; given $K$ cumulative effort, the probability of the exercise remaining unsolved is 
\[ 1 - F_\theta(K) := e^{-\lambda_\theta K}. \]
Instead of allocating a unit of effort between breadth/depth, the agent chooses effort $k \in [0,1]$, at a marginal cost of effort $\gamma < \lambda_H$. The design problem is then
\begin{gather}
    \max_{\alpha(t): \mathbb{R}_+ \to [0,1]} \mathbb{E}_\theta \left[ \int_0^\infty (1-\alpha(t)) e^{-rt} \dd{F}_\theta(K^*(t)) \right] ,\label{prblm:extensive} \\ 
    \textnormal{ subject to: } K^*  \in \arg\max_{K, \dot{K}=k \in [0,1]} \mathbb{E}_\theta\left[ \int_0^\infty \alpha(t) e^{-rt} \dd{F}_\theta(K(t))- \int_0^\infty (1 - F_\theta(K(t))) e^{-rt} \gamma k(t) \dd{t}\right] .  \tag{IC}
\end{gather}
\begin{proposition}\label{prop:extensive}
 With dynamic commitment but moral hazard on only the extensive margin, the solution to \eqref{prblm:extensive} offers a share that increases ($\dot{\alpha} \ge 0$) at almost every time that the agent exerts effort $(k(t) > 0)$.
    \label{prop:isolated_backloading}
\end{proposition}
The result follows from straightforward variational calculus, and the details are left to the appendix. Proposition \ref{prop:extensive} shows that without the intensive margin of moral hazard, the dynamic contract with commitment in the extensive margin is backloaded. This arises precisely because the agent and principal become more pessimistic about difficulty over time, so the principal cedes a larger bonus in order for the agent to continue exploration.

\subsubsection{Nonmonotonicity}
To build some intuition on how these two countervailing forces interact, I will first focus on the most extreme learning case, where the hard problem is impossible (i.e., $\lambda_E = \lambda > \lambda_H = 0$), and then discuss what happens with $\lambda_E > \lambda_H > 0$.

When $\lambda_H = 0$, the key belief to track is the belief $\delta(x,t)$ that an agent at state $(x,t)$ has about whether the problem is hard. By Bayes' rule, this belief is given by 
\[ \delta(x,t) = \frac{\delta_0}{\delta_0 + (1-\delta_0)\exp(-\nu_0 x(1 - e^{-\lambda t/x} )) }. \]
Note that $x (1 - e^{-\lambda t/x})$ is increasing in $x$ fixing $t$. Therefore, comparing the beliefs over difficulty between two agents with different levels of exploration at the same time, the agent that has explored more (larger $x$) is more pessimistic about the difficulty of the problem. When the agent becomes sufficiently pessimistic about the feasibility of the problem, the agent becomes unwilling to explore new approaches. Since the principal does not bear the creativity cost that the agent incurs to explore, the principal would still like the agent to explore even if pessimistic about the difficulty of the problem; thus, to encourage the agent to continue exploring, the principal is willing to relinquish some of the share. 

With learning about difficulty, the principal can increase the share offered to the agent, despite a positive $\dot{\alpha}$ weakening the agent's short-term incentives to explore. This occurs because increasing the share offered to the agent allows the principal to induce more long-run exploration than a static contract would allow. When there is learning about impossibility, the share offered to the agent can rise over time; since better rewards in the future may induce the agent to reduce exploration, the dynamic contract compensates for a rising path with a higher level of $\alpha$. This benefits the principal in the longer-term; Figure \ref{fig:learning} plots an example of the dynamic contract when learning about impossibility, comparing showing how at different priors on difficulty $\delta_0$, the contract takes different shapes. Note that when the prior on difficulty is higher (i.e., $\delta_0 = 0.3$ in the figure), the optimal contract offers a larger share of success for later successes.

\begin{figure}
    \centering
    \begin{subfigure}[b]{0.32\textwidth}
         \centering
         \includegraphics[width=\textwidth]{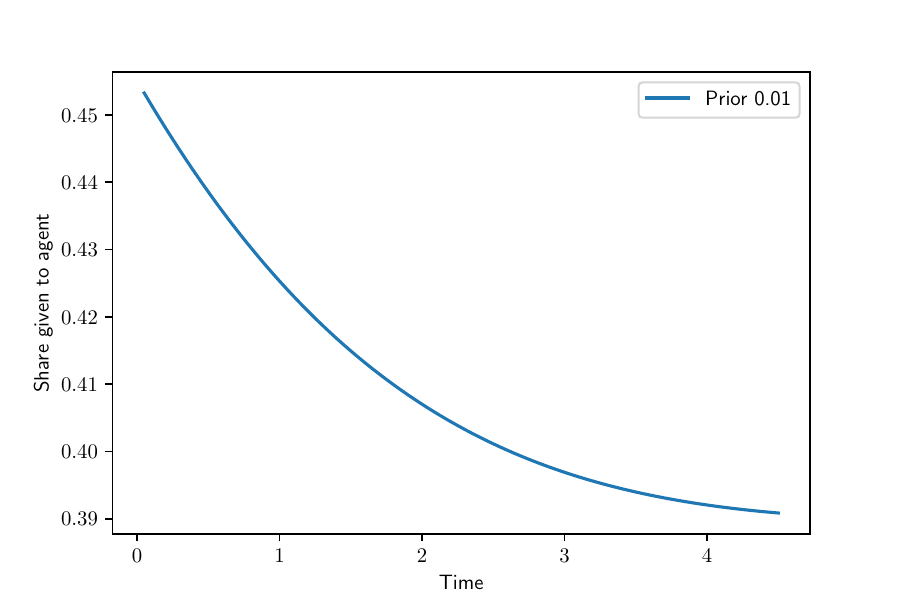}
         \caption{Prior $\delta_0 = 0.01$}
         \label{fig:learning_share_small}
     \end{subfigure}
     \hfill
    \begin{subfigure}[b]{0.32\textwidth}
         \centering
         \includegraphics[width=\textwidth]{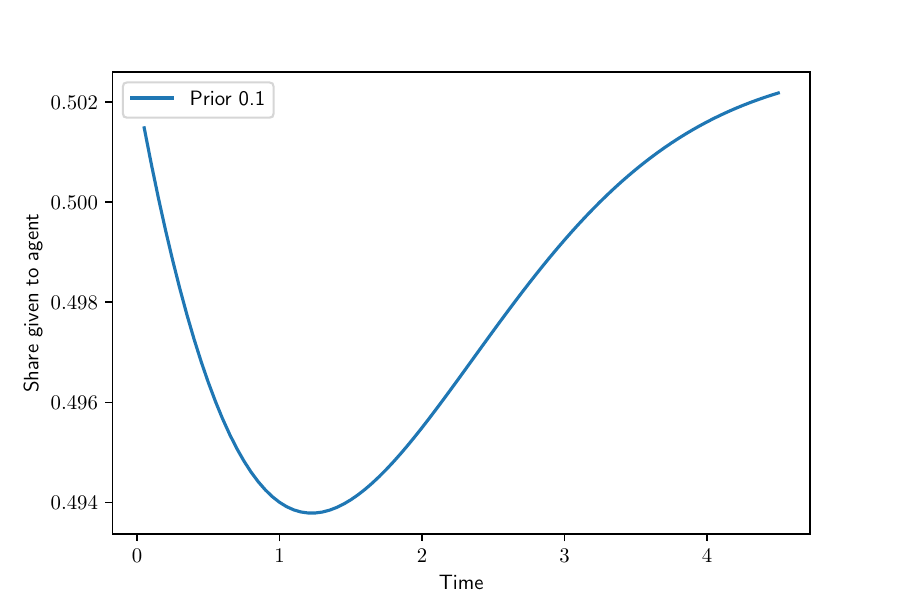}
         \caption{Prior $\delta_0 = 0.1$}
         \label{fig:learning_share_medium}
     \end{subfigure}
     \hfill
    \begin{subfigure}[b]{0.32\textwidth}
         \centering
         \includegraphics[width=\textwidth]{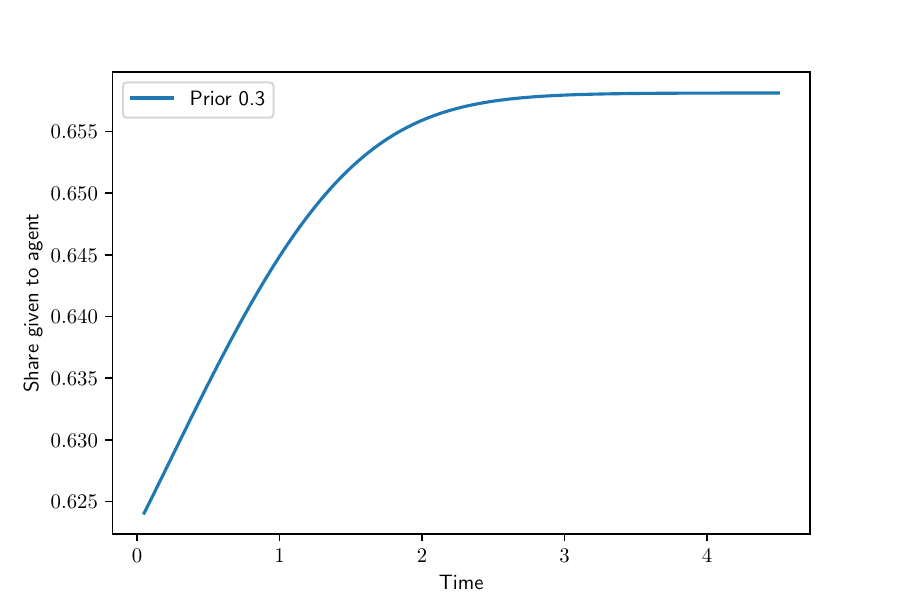}
         \caption{Prior $\delta_0 = 0.3$}
         \label{fig:learning_share_large}
     \end{subfigure}
    \caption{The characterized interior solution when learning about impossibility; illustrated with $\lambda_E = 3$, $\lambda_H = 0$, $\nu_0 = 0.9$, $c = 0.3$, and $r = 1$. The prior $\delta_0$ varies.}
    \label{fig:learning}
\end{figure}

\paragraph{Intermediate Learning} The previous discussion focused on the two extremes of $\lambda_H \in [0,\lambda_E]$. At the upper end of the range when $\lambda_H = \lambda_E$, the principal wishes to frontload incentives as illustrated in Section \ref{sec:dynamic_commitment}. In the intermediate case where $\lambda_H > 0$, it turns out the first force prevails in the long-run; the principal eventually exerts the dynamic commitment pressure.
\begin{proposition}\label{prop:asymptotic_decrease}
    If an optimal interior contract exists and $\lambda_H > 0$, in the asymptotic limit as $t \to \infty$, $\alpha^*(t) \to c/\nu_0$ from above.
\end{proposition}
The proof applies Theorem \ref{thm:dynamic_contract} and is left to the appendix. In general in the interim, the interaction of these two forces can lead to non-monotonicities in the optimal contract. Figure \ref{fig:interaction} illustrates this interaction, where the optimal share is neither quasiconvex nor quasiconcave. In this example, the prior is initially almost sure that the problem is easy, and the contract initially resembles the contract as if the problem were easy for sure. However, as time passes, in an intermediate region of time, the principal increases the share to encourage the agent to continue exploring despite growing pessimism; however, in the long-run, as Proposition \ref{prop:asymptotic_decrease} shows, the frontloading incentive prevails.

\begin{figure}
     \centering
     \begin{subfigure}[b]{0.45\textwidth}
         \centering
         \includegraphics[width=\textwidth]{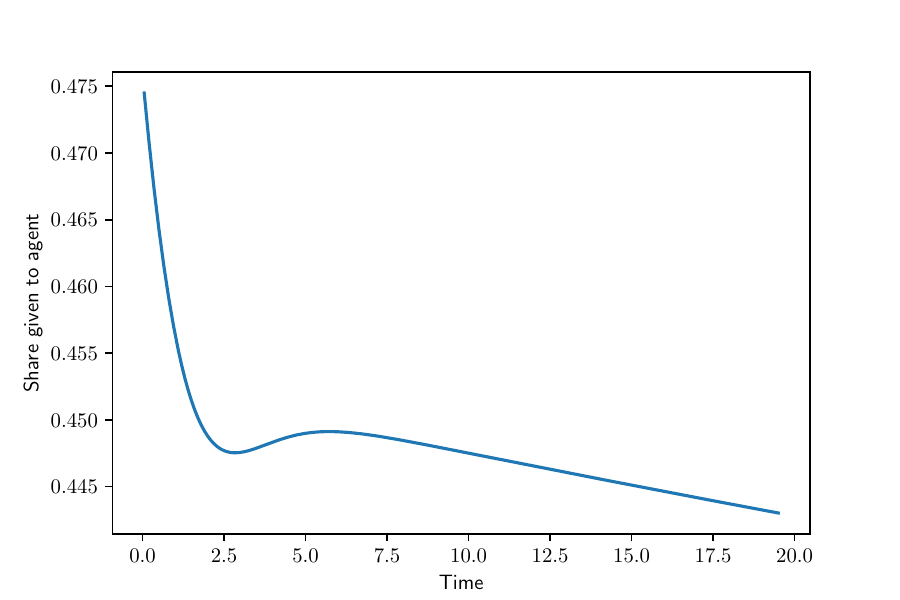}
         \caption{Share  $\alpha$ offered over time}
         \label{fig:interaction_share}
     \end{subfigure}
     \hfill
     \begin{subfigure}[b]{0.45\textwidth}
         \centering
         \includegraphics[width=\textwidth]{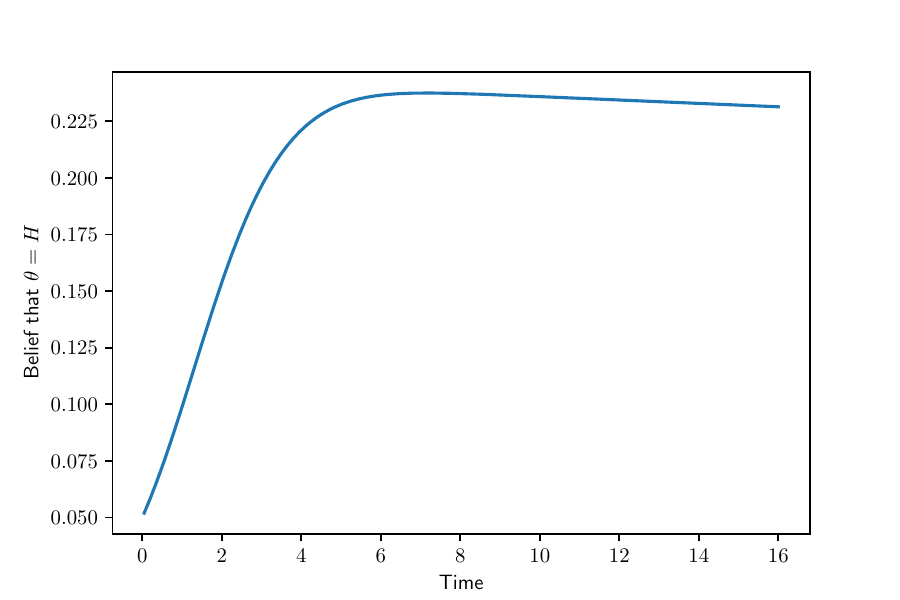}
         \caption{Belief over difficulty $\delta$}
         \label{fig:interaction_belief_d}
     \end{subfigure}
     \hfill
     \begin{subfigure}[b]{0.45\textwidth}
         \centering
         \includegraphics[width=\textwidth]{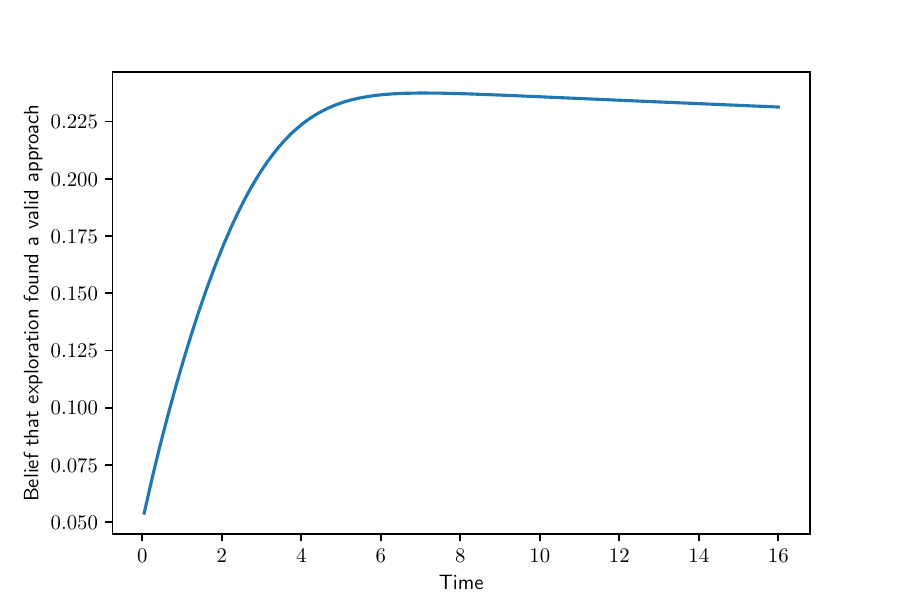}
         \caption{Belief that exploration has yielded a valid approach}
         \label{fig:interaction_belief_v}
     \end{subfigure}
     \hfill
     \begin{subfigure}[b]{0.45\textwidth}
         \centering
         \includegraphics[width=\textwidth]{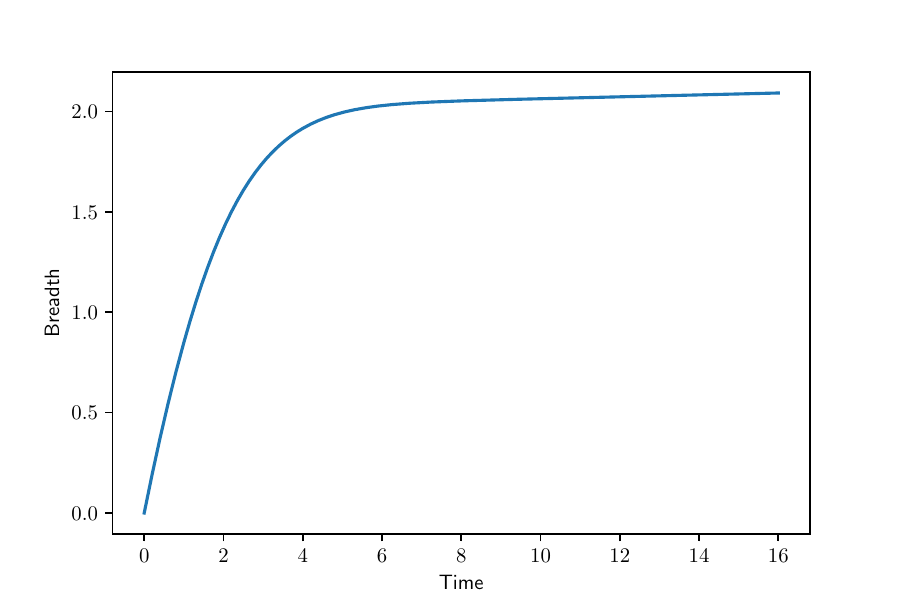}
         \caption{Exploration path induced}
         \label{fig:interaction_x}
     \end{subfigure}
        \caption{The characterized interior solution with interaction between the principal's tradeoffs; illustrated with $\lambda_E = 3$, $\lambda_H = 0.05$, $\delta_0 = 0.05$, $\nu_0 = 0.9$, $c = 0.3$, and $r = 1$. }
        \label{fig:interaction}
\end{figure}

\section{Robustness and Extensions} \label{sec:robustness}
In this section, I discuss how the main characterization generalizes to other environments beyond the pure exponential bandit framework in the baseline model and with heterogeneous approaches.

\subsection{Broader Notions of Difficulty}\label{sec:generalization}
In the baseline model, I modeled difficulty as the perfectly correlated rate at which a valid approach yields a solution. However, the insight of Theorem 1 extends beyond this formulation to a more general notion of difficulty. Here, now suppose that instead of the Poisson rate of success arrival per unit effort on approach $i$ being given by $\lambda_\theta \omega_i$, instead suppose that it is given by an unobserved, randomly drawn rate $\lambda_i \sim_{\textnormal{i.i.d.}} G_\theta$. Here, the rate of breakthrough is drawn independently for each approach, but from an unknown distribution of rates, which the agent learns about; this learning induces the endogenous correlation between approaches. Assume $G_E$ first-order stochastically dominates $G_H$.

In order to preserve the same tradeoff effect in learning, namely the tension between pessimism about an approach's value pushing the agent to explore new approaches and the uncertainty about difficulty inducing the agent to persist with an approach rather than explore, I impose the following assumption that ensures the second part of the tradeoff holds:

\begin{customassump}{}[Difficulty Requires Patience, DRP] Let the Gittins index maximizer of the known difficulty problem when the state is $\theta$ be $\hat{K}_\theta$ (possibly infinite). Then $\hat{K}_H > \hat{K}_E$.
\end{customassump}

Note that this generalizes the specification in the baseline model; to replicate the baseline, let $G_H$ correspond to the two-point distribution over $\{0, \lambda_H\}$ with $\nu_0$ probability of $\lambda_H$, and let $G_E$ be the two-point distribution over $\{0, \lambda_E\}$ with $\nu_0$ probability of $\lambda_E$. 

As before, it will be useful to define the survival probability and expected hazard rate:
\begin{align*}
    S_\theta(K) &:= \int e^{-\lambda K} \dd{G}_\theta(\lambda), &
    \lambda_\theta(K) &:= \frac{\int \lambda e^{-\lambda K} \dd{G}_\theta(\lambda) }{S_\theta(K)}.
\end{align*}
If the cost $c$ is too large, the agent is unwilling to begin solving the problem, and the solution is degenerate. So for the sake of focusing on the interesting case, suppose that it is sufficiently small:
\begin{equation} c < \mathbb{E}_\theta\left[ \int_0^\infty e^{-rt} \lambda_\theta(t)S_\theta(t) \dd{t} \right]. \label{eqn:learning_c_assum_general}\end{equation}
Then the analogue of the Gateaux derivative \eqref{eqn:phi} is given by 
\begin{equation} \hat\phi_\theta(K) := \lambda_\theta(K)  - ( r + \lambda_\theta(K))  \left[ -c + \int_0^K e^{-rt}\lambda_\theta(t)S_\theta(t) \dd{t} \right] - e^{-rK}S_\theta(K)\lambda_\theta (K). 
    \label{eqn:hatphi}
\end{equation}
Using this, the generalization of Theorem \ref{thm:learning} becomes:

\begin{customthm}{1c}\label{thm:learning_general}
    Define the sequence of thresholds $\{ K^*_n \}_{n\in \mathbb{N}_+}$ implicitly by the smallest solution to the equation
    \begin{equation}\label{eqn:def_tn_general}
        (1-\delta_0) S_E(K^*_n) ^n\hat\phi_E(K^*_n)  + \delta_0 S_H(K^*_n)^n \hat\phi_H(K^*_n) = 0.
    \end{equation}    
    Denote $K^*_n = \infty$ if no solution exists. The sequence $K^*_n$ is increasing, and the optimal research path of the agent generates a new approach if $\min_{i \in [N]} K_i \ge K^*_N$, else exerts effort on all approaches in $B$ with effort $1/|B|$.
\end{customthm}
The proof logic is exactly the same as in Theorem \ref{thm:learning}, and the details of the generalization are left to the online appendix.
Note that an equivalent specification of the general model would be having the success arrival time on any given approach independently drawn from some cumulative distribution function $\hat{F}_\theta: \mathbb{R}_+ \to [0,1]$, such that the survival function $1 - \hat{F}_\theta$ is completely monotone. The inversion to recover the distribution of rates $G_\theta$ is then given by the inverse Laplace transform: $G_\theta(\lambda) = \mathcal{L}^{-1}\{1-\hat{F}_\theta\}(\lambda)$.

As a final remark, with arbitrary distributions of rates (or equivalently, arbitrary arrival distributions of successes), the limit model will require different notions of convergence in order to yield a well-specified problem. 

\subsection{Heterogeneous Approaches}\label{sec:heterogeneity}
The baseline model assumed that approaches are ex-ante identical, but the core insights extend even with heterogeneous arms. When approaches differ in their prior probability of being valid, the agent pursues the most probably valid approaches first. Importantly, the equalizing lemma would not equalize over historical effort, but over hazard rates. Interestingly, heterogeneity also provides a microfoundation for the Corollary \ref{corr:juggling_recall} behavior even with known problem difficulty, as increasing pessimism about future approaches also induces the agent to spend longer in between brainstorming. However, extending heterogeneous-arm models to strategic settings and applying the continuum framework present substantial technical issues; here, the equalizing lemma takes advantage of symmetry to aggregate prior approaches into a one-dimensional sufficient statistic (depth), which implies that along the optimal path, the number of approaches started is a deterministic function of the maximum effort spent on any given approach. With heterogeneity, one would have to track a measure over the arms instead of a single number; in particular, the optimal policy would have to equalize the hazard rates of success arrival among approaches and split effort in a way to maintain that equality. A full treatment of heterogeneity remains an important direction for future research.

\section{Discussion and Future Work}\label{sec:conclusion}
This paper has shown that uncertainty about problem difficulty fundamentally alters the way agents pursue solutions. Because failures are ambiguous (suggesting either a poor approach or an inherently hard task) problem solvers optimally alternate between expanding their set of approaches and revisiting previously abandoned ones. This generates dynamic patterns that contrast sharply with the sequential exploration strategies predicted by standard bandit models. Developing a continuum formulation made it possible to extend the framework to strategic environments with incentive frictions, which highlight new forces in familiar economic settings. In a principal–agent environment, moral hazard arises not on the extensive margin of effort, but on the intensive margin of how broadly to search. This creates a novel frontloading force, in contrast to the backloaded continuation-value structures found in classical papers in the literature. 

Several avenues for future research remain open. First, examining competitive and multi-agent environments with overlapping search spaces would shed light on the interplay between competitive pressures and what approaches agents choose to pursue. Allowing for richer informational spillovers across arms could shed light on the balance between duplication and diversity of effort in scientific and technological races. Second, in many settings, the actions taken by problem solvers change how hard the remaining problem becomes: a bad policy implementation might affect the efficacy of future policies. Modeling problem difficulty as evolving endogenously would provide further insights into modeling real-world frictions affecting learning. Finally, incorporating strategic feedback would allow the problem environment itself to respond to the solver’s actions. For example, strategic decisions by a receiver on whether to follow recommendations would influence the recommender's learning, and it is not ex-ante obvious how strategic responses would alter a recommender's exploration-exploitation tradeoff.

By developing a model of problem solving, this paper provides a foundation for a broader research agenda. The results suggest that revisiting previously discarded approaches is a feature of problem solving under uncertainty, and that moral hazard in environments where creativity is essential introduces a reason to frontload incentives, which is counteracted by the presence of learning under uncertainty. Understanding these dynamics more fully (especially in competitive or adaptive environments) offers a promising path for advancing both theory and applications in economics of innovation, contracting, and organizational behavior.
\newpage

\bibliographystyle{agsm}
\bibliography{biblio}

\newpage

\appendix

\section{Appendix: Omitted Proofs}
\subsection{Omitted Proofs in Section \ref{sec:analysis}}
\subsubsection{Omitted Proofs in Section \ref{sec:benchmark}}
Before I prove Proposition \ref{prop:benchmark}, it is useful to note that an equivalent formulation of \eqref{defn:taustar} is 
\begin{equation}  \left(1+   \frac{c(r+\lambda)}{\lambda(1-\nu_0)} \right)=  \frac{\lambda}{r+\lambda}e^{-r\tau_D} +  \left(\frac{r}{r+\lambda} - \frac{cr}{\nu_0\lambda}\right)e^{\lambda\tau_D}. \label{defn:taustar_appendix} \end{equation}
Since the derivation is purely mechanical, it is left to the online appendix.

\begin{proof}[Proof of Proposition \ref{prop:benchmark}]
Consider the following equivalent multi-armed bandit representation, where there are a countably infinite set of arms. Each arm $i$ instead starts in an ``undiscovered'' state $U$, which transitions to the state ``discovered'' $D$ immediately but induces cost $c$. In state $D$, the arm then transitions into a success state $S$ at rate $\lambda \omega_i$. The arm yields a flow reward $r$ when in state $S$, and no flow rewards otherwise. (Although the problem was initially defined as success bringing a lump-sum reward of size 1, this is mathematically equivalent to a perpetual flow of $r$.) That is, let $U_t$ be 0 if the Markov state is $U$ or $D$ at time $t$, and $r$ if it is $S$. Following Corollary 2.1 of \cite{bk2007}, define the continuous-time Gittins index process for an arm at time $t$ as 
\[ g_t :=  \sup_{\tau} \frac{\mathbb{E}\left.\left[ \int_0^\tau e^{-rs } U_s \dd{s} \right| \mathcal{F}_t\right]}{\mathbb{E}\left.\left[ \int_0^\tau e^{-rs} \dd{s} \right| \mathcal{F}_t\right]}. \]

It is straightforward to see that if the Markov chain is in state $S$ at time $t$, the Gittins index process must be $g(S) = r$, and that if the chain transitions to state $S$, the stopping time that generates the highest average discounted payoff must be $\infty$, since state $S$ generates flow payoff $r$, which is larger than the average discounted flow payoff from any strategy that starts in a state that is not $S$.
Further, if the Markov chain is in state $D$ and $K$ effort has been exerted without a transition, \[ \mathbb{E}[\omega_i | \mathcal{F}_t]  = \nu_t := \frac{\nu_0 e^{-\lambda K}}{\nu_0 e^{-\lambda K} + 1 - \nu_0}. \]
Since a stopping time that is $\mathcal{F}_t$-measurable specifies after how long to stop in each state, I can compute the Gittins index process with respect to the stopping times in $D$, so
\begin{align*} g(D,t)  &= \sup_{\tau_D} \left \{ \frac{\nu_t \int_0^{\tau_D} \lambda e^{-\lambda t_D} \left(\int_0^{\infty} re^{-r(t_D + t_S)} \dd{t_S}  \right) \dd{t_D} }{\nu_t \int_0^{\tau_D} \lambda e^{-\lambda t_D} \left(\int_0^{\infty} e^{-rs} \dd{s}  \right) \dd{t_D}  + (1-\nu_t+\nu_te^{-\lambda \tau_D}) \int_0^{\tau_D} e^{-rs}\dd{s}} \right \}, \\
&= \sup_{\tau_D} \left \{ \left( \frac{r\lambda \nu_t}{\lambda + r} \right)\frac{ \left(1 - e^{-(r+\lambda)\tau_D}\right)}{\nu_t \int_0^{\tau_D} \lambda e^{-\lambda t_D} \dd{t_D} + (1-\nu_t+\nu_te^{-\lambda \tau_D})\left(1 - e^{-r\tau_D} \right)} \right \},\\
&= \sup_{\tau_D} \left \{ \left( \frac{r\lambda \nu_t}{\lambda + r} \right)\frac{ \left(1 - e^{-(r+\lambda)\tau_D}\right)}{1- e^{-r\tau_D}(1-\nu_t+\nu_te^{-\lambda \tau_D}) } \right \}.
\end{align*}
Note that $\int_0^\infty e^{-rs} \dd{s} = 1/r$, so the first stpe just multiplies the numerator and denominator by $r$ and simplifies.
It is not hard to check that the expression is strictly decreasing in $\tau_D$, and so attains the sup in the limit as $\tau_D \to 0$, so 
\[ g(D,t) = \frac{r\lambda \nu_t}{r + \lambda \nu_t}. \]
Now, one can back out the Gittins index process when the chain is in state $U$. Note that regardless of the prior history of effort in state $U$, the payoff-relevant expectation $\mathbb{E}[\omega_i | \mathcal{F}_t] = \nu_0$ when in state $U$. A $\mathcal{F}_t$-measurable stopping time consists of three stopping thresholds, one for each state: $\tau_U, \tau_D, \tau_S$. However, by the memoryless property of the exponential distribution, any stopping time maximizing the time-discounted average payoff must not stop in state $U$, so $\tau_U = \infty$, and as argued previously the average-payoff-maximizing stopping time must also set $\tau_S = \infty$. Hence, it is sufficient to take the sup over $\tau_D$. Given this, the discounted payoff of this stopping rule is given by calculation: 
\begin{align*} 
-c + \mathbb{E}\left[\int_0^\tau e^{-rs}U_s \dd{s} \right] =& -c + \left( \nu_0 \int_0^{\tau_D} \lambda e^{-\lambda t_D} \left(\int_0^{\infty} e^{-r(t_D + t_S)} r \, \dd{t_S}\right) \dd{t_D} \right)  \\
=& -c + \left( \nu_0 \int_0^{\tau_D} \lambda e^{-\lambda t_D}e^{-rt_D} \dd{t_D} \right) -c +  \nu_0 \frac{\lambda}{r+\lambda}\left(1- e^{-(r+\lambda) \tau_D}\right)
\end{align*}
The expected discounted time spent is: 
\begin{align*} 
\mathbb{E}\left[\int_0^\tau e^{-rs} \dd{s} \right] 
=&\frac{1}{r}\mathbb{E}\left.\left[\int_0^\tau re^{-rs} \dd{s} \right| D\right]
= \frac{1}{r}\left( 1 - e^{-r\tau_D}(1 - \nu_0 + \nu_0e^{-\lambda\tau_D}) \right).
\end{align*}
And so the Gittins index is:
\begin{align*} g(U)  &= \sup_{\tau_D} \left\{ \frac{-c +  \nu_0 \frac{\lambda}{r+\lambda}\left(1- e^{-(r+\lambda) \tau_D}\right)}{\frac{1}{r}\left( 1 - e^{-r\tau_D}(1 - \nu_0 + \nu_0e^{-\lambda\tau_D}) \right)} \right\}= \sup_{\tau_D} \left\{ \frac{-cr + \nu_0 \left(\frac{r\lambda}{r+\lambda}\right)\left(1- e^{-(r+\lambda) \tau_D}\right)}{1 - e^{-r\tau_D}(1 - \nu_0 + \nu_0e^{-\lambda\tau_D})} \right\}.
\end{align*}
Some algebra yields the following lemma:
\begin{lemma}\label{lem:gittins}
    The Gittins index $g(U)$ is attained when $\tau_D$ is the solution to \eqref{defn:taustar}. Further, the objective 
    \[ \frac{-cr + \nu_0 \left(\frac{r\lambda}{r+\lambda}\right)\left(1- e^{-(r+\lambda) \tau_D}\right)}{1 - e^{-r\tau_D}(1 - \nu_0 + \nu_0e^{-\lambda\tau_D})} \]
    is strictly increasing for $\tau_D < K^*$, and strictly decreasing for $\tau_D > K$.
\end{lemma}
\begin{proof}[Proof of Lemma \ref{lem:gittins}]
    First, for obtaining the sup over $\tau_D$, omit the constant $r$. 
    To maximize 
    \[ G(\tau_D) =  \frac{-c + \nu_0 \left(\frac{\lambda}{r+\lambda}\right)\left(1- e^{-(r+\lambda) \tau_D}\right)}{1 - e^{-r\tau_D}(1 - \nu_0 + \nu_0e^{-\lambda\tau_D})} , \]
    take the derivative 
    \begin{eqnarray*} \frac{dG}{d\tau_D} &=& \frac{\left( \nu_0\lambda \right)e^{-(r+\lambda)\tau_D}\left(1 - e^{-r\tau_D}(1 - \nu_0 + \nu_0e^{-\lambda\tau_D})  \right)}{\left(1 - e^{-r\tau_D}(1 - \nu_0 + \nu_0e^{-\lambda\tau_D}) \right)^2}, \\
    &&- \, \frac{\left(-c + \nu_0 \left(\frac{\lambda}{r+\lambda}\right)\left(1- e^{-(r+\lambda) \tau_D}\right)\right)\left(re^{-r\tau_D}(1 - \nu_0 + \nu_0e^{-\lambda\tau_D}) + \nu_0\lambda e^{-(\lambda+r)\tau_D} \right)}{\left(1 - e^{-r\tau_D}(1 - \nu_0 + \nu_0e^{-\lambda\tau_D}) \right)^2}. \end{eqnarray*}
    Expanding each of the numerators,
    \begin{eqnarray*}
        \frac{dG}{d\tau_D} &=& \frac{\left( \nu_0\lambda \right)e^{-(r+\lambda)\tau_D}\left(1 - (1-\nu_0)e^{-r\tau_D} - \nu_0e^{-(r+\lambda)\tau_D}) \right) }{\left(1 - e^{-r\tau_D}(1 - \nu_0 + \nu_0e^{-\lambda\tau_D}) \right)^2}, \\
        && - \, \frac{\left(-c + \nu_0 \left(\frac{\lambda}{r+\lambda}\right)\left(1- e^{-(r+\lambda) \tau_D}\right)\right)\left( re^{-r\tau_D} (1-\nu_0) + \nu_0(r+\lambda)e^{-(r+\lambda)\tau_D}  \right)}{\left(1 - e^{-r\tau_D}(1 - \nu_0 + \nu_0e^{-\lambda\tau_D}) \right)^2}.
    \end{eqnarray*}   
Pulling out common factor $e^{-(r+\lambda)\tau_D}$, the expression becomes
    \begin{eqnarray*}
        \frac{dG}{d\tau_D} &=&e^{-(r+\lambda)\tau_D}\frac{\left(\nu_0\lambda \right)\left( 1 - (1-\nu_0)e^{-r\tau_D} -  \nu_0e^{-(r+\lambda)\tau_D}) \right) }{\left(1 - e^{-r\tau_D}(1 - \nu_0 + \nu_0e^{-\lambda\tau_D})\right)^2}, \\
        && - \, e^{-(r+\lambda)\tau_D}\frac{\left(-c + \nu_0 \left(\frac{\lambda}{r+\lambda}\right)\left(1- e^{-(r+\lambda) \tau_D}\right)\right)\left( (1-\nu_0)re^{\lambda\tau_D} + \nu_0(r+\lambda) \right)}{\left(1 - e^{-r\tau_D}(1 - \nu_0 + \nu_0e^{-\lambda\tau_D}) \right)^2}.
    \end{eqnarray*}   
    Distributing the numerators,
    \begin{eqnarray*}
        \frac{dG}{d\tau_D} &=& e^{-(r+\lambda)\tau_D}\frac{  \nu_0\lambda - (\nu_0\lambda)(1-\nu_0)e^{-r\tau_D} -  \nu_0(\nu_0\lambda)e^{-(r+\lambda)\tau_D})  }{\left(1 - e^{-r\tau_D}(1 - \nu_0 + \nu_0e^{-\lambda\tau_D}) \right)^2} ,\\
        && - \, e^{-(r+\lambda)\tau_D}\frac{\left(\frac{\nu_0(1-\nu_0)r \lambda}{r+\lambda}\right)\left(e^{\lambda\tau_D}- e^{-r \tau_D}\right)  + \nu_0(\nu_0\lambda)\left(1- e^{-(r+\lambda) \tau_D}\right)}{\left(1 - e^{-r\tau_D}(1 - \nu_0 + \nu_0e^{-\lambda\tau_D}) \right)^2}, \\
        && - \, e^{-(r+\lambda)\tau_D}\frac{-c\left( (1-\nu_0)re^{\lambda\tau_D} + \nu_0(r+\lambda) \right)}{\left(1 - e^{-r\tau_D}(1 - \nu_0 + \nu_0e^{-\lambda\tau_D}) \right)^2}.
    \end{eqnarray*}  
    Simplifying the numerators,
    \begin{eqnarray*}
        \frac{dG}{d\tau_D} &=& e^{-(r+\lambda)\tau_D}\frac{  \nu_0\lambda (1-\nu_0)\left( 1 - e^{-r\tau_D} - \frac{r}{r+\lambda}\left(e^{\lambda\tau_D}- e^{-r \tau_D}\right) \right) +c\left( (1-\nu_0)re^{\lambda\tau_D} + \nu_0(r+\lambda) \right)}{\left(1 - e^{-r\tau_D}(1 - \nu_0 + \nu_0e^{-\lambda\tau_D}) \right)^2} .
    \end{eqnarray*}  
    Factoring out common terms, 
    \begin{eqnarray*}
        \frac{dG}{d\tau_D} &=&\frac{ \nu_0\lambda (1-\nu_0) e^{-(r+\lambda)\tau_D} }{\left(1 - e^{-r\tau_D}(1 - \nu_0 + \nu_0e^{-\lambda\tau_D}) \right)^2}\left[ \left(1+   \frac{c(r+\lambda)}{\lambda(1-\nu_0)} \right) - \frac{\lambda}{r+\lambda}e^{-r\tau_D} - \left(\frac{r}{r+\lambda} - \frac{cr}{\nu_0\lambda}\right)e^{\lambda\tau_D} \right].
    \end{eqnarray*}  
    The bracketed part single-crosses 0 from above, so the maximization objective is quasiconcave and the first-order condition is necessary and sufficient. To see this, note that the derivative of the bracketed term is
    \[\frac{\lambda r}{r+\lambda} (e^{-r\tau_D} - e^{\lambda \tau_D}) + \frac{cr}{\nu_0} e^{\lambda \tau_D}  = e^{\lambda \tau_D}\left[\frac{\lambda r}{\lambda+ r}(e^{-(r+\lambda)\tau_D} - 1) + \frac{cr}{\nu_0} \right]. \]
    So if the bracketed term is decreasing for $\tau_D = \tau_1$, it must also be decreasing at any $\tau_D = \tau_2 > \tau_1$. When $\tau_D = 0$, the bracketed term equal to 
    \[ 1 + \frac{c(r+\lambda)}{\lambda(1-\nu_0)} - \frac{\lambda}{r+\lambda} - \frac{r}{\lambda+r} + \frac{cr}{\nu_0\lambda} = \frac{c(r+\lambda)}{\lambda(1-\nu_0)} + \frac{cr}{\nu_0\lambda}, \]
    which is positive, and as $\tau_D$ increases, the expression diverges to $-\infty$. Therefore, $dG/d\tau_D$ crosses zero from above once, exactly when the condition \eqref{defn:taustar_appendix} holds, 
    which is when $G$ attains its maximum.
\end{proof}

Note that the index in state $D$, $g(D,t)$ is strictly decreasing in $t$ (since $\nu_t$ decreases in $t$) and so to find the stopping time $t$ when it becomes optimal to switch to a new arm, equate 
\( g(U) = g(D,t) \).
\begin{lemma}\label{lem:gittins_2}
    The time $t$ such that $g(U) = g(D,t)$ is precisely $K^*$, the solution to \eqref{defn:taustar}.
\end{lemma}
\begin{proof}[Proof of Lemma \ref{lem:gittins_2}]
First, as a useful fact, the condition for $K^*$ can be rewritten from \eqref{defn:taustar_appendix} as
\begin{equation}
         \frac{ \nu_0\lambda e^{-(r+\lambda)K^*}\left(1 - (1-\nu_0)e^{-rK^*} - \nu_0e^{-(r+\lambda)K^*}) \right)}{re^{-rK^*} (1-\nu_0) + \nu_0(r+\lambda)e^{-(r+\lambda)K^*} } =   -c + \nu_0 \left(\frac{\lambda}{r+\lambda}\right)\left(1- e^{-(r+\lambda) K^*}\right). \label{eqn:useful_foc}
    \end{equation}   
Setting $g(D,t) = g(U)$,  
   \[ \frac{r\lambda \nu_t}{r+\lambda \nu_t} =  \frac{-rc + r\nu_0 \left(\frac{\lambda}{r+\lambda}\right)\left(1- e^{-(r+\lambda) K^*}\right)}{1 - e^{-rK^*}(1 - \nu_0 + \nu_0e^{-\lambda K^*})} \]
   To solve for $\nu_t$, divide out $r$, take the reciprocal, and use the identity \eqref{eqn:useful_foc}
   \[ \frac{r}{\lambda \nu_t} + 1 =  \frac{\left(1 - e^{-rK^*}(1 - \nu_0 + \nu_0e^{-\lambda K^*})\right)(re^{-rK^*} (1-\nu_0) + \nu_0(r+\lambda)e^{-(r+\lambda)K^*} )}{ \nu_0\lambda e^{-(r+\lambda)K^*}\left(1 - (1-\nu_0)e^{-rK^*} - \nu_0e^{-(r+\lambda)K^*}) \right)} \]
   \[ \frac{r}{\lambda \nu_t} + 1 =  \frac{re^{-rK^*} (1-\nu_0) + \nu_0(r+\lambda)e^{-(r+\lambda)K^*} }{ \nu_0\lambda e^{-(r+\lambda)K^*}} = \frac{r (1-\nu_0) + \nu_0(r+\lambda)e^{-\lambda K^*} }{ \nu_0\lambda e^{-\lambda K^*}}  \]
  Subtracting $1$ from both sides,
   \[ \frac{r}{\lambda \nu_t}  = \frac{r (1-\nu_0) + \nu_0(r+\lambda)e^{-\lambda K^*} }{ \nu_0 \lambda e^{-\lambda K^*}} - 1 \]
   \[ \frac{r}{\lambda \nu_t}  = \frac{r (1-\nu_0) + \nu_0(r+\lambda )e^{-\lambda K^*} }{ \nu_0  \lambda e^{-\lambda K^*}} - \frac{ \nu_0 \lambda e^{-\lambda K^*}}{ \nu_0 \lambda e^{-\lambda K^*} } =  \frac{r (1-\nu_0) + \nu_0 re^{-\lambda K^*} }{ \nu_0  \lambda e^{-\lambda K^*}}  \]
   Simplifying, 
   \[ \nu_t = \frac{\nu_0 e^{-\lambda K^*} }{(1-\nu_0) + \nu_0e^{-\lambda K^*}}, \]
   and so $t = K^*$.
\end{proof}

Finally, note that if $\min_{i \in [N]} K_i > K^*$, then the Gittins index for each arm is less than $g(U)$, so the optimal strategy must pull arm 0. In the case of equality $\min_{i \in [N]} K_i = K^*$, note that due to right-continuity, if the strategy pulls an arm $i$ that is not arm 0 at time $t$, then there exists an $\epsilon > 0$ such that the strategy assigns positive effort to arm $i$ on $[t, t+\epsilon)$; however, at time $t+\epsilon/2$, it must be the case that $p_{t+\epsilon/2} < \nu_t$, and so $g(D, t+\epsilon/2) < g(D,t) \le g(U)$, and so I have a contradiction of the index theorem. To finish the proof, it remains to show that effort must be equally split between all arms in $B$ when $\min_{i\in [N]} K_i$. By the index theorem, if any arm $j$ has $K_j > \min_{i \in [N]} K_i$, it must be that $g(D,K_j) < g(D, K_i)$, and hence the optimal strategy cannot pull arm $j$. To show that the effort must be split equally, suppose for sake of contradiction that the arm assigns effort larger than $1/|B|$ to arm $i \in B$ and effort less than $1/|B|$ to some arm $j \in B$. By right-continuity, there exists some $\epsilon$ such that the optimal strategy assigns positive effort to both $i$ and $j$ in $[t, t+\epsilon)$. Further, within this interval, since the effort assigned to $i$ was larger than the effort assigned to $j$, there must be some $\delta < \epsilon$ such that the cumulative effort $K_j(t+\delta) < K_i(t+\delta)$. But this implies that at time $t+\delta$, $g(D,K_j(t+\delta)) > g(D,K_i(t+\delta))$ but the optimal strategy exerted positive effort on arms $j$ and $i$, a violation of the index theorem. Hence, the effort must be equalized over arms in $B$.
\end{proof}

\begin{proof}[Proof of Proposition \ref{prop:baseline_comp}]
    Rearranging \eqref{defn:taustar} to apply the implicit function theorem,
    \[ 0 = 1+   \frac{c(r+\lambda)}{\lambda(1-\nu_0)} - \frac{\lambda}{r+\lambda}e^{-rK} -  \left(\frac{r}{r+\lambda} - \frac{cr}{\nu_0\lambda}\right)e^{\lambda K}  := \varphi(\lambda, c, K). \]
    Note that
    \begin{align*}
        \frac{\partial \varphi}{\partial K} &= \frac{\lambda r}{r+\lambda} e^{-rK} - \left(\frac{\lambda r}{r+\lambda} - \frac{cr}{\nu_0} \right) e^{\lambda K} = e^{\lambda K} \left[ \frac{\lambda r}{r+\lambda} (e^{-(r+\lambda)K} - 1) + \frac{cr}{\nu_0}  \right].
    \end{align*}
    Since the bracketed term is decreasing and $e^{\lambda K} > 0$, the derivative $\partial \varphi/\partial K$ single-crosses zero, and so $\varphi$ is quasiconcave in $K$. Hence, since $\varphi$ is zero at $K^*$, this implies that $\varphi$ is decreasing in $K$ at $K^*$. 
    By the implicit function theorem,
    \begin{align*}
         \frac{\dd K^*}{\dd c} &= - \frac{\partial \varphi / \partial c|_{K^*}}{\partial \varphi/\partial K |_{K^*}}, & \frac{\dd K^*}{\dd \lambda} &= - \frac{\partial \varphi / \partial \lambda|_{K^*}}{\partial \varphi/\partial K |_{K^*}}.
    \end{align*}
    Note that since $\varphi$ decreases in $K$ at $K^*$, the sign of each of the derivatives is exactly the same as the sign of the respective partial derivative on $\varphi$ evaluated at $K^*$. 
    For $c$,
    \begin{align*}
        \frac{\partial \varphi}{\partial c}\Bigr|_{K^*} &=\frac{r+\lambda}{\lambda(1-\nu_0)} + \frac{r}{\nu_0\lambda}e^{\lambda K^*} > 0.
    \end{align*}
    so $K^*$ increases in $c$. For $\lambda$, 
    \begin{align}
        \frac{\partial \varphi}{\partial \lambda}\Bigr|_{K^*} &= - \frac{cr}{\lambda^2(1-\nu_0)} - \frac{r}{(r+\lambda)^2}e^{-rK^*} + \left( \frac{r}{(r+\lambda)^2} - \frac{cr}{\lambda^2\nu_0} \right)e^{\lambda K^*} - \left( \frac{ r}{r+\lambda} - \frac{cr}{\nu_0 \lambda } \right)K^*e^{\lambda K^*} \notag\\
        &= - \frac{cr}{\lambda^2(1-\nu_0)} - \frac{r}{(r+\lambda)^2}e^{-rK^*} + \left( \frac{r}{(r+\lambda)^2} - \frac{cr}{\lambda^2\nu_0} \right)e^{\lambda K^*} - \left( \frac{ r}{\lambda(r+\lambda)} - \frac{cr}{\nu_0 \lambda^2 } \right)\lambda K^*e^{\lambda K^*}\notag\\
        &= - \frac{cr}{\lambda^2(1-\nu_0)} - \frac{r}{(r+\lambda)^2}e^{-rK^*} - \frac{r^2}{\lambda (r+\lambda)^2}e^{\lambda K^*} +\left( \frac{ r}{\lambda(r+\lambda)} - \frac{cr}{\nu_0 \lambda^2 } \right)(1 - \lambda K^*)e^{\lambda K^*} \notag\\
        &< - \frac{cr}{\lambda^2(1-\nu_0)} - \frac{r}{(r+\lambda)^2}e^{-rK^*} - \frac{r^2}{\lambda (r+\lambda)^2}e^{\lambda K^*} +\left( \frac{ r}{\lambda(r+\lambda)} - \frac{cr}{\nu_0 \lambda^2 } \right) \notag \\
        &= - \frac{cr}{\lambda^2(1-\nu_0)} - \frac{cr}{\nu_0 \lambda^2 } +\frac{ r}{\lambda(r+\lambda)} \left( 1 - \frac{\lambda}{\lambda + r}e^{-rK^*}  - \frac{r}{r+\lambda }e^{\lambda K^*} \right)  \label{eqn:partial_deriV_Hambda}
    \end{align}
    where the inequality follows from the fact that $(1-z)e^z < 1$ for $z > 0$. Further, since $\varphi(c,\lambda,\nu_0,K^*) = 0$, it follows that 
    \[ 1  - \frac{\lambda}{r+\lambda}e^{-rK} - \frac{r}{r+\lambda}e^{\lambda K}  = - \frac{cr}{\nu_0 \lambda }e^{\lambda K} - \frac{c(r+\lambda)}{\lambda(1-\nu_0)} < 0. \]
    Combining this observation with \eqref{eqn:partial_deriV_Hambda} implies that $\partial \varphi/\partial\lambda|_{K^*} < 0$, and so $K^*$ must be decreasing in $\lambda$.
\end{proof}

\subsubsection{Proof of Theorem \ref{thm:learning}}\label{sec:learning_proof_sketch}
I first introduce some preliminary objects and notation to handle the complexity of the state space, and to switch to working with the state and action paths through time. Observe that any strategy $\sigma$ induces a deterministic state path. More precisely, let the continuation \textit{state path} of $\sigma$ be the function $\xi^\sigma: \mathbb{R}_+ \to \Sigma$ constructed as 
\[ \xi^\sigma(0) = s, \quad \dot{\xi}^\sigma_k(t) = \sigma_k(\xi^\sigma(t)) \textnormal{ if }\sigma(\xi^\sigma(t)) \neq \textnormal{brainstorm},  \]
\[ \xi^\sigma(t) = \left(\lim_{t' \to t^-}\xi^\sigma(t') , \ 0 \right) \textnormal{ if } \sigma\left(\lim_{t' \to t^-}\xi^\sigma(t')\right) = \textnormal{brainstorm}.  \]
Intuitively, the first two expressions imply that when the policy $\sigma$ does not brainstorm new approaches, $\xi^\sigma(\cdot)$ is the state trajectory that would occur following strategy $\sigma$ starting from state $s$. The second part implies that when the path approaches a state where the strategy $\sigma$ chooses to brainstorm, the path jumps to the state with a new approach. 

Starting from state $s$, any strategy $\sigma$ induces a (possibly finite or empty) sequence of times where the strategy brainstorms new approaches $\{t^\sigma_i\}$, defined as the sequence of times $t$ such that $\sigma\left(\lim_{t' \to t^-}\xi^\sigma(t')\right) = \textnormal{brainstorm}.$ In the constructions that follow, I will fix the brainstorming times.\footnote{This does \textit{not} imply that strategies $\sigma$ and $\sigma'$ brainstorm at the same states; this is because the state paths under $\sigma$ and $\sigma'$ differ.}
Let $N^{\sigma,s}_t$ denote the number of approaches available at time $t$ given strategy $\sigma$ starting from state $s$. 

Similarly, define the \textit{action path} induced by $\sigma$ as $k^\sigma_n(t) := \sigma_n(\xi^\sigma(t))$. Conversely, given any c\`adl\`ag $k:\mathbb{R}^\infty \to [0,1]^\infty$ that satisfies (i) $k_n(t) = 0$ whenever $n > N^{\sigma,s}_t$ (i.e. never exerts effort on an approach before it is brainstormed by $\sigma$) and (ii) $\sum_{n=1}^\infty k_n(t) = 1$ for all $t$, the state path is constructed as
\[ \xi^k_n(t) = s_n + \int_0^t k_n(t) \dd t \quad \forall n, \]
where $s_n = 0$ for $n > N_0^{\sigma,s}$, and strategy $\sigma^k$ defined as 
\[ \sigma^k(\xi^k(t)) = k(t) := (k_1(t), k_2(t),... ), \quad \sigma^k\left(\lim_{t' \to t^\sigma_i} \xi^k(t) \right) = \textnormal{brainstorm} \quad \forall t^\sigma_i, \]
induces $k$ as the action path and brainstorms new approaches at exactly the same times as $\sigma$.\footnote{ Technically, the strategy $\sigma^k$ is underspecified, since a strategy should specify actions at all states; the action path only specifies actions at states that are on-path following $k$ starting from state $s$. In particular, it is possible for two different strategies to have the same action path $k$ and arrival times of new approaches, because they can differ off the state path. However, this underspecification is payoff-irrelevant, since it only changes behavior off of any realized state path.} Because of this, I work directly with action paths (functions of one-dimensional time) instead of strategies (functions of high-dimensional state).

For Lemma \ref{lem:lowest_history_of_effort}, I will therefore construct interchange strategies by specifying the action paths and fixing the brainstorming times; since the brainstorming times are fixed, I focus on the expected breakthrough time, and show that the interchange results in an earlier success in a first-order stochastic dominance sense.  Let \[ V_\theta(k | s) := \mathbb{E}^k\left[e^{-r \tau}  - \sum_{k=1}^\infty e^{-r\tau_k} c \mathbb{I}[\tau > \tau_i] \mid  s, \lambda_\theta\right]\] denote the conditional value of the action path $k$ starting from $s$, given $\theta$.\footnote{ The conditional expectation notation for $s$ is conditioning on the event that $s$ realizes, which conveys information about the $\omega_i$'s.}

Recall that $\nu_\theta(K)$ denotes the probability an arm is good given historical effort $K$ conditional on $\theta$, and $S_\theta(K)$ denotes the survival probability of an arm to historical effort $K$ given $\theta$.
Note that conditional on $\theta$, all the approaches become conditionally independent; hence, one can show the following Lemma, which characterizes the expected breakthrough time. 

\begin{lemma}
    \label{lem:exp_bkth_cdf}
    Given fixed brainstorming times and conditioning on $s, \lambda_\theta$, the action path $k$ induces breakthrough time with the following cumulative distribution function: 
    \begin{equation} F(t;s,\theta,k) = 1 - \prod_{n =1}^\infty \frac{S_\theta( \xi_n^k(t))}{S_\theta(s_n)}. \label{eqn:bk_cdf}\end{equation}
\end{lemma}
\begin{proof}
Let $\tau^k_\theta$ denote the stochastic breakthrough arrival time given path $k$, conditional on $s$ and $\theta$. 
The CDF of the $\tau^k_\theta$ for given realization of $\{ \omega_n \}$ is 
\[ F(t\mid s,\theta,k,\{\omega_n\}) = 1 - \exp\left( - \lambda_\theta \sum_{m=1}^\infty \omega_m (\xi^k_m(t) - s_m)  \right) . \]
Taking the expectation over $\{\omega_n\}$, 
\begin{align*} F(t\mid s,\theta,k) &= 1 - \mathbb{E}\left[\left.\exp\left( - \lambda_\theta \sum_{m=1}^\infty \omega_m (\xi^k_m(t) - s_m)  \right) \right| s, \lambda_\theta \right], \\
&= 1 - \prod_{m=1}^\infty \mathbb{E}\left[\left.\exp\left( - \lambda_\theta \omega_m (\xi^k_m(t) - s_m)  \right) \right| s, \lambda_\theta \right].
\end{align*}
where the last line follows because $\omega_m$ and $\omega_n$ are conditionally independent given $\lambda_\theta$, since $\theta$ is the only source of correlation across arms. The expectation
\begin{align*}
    \mathbb{E}\left[\left. e^{-\lambda_\theta \omega_n (\xi^k_n(t) - s_n)}\right|s,\lambda_\theta \right]  &= \frac{(1-\nu_0)1 + \nu_0e^{-\lambda_\theta s_n} e^{-\lambda_\theta (\xi_n^k(t)- s_n)}}{\nu_0 e^{-\lambda_\theta s_n} + 1-\nu_0} = \frac{1-\nu_0 + \nu_0e^{-\lambda_\theta \xi_n^k(t)}}{\nu_0 e^{-\lambda_\theta s_n} + 1-\nu_0}.
\end{align*}
Plugging this into the CDF expression yields the result.
\end{proof}

The value of an action path $k$ is then given by the law of total expectation: 
\[ V(k|s) = \sum_{\theta=H,L} \mathbb{P}[\theta| s] V_\theta(k|s). \]
Now, I establish the first main Lemma.

\begin{proof}[Proof of Lemma \ref{lem:lowest_history_of_effort}] 
Consider the continuation problem at state $s$, where $N$ approaches are available, and renormalize time to zero at $s$. Consider a strategy $\sigma$, where $\sigma_i(s) > 0$ for some $i \not\in B(s)$. I will argue that $\sigma$ cannot be optimal via an interchange argument; that is, I will construct a strategy which improves the agent's objective by briefly swapping effort between $i$ and some $j \in B(s)$.  Fix a $j \in B(s)$. Let $t'$ be the next time (possibly also $t$) that the action path $k^\sigma$ induced by $\sigma$ exerts effort on $j$:
\[ t' := \inf \left \{ \hat{t} \mid \hat{t} \ge t, \, k^\sigma_j(\hat{t}) > 0 \right\}, \]
where $t' = \infty$ if for all $t' \ge t$, $k^\sigma_j(t') = 0$. Consider two cases: $t' = \infty$ and $ t' < \infty$.

\textbf{Case 1:} Suppose $t' = \infty$; this implies that $k^\sigma_j(t) = 0$ for all $t$. I claim that the agent can improve their objective by swapping effort from $i$ to $j$ instead. Formally, consider the alternative $k$, constructed as follows:
\begin{align}
    k_i(t) &= k^\sigma_j(t) = 0, & k_j(t) &= k^\sigma_i(t), & k_n(t) &= k^\sigma_n(t) \quad \forall n \not\in \{i, j\}. \label{eqn:lem1_interchange_1}
\end{align}
That is, $k$ relative to $k^\sigma$ shifts effort from $i$ to $j$ instead. Let $K^\sigma(t) := \xi^\sigma_i(t) - s_i$ be the total effort exerted by strategy $\sigma$ on $i$ at time $t$ after state $s$; note by path right-continuity and the assumption that $\sigma_i(s) > 0$, $K_i^\sigma(t) > 0$ for any $t > 0$. Then by Lemma \ref{lem:exp_bkth_cdf},
\begin{align*} \frac{1 - F(t;s,\theta,k)}{1-F(t;s,\theta,k^\sigma)} &= \prod_{n=1}^\infty \frac{ S_\theta(\xi_n^k(t))}{S_\theta(\xi^\sigma_n(t))} = \left(\frac{S_\theta(\xi^k_i(t))}{S_\theta(\xi^\sigma_i(t))} \right)\left(\frac{S_\theta(\xi^k_j(t))}{S_\theta(\xi^\sigma_j(t))} \right), \\
&= \left(\frac{S_\theta( s_i) }{S_\theta (K^\sigma(t) + s_i)} \right)\left(\frac{S_\theta(K^\sigma(t) + s_j)}{S_\theta (s_j) } \right), \\
&= \frac{(1-\nu_0)^2 + \nu_0(1-\nu_0)(e^{-\lambda_\theta s_i} + e^{-\lambda_\theta(K^\sigma(t) + s_j)}) + \nu_0^2e^{-\lambda_\theta(K^\sigma(t) + s_j + s_i)}}{(1-\nu_0)^2 + \nu_0(1-\nu_0)(e^{-\lambda_\theta s_j} + e^{-\lambda_\theta(K^\sigma(t) + s_i) } )+ \nu_0^2e^{-\lambda_\theta(K^\sigma(t) + s_j + s_i)}}.
\end{align*}
where the second line used the fact that $\xi_i^k(t) = s_i$ and $\xi_j^\sigma(t) = s_j$. Note that since $s_j < s_i$ by supposition $j \in B(s)$ and $i \not\in B(s)$, so for any $t > 0$,
\[ e^{-\lambda_\theta s_i} + e^{-\lambda_\theta(K^\sigma(t) + s_j)} < e^{-\lambda_\theta s_j} + e^{-\lambda_\theta(K^\sigma(t) + s_i) }. \] 
This is easy to see because $e^{-\lambda x}$ is a convex function, and $s_i + K^\sigma(t) > s_i > s_j$.

Therefore, the distribution of the breakthrough time under $k$ is first-order stochastic dominated by the distribution of the breakthrough time under $k^\sigma$, which implies that $\mathbb{E}[e^{-r\tau}\mid s, \lambda_\theta]$ is larger under $k$ than under $k^\sigma$, and so $V_\theta(k) > V_\theta(k^\sigma)$. Since this holds for both $\theta = H, E$ and $\mathbb{P}[\theta | s]$ is unaffected by $k$, the strategy $\sigma$ cannot be optimal.

\textbf{Case 2:} Suppose $t'$ is finite; then there exists a time $\tilde{t}$ such that $k^\sigma_j(\tilde{t}) > 0$. Take any $\tilde{t}$ where $k^\sigma_j(\tilde{t}) > 0$ and  $\xi^\sigma_j(\tilde{t}) < s_i$ (which must exist since $\xi^\sigma_j(0) = s_j < s_i$). I will construct an interchanged action path by shifting some effort from $i$ to $j$ at time $0$, and then shifting the same amount of effort from $j$ back to $i$ at time $\tilde{t}$. Consider an incremental amount of effort $\tilde{K} > 0$, and note that for sufficiently small $\tilde{K}$, there exists $\epsilon, \tilde{\epsilon}$ such that
\[ \tilde{K} = \xi^\sigma_j(\tilde{t} + \tilde{\epsilon}) - \xi^\sigma_j(\tilde{t}) = \xi^\sigma_i(\epsilon) - \xi^\sigma_i(0) . \]
Intuitively, $\epsilon'$ is chosen such that $k^\sigma$ exerts the same total (positive) effort on $j$ in the window $[\tilde{t}, \tilde{t} + \tilde{\epsilon})$ as it does on $i$ in the window $[0,\epsilon)$. 
Choose a total interchange effort $\tilde{K}$ small enough such that $\tilde{K} < s_i - \xi_j^\sigma(\tilde{t}) $
and 
\[ \xi^\sigma_i(\epsilon)  + \xi^\sigma_j(\epsilon) < 2 s_i, \]
which can be done because when $t \to 0$, $\xi_j^\sigma(t) \to s_j < s_i$ and $\xi_i^\sigma(t) \to s_i$. 

Now, consider the following interchanged action path $k$, swapping  this effort:
\begin{align}
    k_i(t) &= \begin{cases} 0  & t < \epsilon \\
    k_i^\sigma(t) +k_j^\sigma(t) & t \in [\tilde{t}, \tilde{t} + \tilde{\epsilon}) \\
    k_i^\sigma(t) & t \not\in [0,  \epsilon) \cup [\tilde{t}, \tilde{t} + \tilde{\epsilon}) \\
    \end{cases} , \notag \\
    k_j(t) &= \begin{cases} k_j^\sigma(t) + k_i^\sigma(t)  & t < \epsilon \\
    0 & t \in [\tilde{t}, \tilde{t} + \tilde{\epsilon}) \\
    k_j^\sigma(t) & t \not\in [0, \epsilon) \cup [\tilde{t}, \tilde{t} + \tilde{\epsilon}) \\
    \end{cases} \notag \\
    k_n(t) &= k^\sigma_n(t) \quad \forall n \not\in \{i, j\}. \label{eqn:lem1_interchange_2}
\end{align}
Intuitively, $k$ exerts effort on $j$ instead of $i$ in the window $[0, \epsilon)$, exerts effort on $i$ instead of $j$ in the window $[\tilde{t}, \tilde{t} + \epsilon)$, and otherwise behaves identically to $k^\sigma$. Then note that by construction, the state paths $\xi^k$ are identical to $\xi^\sigma$ on $n \neq i, j$, and for all times above $\tilde{t} + \tilde{\epsilon}$. Hence, $F(t;s,\theta,k) = F(t;s,\theta,k^\sigma)$ for $t \ge \tilde{t} + \tilde{\epsilon}$. 

I claim that the interchanged strategy $k$ generates a breakthrough time that is first-order stochastically dominated by the breakthrough time generated by $k^\sigma$; equivalently, $F(t;s,\theta,k) \ge F(t;s,\theta,k^\sigma)$ for all $t$. By the previous paragraph, this clearly holds for $t \ge \tilde{t} + \tilde{\epsilon}$.

For times $t < \epsilon$, 
\begin{align*} \frac{1 - F(t;s,\theta,k)}{1-F(t;s,\theta,k^\sigma)} 
&= \left(\frac{S_\theta( s_i)}{S_\theta(\xi^\sigma_i(t))} \right)\left(\frac{S_\theta(\xi^\sigma_j(t) + \xi^\sigma_i(t) - s_i)}{S_\theta(\xi^\sigma_j(t))} \right), \\
&= \frac{(1-\nu_0)^2 + \nu_0(1-\nu_0)(e^{-\lambda_\theta s_i} + e^{-\lambda_\theta(\xi^\sigma_j(t) + \xi^\sigma_i(t) - s_i)}) + \nu_0^2e^{-\lambda_\theta(\xi_j^\sigma(t) + \xi_i^\sigma(t))}}{(1-\nu_0)^2 + \nu_0(1-\nu_0)(e^{-\lambda_\theta \xi_i^\sigma(t) } + e^{-\lambda_\theta\xi_j^\sigma(t) } )+ \nu_0^2e^{-\lambda_\theta(\xi_j^\sigma(t) + \xi_i^\sigma(t))}}.
\end{align*}
Note that once again, because by construction $\xi_j^\sigma(t) + \xi_i^\sigma(t) - s_i < s_i$, 
\[ e^{-\lambda_\theta s_i} + e^{-\lambda_\theta(\xi^\sigma_j(t) + \xi^\sigma_i(t) - s_i)} < e^{-\lambda_\theta \xi_i^\sigma(t) } + e^{-\lambda_\theta\xi_j^\sigma(t) } . \]
Therefore, for $t < \epsilon$, $F(t;s,\theta,k) > F(t;s,\theta,k^\sigma)$. 

Now, for times $t \in [\epsilon, \tilde{t})$, 
\begin{align*} \frac{1 - F(t;s,\theta,k)}{1-F(t;s,\theta,k^\sigma)} 
&= \left(\frac{S_\theta (\xi_i^\sigma(t) - \tilde{K})}{S_\theta(\xi^\sigma_i(t))} \right)\left(\frac{S_\theta(\xi^\sigma_j(t) + \tilde{K})}{S_\theta(\xi^\sigma_j(t))} \right), \\
&= \frac{(1-\nu_0)^2 + \nu_0(1-\nu_0)(e^{-\lambda_\theta (\xi_i^\sigma(t) -\tilde{K})} + e^{-\lambda_\theta(\xi^\sigma_j(t) + \tilde{K})}) + \nu_0^2e^{-\lambda_\theta(\xi_j^\sigma(t) + \xi_i^\sigma(t))}}{(1-\nu_0)^2 + \nu_0(1-\nu_0)(e^{-\lambda_\theta \xi_i^\sigma(t) } + e^{-\lambda_\theta\xi_j^\sigma(t) } )+ \nu_0^2e^{-\lambda_\theta(\xi_j^\sigma(t) + \xi_i^\sigma(t))}}.
\end{align*}
Since by construction, $\xi^\sigma_j(t) \le \xi^\sigma_j(\tilde{t}) < \xi^\sigma_j(\tilde{t}) + \tilde{K} < s_i \le \xi^\sigma_i(t)$, it follows that $\xi_j^\sigma(t) + \tilde{K}$ and $\xi_i^\sigma(t) - \tilde{K}$ are inside the interval $(\xi_j^\sigma(t), \xi_i^\sigma(t))$, so 
\[ e^{-\lambda_\theta (\xi_i^\sigma(t) - \tilde{K}) } + e^{-\lambda_\theta (\xi_j^\sigma(t) + \tilde{K})} < e^{-\lambda_\theta \xi_i^\sigma(t) } + e^{-\lambda_\theta\xi_j^\sigma(t) },\]
which implies that for $t \in [t+\epsilon, \tilde{t})$, $F(t;s,\theta,k) > F(t;s,\theta,k^\sigma)$. 

Finally, when $t \in [\tilde{t}, \tilde{t} + \tilde{\epsilon})$, 
\begin{align*} \frac{1 - F(t;s,\theta,k)}{1-F(t;s,\theta,k^\sigma)} &= \left(\frac{S_\theta (\xi_i^\sigma(t) - \tilde{K}+ \xi^\sigma_j(t) - \xi^\sigma_j(\tilde{t}) )}{S_\theta(\xi^\sigma_i(t))} \right)\left(\frac{S_\theta(\xi^\sigma_j(\tilde{t}) + \tilde{K})}{S_\theta(\xi^\sigma_j(t))} \right), \\
&= \frac{(1-\nu_0)^2 + \nu_0(1-\nu_0)(e^{-\lambda_\theta (\xi_i^\sigma(t) - \tilde{K}+ \xi^\sigma_j(t) - \xi^\sigma_j(\tilde{t}) )}+ e^{-\lambda_\theta(\xi^\sigma_j(\tilde{t}) + \tilde{K})}) + \nu_0^2e^{-\lambda_\theta(\xi_j^\sigma(t) + \xi_i^\sigma(t))}}{(1-\nu_0)^2 + \nu_0(1-\nu_0)(e^{-\lambda_\theta \xi_i^\sigma(t) } + e^{-\lambda_\theta\xi_j^\sigma(t) } )+ \nu_0^2e^{-\lambda_\theta(\xi_j^\sigma(t) + \xi_i^\sigma(t))}}.
\end{align*}
By construction, $\xi_j^\sigma(t) < s_i < \xi_i^\sigma(t) $ and $\xi_j^\sigma(t) < \xi_j^\sigma(\tilde{t} + \tilde{\epsilon}) = \xi^\sigma_j(\tilde{t}) + \tilde{K} < s_i$. So by convexity, 
\[ e^{-\lambda_\theta (\xi_i^\sigma(t) - \tilde{K}+ \xi^\sigma_j(t) - \xi^\sigma_j(\tilde{t}) )}+ e^{-\lambda_\theta(\xi^\sigma_j(\tilde{t}) + \tilde{K})} < e^{-\lambda_\theta \xi_i^\sigma(t) } + e^{-\lambda_\theta\xi_j^\sigma(t) },\]
And so for any $t$, $F(t;s,\theta,k) \ge F(t;s,\theta,k^\sigma)$, and therefore $\mathbb{E}[e^{-r\tau}\mid s, \lambda_\theta]$ is larger under $k$ than under $k^\sigma$, and so $V_\theta(k|s) > V_\theta(k^\sigma|s)$. Since this holds for both $\theta = H, E$, it follows that $V(k|s) > V(k^\sigma|s)$, and so $\sigma$ cannot be optimal.
\end{proof}

As a consequence of Lemma \ref{lem:lowest_history_of_effort} and right-continuity of action paths, at any state $s$ where the set of best approaches has multiple elements ($|B(s)| > 1$), the optimal policy must split effort evenly between all best approaches; i.e., exert effort $1/|B(s)|$ on each approach.

Now, since Lemma \ref{lem:lowest_history_of_effort} dictates what approaches the agent exerts effort on, it suffices to optimize when new approaches are obtained. Further, Lemma \ref{lem:lowest_history_of_effort} also implies that any sequence of brainstorming times maps to a unique sequence of effort thresholds, where $K^*_i$ denotes the effort exerted on $i$ when the next approach $i+1$ is brainstormed.

Lemma \ref{lem:increasing_thresholds} is more difficult to establish. I will do so in two smaller steps. I will first establish that if the sequence $\{ K^*_i \}$ starts decreasing, it cannot increase again. 
\begin{lemma}\label{lem:threshold_decrease_after}
    Suppose the optimal threshold sequence satisfies $K^*_i \ge K^*_{i+1}$ for some $i\ge 1$. Then $K^*_{i+1} \ge K^*_{i+2}$.
\end{lemma}
\begin{proof}
I will first sketch the proof. Proceed by contradiction; suppose $K^*_i \ge K^*_{i+1} < K^*_{i+2}$, and let $t_{i}$ denote the time the $i$th approach is brainstormed (similarly for $i+1$, $i+2$). Consider the following interchanges:
    \begin{enumerate}
        \item Generate approach $i+2$ an instant sooner, but work on both $i+1$ and $i+2$ until the same threshold $K^*_{i+2}$ before brainstorming approach $i+3$.
        \item Generate approach $i+2$ an instant later, moving ahead effort on $i+1$ from after $i+2$ reached $K^*_{i+1}$ to before $i+2$ was brainstormed.
    \end{enumerate}
See Figure \ref{fig:int} for a depiction of these interchanges. Optimality of the policy with respect to these two interchanges determines $K^*_{i+1}$ as a function of the time $t_{i+2}$ belief over $\theta$. Having pinned down $K^*_{i+1}$, consider a third interchange:
\begin{enumerate}
    \item[3.] Generate $i+1$ an instant sooner, and delaying effort on $i$ until right before $i+2$ is brainstormed.
\end{enumerate}
The third interchange is depicted in Figure \ref{fig:int_contr}. Using the value of $K^*_{i+1}$ determined by the first two interchanges, I establish that the effort on approach $i+1$ is sufficiently valuable, so the third interchange (moving effort on $i+1$ up) violates optimality.

    More precisely, suppose, for sake of contradiction, that optimally $K^*_i \ge K^*_{i+1} < K^*_{i+2}$. Let $k$ be the original (supposedly) optimal action path, and define $k'$ the action path under the first interchange as follows:
    \begin{align*}
        k'_{i+1}(t) &:= \begin{cases}
         0 & t \in [t_{i+2} - \epsilon, t_{i+2} - 2\epsilon + K^*_{i+1}), \\
         1/2 & t \in [t_{i+2} - 2\epsilon + K_{i+1}^*, t_{i+2} + K_{i+1}^*), \\
         k_{i+1}(t) & t < t_{i+2} - \epsilon \textnormal{ or } t\ge t_{i+2} + K^*_{i+1},
        \end{cases}, \\
        k'_{i+2}(t) &:= \begin{cases}
         1 & t \in [t_{i+2} - \epsilon, t_{i+2} - 2\epsilon + K^*_{i+1}), \\
         1/2 & t \in [t_{i+2} - 2\epsilon + K_{i+1}^*, t_{i+2} + K_{i+1}^*), \\
         k_{i+2}(t) & t < t_{i+2} - \epsilon \textnormal{ or } t\ge t_{i+2} + K^*_{i+1},
        \end{cases}, \\
        k'_j(t) &:= k_j(t) & j \not\in \{i+1,i+2\}.
    \end{align*}
    and similarly define $k''$ for interchange 2: 
    \begin{align*}
        k''_{i+1}(t) &:= \begin{cases}
         1 & t \in [t_{i+2}, t_{i+2} + \epsilon) , \\
         0 & t \in [t_{i+2} + \epsilon, t_{i+2} + 2\epsilon + K_{i+1}^*), \\
         k_{i+1}(t) & t < t_{i+2}  \textnormal{ or } t\ge t_{i+2} + K^*_{i+1} + 2\epsilon,
        \end{cases}, \\
        k''_{i+2}(t) &:= \begin{cases}
         0 & t \in [t_{i+2}, t_{i+2} + \epsilon) , \\
         1 & t \in [t_{i+2} + \epsilon, t_{i+2} + 2\epsilon + K_{i+1}^*), \\
         k_{i+2}(t) & t < t_{i+2}  \textnormal{ or } t\ge t_{i+2} + K^*_{i+1} + 2\epsilon,
        \end{cases}, \\
        k''_j(t) &:= k_j(t) & j \not\in \{i+1,i+2\}.
    \end{align*}
   
    Since the original action path was optimal with respect to interchange 1, the change in value of the interchange relative to the original path must have been nonpositive. Further, note that after $t_{i+2}+K^*_{i+1}$, both the original and interchanged action path reach the same state and are identical, so the continuation values cancel. So focus on the time interval from time $t_{i+2} - \epsilon$ until $t_{i+2} + K^*_{i+1}$).
    
    The probability that a breakthrough arrives in a small interval $\epsilon$ on an arm with historical effort $K$ is $\lambda_\theta \nu_\theta(K)\epsilon + O(\epsilon^2)$ where $O(\epsilon^2)$ denotes higher-order terms. Then the payoff from the original action path in the time interval from time $t_{i+2} - \epsilon$ until $t_{i+2} + K^*_{i+1}$) is
    \[ \mathbb{E}\left[ \lambda_\theta \nu_\theta(K^*_{i+1})\epsilon+e^{-r\epsilon}(1 - \lambda_\theta \nu_\theta(K^*_{i+1})\epsilon)\left(-c + \int_0^{K^*_{i+1}} e^{-rt}(\nu_0 \lambda_\theta e^{-\lambda_\theta t})  \dd{t} \right)   \mid t_{i+2}-\epsilon \right] + O(\epsilon^2). \]
    The first interchanged strategy $k'$, on the same time interval, generates payoff 
    \[ \mathbb{E}\left[ -c + \int_0^{K^*_{i+1}} e^{-rt}(\nu_0 \lambda_\theta e^{-\lambda_\theta t})  \dd{t} + e^{-rK^*_{i+1}}(1-\nu_0 + \nu_0e^{-\lambda_\theta K^*_{i+1}})\lambda_\theta \nu_\theta(K^*_{i+1})\epsilon \mid t_{i+2}-\epsilon \right] + O(\epsilon^2). \]
    Note that higher-order terms in $\epsilon$ are absorbed into the $O(\epsilon^2)$ part. By optimality of the original path, the difference must be nonnegative:
    \begin{align*} \mathbb{E} \left[ \begin{pmatrix}
       \lambda_\theta \nu_\theta(K^*_{i+1})\epsilon (1 - e^{-rK^*_{i+1}}(1-\nu_0 + \nu_0e^{-\lambda_\theta K^*_{i+1}}))  \\
       +(e^{-r\epsilon}(1 - \lambda_\theta \nu_\theta(K^*_{i+1})\epsilon) - 1)\left(-c + \int_0^{K^*_{i+1}} e^{-rt}(\nu_0 \lambda_\theta e^{-\lambda_\theta t})  \dd{t} \right) 
    \end{pmatrix} \mid t_{i+2}-\epsilon \right] + O(\epsilon^2)\ge 0. \end{align*}
    Using that $e^{-r\epsilon} = 1 - r\epsilon + O(\epsilon^2)$ for small $\epsilon$, dividing through by $\epsilon$ and taking out the higher order terms, this implies that  
    \begin{align*} \mathbb{E} \left[ \begin{pmatrix}
       \lambda_\theta \nu_\theta(K^*_{i+1}) (1 - e^{-rK^*_{i+1}}(1-\nu_0 + \nu_0e^{-\lambda_\theta K^*_{i+1}}))  \\
       -( r + \lambda_\theta \nu_\theta(K^*_{i+1}))\left(-c + \int_0^{K^*_{i+1}} e^{-rt}(\nu_0 \lambda_\theta e^{-\lambda_\theta t})  \dd{t} \right)    
    \end{pmatrix} \mid t_{i+2}-\epsilon \right] + \Omega(\epsilon)\ge 0. \end{align*}
    In the limit then as $\epsilon \to 0$, this implies that 
    \begin{align} \mathbb{E} \left[ \begin{pmatrix}
       \lambda_\theta \nu_\theta(K^*_{i+1}) (1 - e^{-rK^*_{i+1}}(1-\nu_0 + \nu_0e^{-\lambda_\theta K^*_{i+1}}))  \\
       -( r + \lambda_\theta \nu_\theta(K^*_{i+1}))\left(-c + \int_0^{K^*_{i+1}} e^{-rt}(\nu_0 \lambda_\theta e^{-\lambda_\theta t})  \dd{t} \right)   
    \end{pmatrix} \mid t_{i+2} \right]\ge 0. 
    \label{eqn:int1}
    \end{align}
    By the similar argument, for the second interchange, the relevant time interval is $[t_{i+2} + \epsilon, t_{i+2} + K^*_{i+1} + 2\epsilon]$. The original strategy yields payoff (to first order in $\epsilon$)
    \[ \mathbb{E}\left[ -c + \int_0^{K^*_{i+1}} e^{-rt}(\nu_0 \lambda_\theta e^{-\lambda_\theta t})  \dd{t} + e^{-rK^*_{i+1}}(1-\nu_0 + \nu_0e^{-\lambda_\theta K^*_{i+1}})\lambda_\theta \nu_\theta(K^*_{i+1})\epsilon \mid t_{i+2} \right] + O(\epsilon^2). \]
    and interchange $k''$ yields
    \[ \mathbb{E}\left[ \lambda_\theta \nu_\theta(K^*_{i+1})\epsilon+e^{-r\epsilon}(1 - \lambda_\theta \nu_\theta(K^*_{i+1})\epsilon)\left(-c + \int_0^{K^*_{i+1}} e^{-rt}(\nu_0 \lambda_\theta e^{-\lambda_\theta t})  \dd{t} \right)   \dd{t} \mid t_{i+2}-\epsilon \right] + O(\epsilon^2). \]
    Taking the difference and $\epsilon \to 0$, this implies that  
    \begin{align} \mathbb{E} \left[ \begin{pmatrix}
       \lambda_\theta \nu_\theta(K^*_{i+1}) (1 - e^{-rK^*_{i+1}}(1-\nu_0 + \nu_0e^{-\lambda_\theta K^*_{i+1}}))  \\
       -( r + \lambda_\theta \nu_\theta(K^*_{i+1}))\left(-c + \int_0^{K^*_{i+1}} e^{-rt}(\nu_0 \lambda_\theta e^{-\lambda_\theta t})  \dd{t} \right)   
    \end{pmatrix} \mid t_{i+2} \right]\le 0. 
    \label{eqn:int2}
    \end{align}
    The two interchanges \eqref{eqn:int1} and \eqref{eqn:int2}, thus jointly imply that 
    \[ \mathbb{E} \left[ \begin{pmatrix}
       \lambda_\theta \nu_\theta(K^*_{i+1}) (1 - e^{-rK^*_{i+1}}S_\theta(K_{i+1}^*))  \\
       -( r + \lambda_\theta \nu_\theta(K^*_{i+1}))\left(-c + \int_0^{K^*_{i+1}} e^{-rt}(\nu_0 \lambda_\theta e^{-\lambda_\theta t})  \dd{t} \right)  
    \end{pmatrix} \mid t_{i+2} \right] = 0.  \]
    Using \eqref{eqn:phi}, the optimality condition implies 
    \[ \mathbb{E}[\phi_\theta(K^*_{i+1}) \mid t_{i+2}] = 0. \]
    \begin{lemma}\label{lem:phi_decrease}
        The function $\phi$ defined in \eqref{eqn:phi} is decreasing.
    \end{lemma}
    \begin{proof}
    Taking the derivative,
    \begin{align} \frac{\dd}{\dd K }\phi_\theta(K ) =& \lambda_\theta \nu_\theta'(K)\left[ 1 + c - \frac{\nu_0\lambda_\theta}{\lambda_\theta + r}(1 - e^{-(r+\lambda_\theta) K}) - e^{-rK}S_\theta(K) \right]  \notag \\ 
    &- (r + \lambda_\theta \nu_\theta(K)) \nu_0\lambda_\theta e^{-(r+\lambda_\theta)K} + \left[ re^{-rK}(1-\nu_0 + \nu_0e^{-\lambda_\theta K}) +  e^{-rK}\lambda_\theta \nu_0e^{-\lambda_\theta K} \right]\lambda_\theta \nu_\theta(K)  \notag \\
    =& \lambda_\theta \nu_\theta'(K)\left[ 1 + c - \frac{\nu_0\lambda_\theta}{\lambda_\theta + r}(1 - e^{-(r+\lambda_\theta) K}) - e^{-rK}S_\theta(K) \right]  \notag \\
    &- (r + \lambda_\theta \nu_\theta(K)) \nu_0\lambda_\theta e^{-(r+\lambda_\theta)K} +  re^{-rK} \lambda_\theta \nu_0 e^{-\lambda_\theta K} +  e^{-rK}\lambda_\theta \nu_0e^{-\lambda_\theta K} \lambda_\theta \nu_\theta(K) \notag \\
    =& \lambda_\theta \nu_\theta'(K)\left[ 1 + c - \frac{\nu_0\lambda_\theta}{\lambda_\theta + r}(1 - e^{-(r+\lambda_\theta) K}) - e^{-rK}S_\theta(K) \right] < 0.\notag
    \end{align}
    Note that the second equality used the fact that $\nu_\theta(K) S_\theta(K) = \nu_0e^{-\lambda_\theta K}$, and the inequality follows because the belief $\nu_\theta(K)$ is decreasing in $K$ and the bracketed term is positive since 
    \begin{align*} 1 + c - \frac{\nu_0\lambda_\theta}{\lambda_\theta + r}(1 - e^{-(r+\lambda_\theta) K}) - e^{-rK}S_\theta(K) =& 1 + c - \int_0^\infty \nu_0 e^{-r \min(t,K)} \lambda_\theta e^{-\lambda_\theta t} \dd{t} - (1-\nu_0) e^{-rK} \\
    =& c + \int_0^\infty \nu_0 (1 - e^{-r \min(t,K)}) \lambda_\theta e^{-\lambda_\theta t} \dd{t} + (1-\nu_0)(1 - e^{-rK}), \end{align*} which is positive.
    \end{proof}

    Let $K_E$ denote the Proposition \ref{prop:benchmark} solution threshold when $\lambda = \lambda_E$ for sure (i.e., solves \eqref{defn:taustar}), and respectively $K_H$ for $\lambda_H$. Recall that $\phi_E(K_E) = 0$ and $\phi_H(K_H) = 0$, and Proposition \ref{prop:baseline_comp} implies that $K_E < K_H$. Since $\phi_\theta$ is strictly decreasing from Lemma \ref{lem:phi_decrease}, this immediately implies that $\phi_H(K_E) > 0$, and so $\phi_H(K) > 0 > \phi_E(K)$ for any $K \in (K_E, K_H)$. In consequence, since $\phi_\theta$ is strictly decreasing, it follows that $K^*_{i+1} \in (K_E, K_H)$, and $\phi_E(K^*_{i+1}) < 0 < \phi_H(K^*_{i+1})$. Further, note that by Bayes' rule, if $q_i$ denotes the belief when approach $i$ is brainstormed,
    \[ q_{i+1}(E) = \frac{q_i(E) S_E(K^*_{i})}{q_i(E) S_E(K^*_{i}) + q_i(H) S_H(K^*_{i})}\]
    Since $S_E(K) < S_H(K)$ for any $K> 0$, the subjective belief $q_{i}(E) > q_{i+1}(E)$ and $q_{i}(H) < q_{i+1}(H)$. Therefore, since $\phi_E(K^*_{i+1}) < 0$ and $\phi_H(K^*_{i+1})) > 0$, it follows that
    \[ \mathbb{E}[\phi_\theta(K^*_{i+1}) \mid t_{i+1}] < \mathbb{E}[\phi_\theta(K^*_{i+1}) \mid t_{i+2}] = 0.\]
    Rewriting out $\phi_\theta$,
    \begin{align*} &\mathbb{E}\left[ \begin{pmatrix} \lambda_\theta \nu_\theta(K^*_{i+1})  \left( 1 + c - \frac{\nu_0\lambda_\theta}{\lambda_\theta + r}(1 - e^{-(r+\lambda_\theta) K^*_{i+1}}) - e^{-rK^*_{i+1}}S_\theta(K^*_{i+1}) \right)  \\
     - r  \left( -c + \frac{\nu_0\lambda_\theta}{\lambda_\theta + r}(1 - e^{-(r+\lambda_\theta) K^*_{i+1}}) \right) \end{pmatrix} \mid t_{i+1} \right] < 0. \end{align*}
    Since $\nu_\theta$ is decreasing and the contradiction supposition assumed $K_i \ge K^*_{i+1}$, $\nu_\theta(K^*_i) < \nu_\theta(K^*_{i+1})$ so it follows that 
    \begin{align*} &\mathbb{E}\left[ \begin{pmatrix} \lambda_\theta \nu_\theta(K^*_{i})  \left( 1 + c - \frac{\nu_0\lambda_\theta}{\lambda_\theta + r}(1 - e^{-(r+\lambda_\theta) K^*_{i+1}}) - e^{-rK^*_{i+1}}S_\theta(K^*_{i+1}) \right)  \\
     - r  \left( -c + \frac{\nu_0\lambda_\theta}{\lambda_\theta + r}(1 - e^{-(r+\lambda_\theta) K^*_{i+1}}) \right) \end{pmatrix} \mid t_{i+1} \right] < 0, \end{align*}
    or 
    \begin{align} &\mathbb{E}\left[ \begin{pmatrix} \lambda_\theta \nu_\theta(K^*_{i})  \left( 1 - e^{-rK^*_{i+1}}S_\theta(K^*_{i+1}) \right)  \\
     - (r+\lambda_\theta \nu_\theta(K^*_{i}))  \left( -c + \frac{\nu_0\lambda_\theta}{\lambda_\theta + r}(1 - e^{-(r+\lambda_\theta) K^*_{i+1}}) \right) \end{pmatrix} \mid t_{i+1} \right] < 0. 
     \label{eqn:int_contradiction}
     \end{align}
    To finish the contradiction, consider a third and final interchange $\hat{k}$:
    \begin{align}
        \hat{k}_{i}(t) &:= \begin{cases}
         0 & t \in [t_{i+1}-\epsilon, t_{i+1} - \epsilon +  K_{i+1}^*) , \\
         1 & t \in [t_{i+1} - \epsilon + K_{i+1}^*, t_{i+1} + K_{i+1}^*), \\
         k_{i}(t) & t < t_{i+1} - \epsilon  \textnormal{ or } t\ge t_{i+1} + K^*_{i+1},
        \end{cases}, \notag  \\
        \hat{k}_{i+1}(t) &:= \begin{cases}
         1 & t \in [t_{i+1}-\epsilon, t_{i+1} - \epsilon +  K_{i+1}^*) , \\
         0 & t \in [t_{i+1} - \epsilon + K_{i+1}^*, t_{i+1} + K_{i+1}^*), \\
         k_{i+1}(t) & t < t_{i+1} - \epsilon  \textnormal{ or } t\ge t_{i+1} + K^*_{i+1},
        \end{cases}, \label{eqn:hat_inter} \\
        \hat{k}_j(t) &:= k_j(t) & j \not\in \{i,i+1\}. \notag 
    \end{align}
    See Figure \ref{fig:int_contr} for a depiction of the interchanged policy. The relevant interval now is from $t_{i+1}- \epsilon$ to $t_{i+1} + K^*_{i+1} = t_{i+2}$. To first order, the original path on the interval yielded payoff
    \[ \mathbb{E}\left[ \lambda_\theta \nu_\theta(K^*_{i})\epsilon+e^{-r\epsilon}(1 - \lambda_\theta \nu_\theta(K^*_{i})\epsilon)\left(-c + \int_0^{K^*_{i+1}} e^{-rt}(\nu_0 \lambda_\theta e^{-\lambda_\theta t})  \dd{t} \right)   \dd{t} \mid t_{i+1}-\epsilon \right] + O(\epsilon^2). \]
    The interchanged strategy $\hat{k}$ yields 
    \[ \mathbb{E}\left[ -c + \int_0^{K^*_{i+1}} e^{-rt}(\nu_0 \lambda_\theta e^{-\lambda_\theta t})  \dd{t} + e^{-rK^*_{i+1}}(1-\nu_0 + \nu_0e^{-\lambda_\theta K^*_{i+1}})\lambda_\theta \nu_\theta(K^*_{i})\epsilon \mid t_{i+1}-\epsilon \right] + O(\epsilon^2). \]
    Optimality implies that the payoff difference is nonnegative; taking $\epsilon \to 0$, this implies that
    \begin{align*} \mathbb{E} \left[ \begin{pmatrix}
       \lambda_\theta \nu_\theta(K^*_{i}) (1 - e^{-rK^*_{i+1}}(1-\nu_0 + \nu_0e^{-\lambda_\theta K^*_{i+1}}))  \\
       -( r + \lambda_\theta \nu_\theta(K^*_{i}))\left(-c + \int_0^{K^*_{i+1}} e^{-rt}(\nu_0 \lambda_\theta e^{-\lambda_\theta t})  \dd{t} \right)   \dd{t} 
    \end{pmatrix} \mid t_{i+1} \right]\ge 0,
    \end{align*}
    which contradicts \eqref{eqn:int_contradiction}, completing the proof. 
\end{proof}

 \begin{figure}
     \centering
     \begin{subfigure}[b]{0.3\textwidth}
         \centering
         \begin{tikzpicture}
             \draw[->] (0,0) -- (5,0);
             \filldraw (1,0) circle[radius=2pt] node[anchor=north]{$t_{i+2}$};
             \filldraw (4,0) circle[radius=2pt] node[anchor=north]{$t_{i+2} + K^*_{i+1}$};
             \fill [orange,opacity=0.5] (0,0) rectangle (1,2);
             \fill [blue,opacity=0.5] (1,0) rectangle (4,2);
             \fill [blue,opacity=0.5] (4,0) rectangle (5,1);
             \fill [orange,opacity=0.5] (4,1) rectangle (5,2);
             \draw (1,0) -- (1,2);
             \draw (4,0) -- (4,2);
         \end{tikzpicture}
         \caption{Original action path $k$}
         \label{fig:int_0}
     \end{subfigure}
     \hfill
     \begin{subfigure}[b]{0.3\textwidth}
         \centering         
         \begin{tikzpicture}
             \draw[->] (0,0) -- (5,0);
             \filldraw (1,0) circle[radius=2pt] node[anchor=north]{$t_{i+2}$};
             \filldraw (4,0) circle[radius=2pt] node[anchor=north]{$t_{i+2} + K^*_{i+1}$};
             \fill [orange,opacity=0.5] (0,0) rectangle (0.75,2);
             \fill [blue,opacity=0.5] (0.75,0) rectangle (3.5,2);
             \fill [blue,opacity=0.5] (3.5,0) rectangle (5,1);
             \fill [orange,opacity=0.5] (3.5,1) rectangle (5,2);
             \draw[thick] (0.75,0) rectangle (1,2);
             \draw[thick] (3.5,1) rectangle (4,2);
             \draw[<->, very thick] (1.1,1) -- (3.4,1.5);
         \end{tikzpicture}
         \caption{Interchange 1}
         \label{fig:int_1}
     \end{subfigure}
     \hfill
     \begin{subfigure}[b]{0.3\textwidth}
         \centering         
         \begin{tikzpicture}
             \draw[->] (0,0) -- (5,0);
             \filldraw (1,0) circle[radius=2pt] node[anchor=north]{$t_{i+2}$};
             \filldraw (4,0) circle[radius=2pt] node[anchor=north]{$t_{i+2} + K^*_{i+1}$};
             \fill [orange,opacity=0.5] (0,0) rectangle (1.25,2);
             \fill [blue,opacity=0.5] (1.25,0) rectangle (4.5,2);
             \fill [blue,opacity=0.5] (4.5,0) rectangle (5,1);
             \fill [orange,opacity=0.5] (4.5,1) rectangle (5,2);
             \draw[thick] (1,0) rectangle (1.25,2);
             \draw[thick] (4,1) rectangle (4.5,2);
             \draw[<->, very thick] (1.35,1) -- (3.9,1.5);
         \end{tikzpicture}
         \caption{Interchange 2}
         \label{fig:int_2}
     \end{subfigure}
        \caption{The original action path, and the interchanges. Orange denotes effort on approach $i+1$, blue on effort $i+2$.}
        \label{fig:int}
\end{figure}
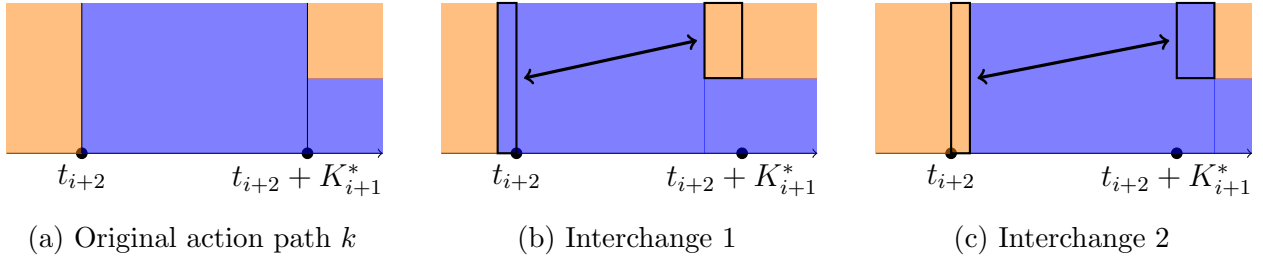
\begin{figure}
     \centering
     \begin{subfigure}[b]{0.45\textwidth}
         \centering
         \begin{tikzpicture}
             \draw[->] (0,0) -- (5,0);
             \filldraw (1,0) circle[radius=2pt] node[anchor=north]{$t_{i+1}$};
             \filldraw (4,0) circle[radius=2pt] node[anchor=north]{$t_{i+1} + K^*_{i+1}$};
             \fill [green,opacity=0.5] (0,0) rectangle (1,2);
             \fill [orange,opacity=0.5] (1,0) rectangle (4,2);
             \fill [blue,opacity=0.5] (4,0) rectangle (5,2);
             \draw (1,0) -- (1,2);
             \draw (4,0) -- (4,2);
         \end{tikzpicture}
         \caption{Original action path $k$}
         \label{fig:intc_0}
     \end{subfigure}
     \hfill
     \begin{subfigure}[b]{0.45\textwidth}
         \centering         
         \begin{tikzpicture}
             \draw[->] (0,0) -- (5,0);
             \filldraw (1,0) circle[radius=2pt] node[anchor=north]{$t_{i+1}$};
             \filldraw (4,0) circle[radius=2pt] node[anchor=north]{$t_{i+1} + K^*_{i+1}$};
             \fill [green,opacity=0.5] (0,0) rectangle (0.75,2);
             \fill [orange,opacity=0.5] (0.75,0) rectangle (3.75,2);
             \fill [green,opacity=0.5] (3.75,0) rectangle (4,2);
             \fill [blue,opacity=0.5] (4,0) rectangle (5,2);
             \draw[thick] (1,0) rectangle (0.75,2);
             \draw[thick] (3.75,0) rectangle (4,2);
             \draw[<->, very thick] (1.1,1) -- (3.65,1);
         \end{tikzpicture}
         \caption{Interchange 3}
         \label{fig:intc_1}
     \end{subfigure}
        \caption{The original action path, and the contradiction interchange. Green denotes effort on $i$, orange on $i+1$, and blue on $i+2$.}
        \label{fig:int_contr}
\end{figure}
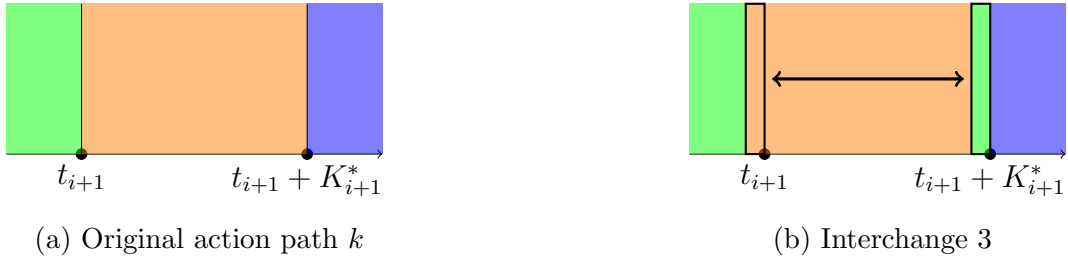

Lemma \ref{lem:threshold_decrease_after} implies that if the optimal sequence of thresholds starts to decrease, it must do so forever. However, if the the thresholds decrease forever, Lemma \ref{lem:lowest_history_of_effort} implies that approaches cannot be revisited, so the thresholds must solve a simpler optimal sequence of stopping times (since it is without loss for the agent to have no recall).

\begin{proof}[Proof of Lemma \ref{lem:increasing_thresholds}] I will prove the stronger result, that the sequence of thresholds is \textit{strictly} increasing.
    Suppose, for sake of contradiction, that optimally $K^*_i \ge K^*_{i+1}$ for some $i$ (and let $i$ denote the first such $i$). By Lemma \ref{lem:threshold_decrease_after}, this therefore implies that $K^*_{i+1} \ge K^*_{i+2}$; induction then immediately implies that $\{K^*_j\}_{j \ge i}$ must be weakly decreasing. Because $\{K^*_j\}_{j \ge i}$ must be weakly decreasing, Lemma \ref{lem:lowest_history_of_effort} then implies that after the $j+1$ approach is brainstormed, the policy never works on approach $j$ again after $j+1$ is brainstormed for all $j \ge i$. 
    So in the continuation problem after $j+1$ is brainstormed, it is \textit{without loss} to restrict the DM to never return to approaches. Therefore, the brainstorming times $\{K^*_j \}_{j \ge i}$ must solve the restricted series of stopping problems where the DM \textit{cannot} revisit approaches. 
    The continuation value of a sequence of thresholds $\{K_j\}_{j \ge i}$ immediately before approach $i+1$ is brainstormed is
    \[ 
    V_{i+1}(\{K_j\}_{j > i}) = \sum_\theta \delta_{i+1}(\theta)\begin{bmatrix} -c + \int_0^{K_{i+1}} \nu_0 \lambda_\theta e^{-(r+ \lambda_\theta) t} \dd t \\
    + \sum_{j = 2}^\infty \prod_{n=1}^{j-1} e^{-rK_{i+n}}S_\theta(K_{i+n})   \left[ -c + \int_0^{K_{i+j}} \nu_0 \lambda_\theta e^{-(r+ \lambda_\theta) t} \dd t  \right] \end{bmatrix},
    \]
    where $\delta_{i+1}(\theta)$ denotes the belief at time $t_{i+1}$ that the difficulty is $\theta$.
    The product term accounts for the accumulated discount factor and probability of survival, and the bracketed term is the contribution to the continuation value of approach $i+j$ for $j \ge 1$.
    It will be useful to relate the continuation value to the general stopping problem given by the agent's problem without recall. More precisely, consider the restricted optimization problem:
    \begin{equation}
    V(\delta) = \max_{\{K_j\}_{j > i}} \sum_\theta \delta(\theta) \begin{bmatrix} -c + \int_0^{K_{i+1}} \nu_0 \lambda_\theta e^{-(r+ \lambda_\theta) t} \dd t \\
    + \sum_{j = 2}^\infty \prod_{n=1}^{j-1} e^{-rK_{i+n}}S_\theta(K_{i+n})   \left[ -c + \int_0^{K_{i+j}} \nu_0 \lambda_\theta e^{-(r+ \lambda_\theta) t} \dd t  \right] \end{bmatrix}.
    \label{eqn:no_recall}
    \end{equation}
    where the implicit notation $\delta(\theta)$ is used to mean $\delta(H) = \delta$ and $\delta(E) = 1-\delta$. Intuitively, the problem optimizes a sequence of stopping thresholds to irreversibly move on to new approaches.
    Separate the terms by $\theta$, so
    \[ V_{\theta}(\{K_j\}_{j > i}) = -c + \int_0^{K_{i+1}} \nu_0 \lambda_\theta e^{-(r+ \lambda_\theta) t} \dd t + \sum_{j = 2}^\infty \prod_{n=1}^{j-1} e^{-rK_{i+n}}S_\theta(K_{i+n})   \left[ -c + \int_0^{K_{i+j}} \nu_0 \lambda_\theta e^{-(r+ \lambda_\theta) t} \dd t  \right]  \]
    Note that given $\theta$, only the sequence $\{K_j\}_{j > i}$ matters, not the belief over $\theta$.
    Recursion simplifies the expressions:
    \[ V_{\theta}(\{K_j\}_{j > i}) = -c + \int_0^{K_{i+1}} \nu_0 \lambda_\theta e^{-(r+ \lambda_\theta) t} \dd t + e^{-rK_{i+1}} S_\theta(K_{i+1}) V_{\theta}(\{K_j\}_{j > i+1}). \]
    Using this, the original maximization \eqref{eqn:no_recall} can be rewritten as
    \begin{align*}
        V(\delta) &= \max_{\{K_j\}_{j > i}} \sum_\theta \delta(\theta) V_{\theta}(\{K_j\}_{j > i}), \\
        &= \max_{\{K_j\}_{j > i}} \sum_\theta \delta(\theta) \left[ -c + \int_0^{K_{i+1}} \nu_0 \lambda_\theta e^{-(r+ \lambda_\theta) t} \dd t + e^{-rK_{i+1}} S_\theta(K_{i+1}) V_{\theta}(\{K_j\}_{j > i+1}) \right].
    \end{align*} 
    For convenience, relabel the sequence $\{K_j\}_{j > i}$ so the first term is $K_1$:  
    \begin{align}
        V(\delta)
        &= \max_{\{K_j\}_{j \ge 1}} \sum_\theta \delta(\theta) \left[ -c + \int_0^{K_{1}} \nu_0 \lambda_\theta e^{-(r+ \lambda_\theta) t} \dd t + e^{-rK_{1}} S_\theta(K_{1}) V_{\theta}(\{K_j\}_{j \ge 2}) \right]. \label{eqn:alt_problem_no_recall}
    \end{align}  
    Let $f(\delta, K_1, \{K_j\}_{j \ge 2})$ denote the objective of the optimization in \eqref{eqn:alt_problem_no_recall}.
    \begin{lemma}\label{lem:supermodular}
        For the optimal $\{K_j\}_{j \ge 2}$, the maximizer $K_1$ of $f$ is strictly increasing in $\delta(H)$.
    \end{lemma}
    \begin{proof}[Proof of Lemma \ref{lem:supermodular}]
    The cross derivative of the objective is 
    \begin{align}
        \frac{\partial^2 f}{\partial \delta(H) \, \partial K_1} 
        =& - e^{-rK_1} S_E(K_1) \left[ \lambda_E \nu_E(K_1)  - (r+\lambda_E \nu_E(K_1))V_{E}(\{K_j\}_{j \ge 2}) \right] \notag \\
        & + e^{-rK_1}S_H(K_1) \left[ \lambda_H \nu_H(K_1) - (r+\lambda_H \nu_H(K_1))V_{H}(\{K_j\}_{j \ge 2}) \right].\label{eqn:cross_deriv}
    \end{align}
    I will show that this quantity is positive. Optimality of $K_1$ implies that 
    \begin{align*}
        0 = \frac{\partial f}{\partial K_1} =& \sum_{\theta} \delta(\theta) \left[ \nu_0 \lambda_\theta e^{-(r+\lambda_\theta)K_1} - re^{-rK_1}S_\theta(K_1) V_\theta(\{K_j\}_{j\ge 2}) - \nu_0\lambda_\theta e^{-(\lambda_\theta + r)K_1}V_\theta(\{K_j\}_{j\ge 2}) \right] \\
        =&\sum_{\theta} \delta(\theta) e^{-rK_1} S_\theta(K_1) \left[ \lambda_\theta \nu_\theta(K_1) - (r + \lambda_\theta \nu_\theta(K_1))V_\theta(\{K_j\}_{j\ge 2})  \right] .
    \end{align*}
    Note that $e^{-rK_1}$ has no $\theta$-dependence, so the first-order condition simplifies to:
    \begin{equation} 0 = \sum_{\theta} \delta(\theta) S_\theta(K_1) \left[ \lambda_\theta \nu_\theta(K_1) - (r + \lambda_\theta \nu_\theta(K_1))V_\theta(\{K_j\}_{j\ge 2})  \right].  \label{eqn:cross_deriv_foc} \end{equation}
    I consider two cases, depending on whether $K_1 < K_H$ or $K_1 \ge K_H$.

    \textbf{Case 1:} Suppose $K_1 < K_H$. The maximum that $V_H(\{K_j\}_{j \ge 2})$ can attain is given by the optimal policy when $H$ is known for sure, which is exactly the benchmark solution. Hence, 
    \[ V_H(\{K_j\}_{j \ge 2}) \le \frac{\lambda_H \nu_H(K_H)}{\lambda_H\nu_H(K_H) + r} < \frac{\lambda_H\nu_H(K_1)}{\lambda_H\nu_H(K_1) + r}  \]
    where the second inequality follows since $\nu_H$ is strictly decreasing and $K_1 < K_H$ by supposition. Thus, the inequality implies that $\lambda_H \nu_H(K_1) - (r + \lambda_H \nu_H(K_1))V_H(\{K_j\}_{j\ge 2}) > 0$; the first-order condition \eqref{eqn:cross_deriv_foc} summing to zero immediately implies that $\lambda_E \nu_E(K_1) - (r + \lambda_E \nu_E(K_1))V_E(\{K_j\}_{j\ge 2}) < 0$, and so \eqref{eqn:cross_deriv} must be positive.

    \textbf{Case 2:} Suppose $K_1 \ge K_H$. Recall $\phi_E$ as defined from \eqref{eqn:phi}. By definition, $\phi_E(K_E) = 0$ so Lemma \ref{lem:phi_decrease} implies that $\phi_E(K_1) < 0$, so therefore 
    \begin{equation} \frac{\lambda_E \nu_E(K_1) }{\lambda_E \nu_E(K_1) + r} <  \frac{-c + \int_0^{K_1} \nu_0 \lambda_E e^{-(r+ \lambda_E )t} \dd{t} }{1 - e^{-rK_1}S_E(K_1)} = V_E(\{K_1, K_1,K_1,...\}). \label{eqn:vh_bound1}
    \end{equation}
    Further, recall that by Lemma \ref{lem:gittins} from the benchmark result, it was shown that the Gittins index maximand
    \[\frac{-c + \int_0^{K} \nu_0 \lambda_E e^{-(r+ \lambda_E )t} \dd{t} }{1 - e^{-rK}S_E(K)} \]
    has positive derivative at $K < K_E$ and negative derivative above $K > K_E$. Since $K_1 \ge K_H > K_E$, it follows that 
    \begin{align*} V_E(\{K_H, K_H, K_H,...\}) &= \frac{-c + \int_0^{K_H} \nu_0 \lambda_E e^{-(r+ \lambda_E )t} \dd{t} }{1 - e^{-rK_H}S_E(K_H)} \\
    &\ge  \frac{-c + \int_0^{K_1} \nu_0 \lambda_E e^{-(r+ \lambda_E )t} \dd{t} }{1 - e^{-rK_1}S_E(K_1)} = V_E(\{K_1, K_1,K_1,...\}). \end{align*}
    Note that since 
    \[ V_H(\{K_H,K_H,K_H,...\}) \ge V_H(\{K_j\}_{j\ge 2}) \]
    and $\{K_j\}_{j \ge 2}$ optimizes \eqref{eqn:alt_problem_no_recall}, 
    \begin{equation}
        V_E(\{K_j\}_{j\ge 2}) \ge V_E(\{K_H,K_H,K_H,...\})  \label{eqn:vh_bound2}
    \end{equation} 
    else the maximization could be improved by switching the continuation strategy to $K_H, K_H, K_H, ...$. Combining \eqref{eqn:vh_bound1} and \eqref{eqn:vh_bound2}, this implies that 
    \[ V_E(\{K_j\}_{j\ge 2}) > \frac{\lambda_E \nu_E(K_1) }{\lambda_E \nu_E(K_1) + r}. \]
    Hence $\lambda_E \nu_E(K_1) - (r + \lambda_E \nu_E(K_1))V_E(\{K_j\}_{j\ge 2}) < 0$; the first-order condition \eqref{eqn:cross_deriv_foc} summing to zero immediately implies that $\lambda_H \nu_H(K_1) - (r + \lambda_H \nu_H(K_1))V_H(\{K_j\}_{j\ge 2}) > 0$, and so \eqref{eqn:cross_deriv} must be positive.

    It is easy to see that 
    \[\frac{\partial f}{\partial K_1} =\sum_{\theta} \delta(\theta) e^{-rK_1} S_\theta(K_1) \left[ \lambda_\theta \nu_\theta(K_1) - (r + \lambda_\theta \nu_\theta(K_1))V_\theta(\{K_j\}_{j\ge 2})  \right] \]
    is decreasing in $K_1$, since $S_\theta(K_1)$, $e^{-rK_1}$, and $\nu_\theta(K_1)S_\theta(K_1) = \nu_0 e^{-\lambda_\theta K_1}$ are all decreasing in $K_1$, so the second-order condition for the FOC to be necessary and sufficient for a maximizer is satisfied. So by the implicit function theorem, the maximizer $K_1$ which satisfies the first order condition $\partial f/ \partial K_1 = 0$ increases in $\delta(H)$.  
    \end{proof}
    
    To finish the argument, note that by construction, $V(\delta_{j}(H)) = V^*_{j}$ for all $j > i$. Further, if $K^*_1$ denotes the maximizer in $K_1$ of the problem \eqref{eqn:alt_problem_no_recall}, then $K^*_{j} = K^*_1(\delta_j(H))$ for all $j > i$.
    However, because $\{ \delta_j(H) \}_{j \ge i}$ is strictly increasing, Lemma \ref{lem:supermodular} implies that $\{K^*(\delta_j(H))\}_{j \ge i} $ must be strictly increasing, so $\{K^*_j\}_{j\ge i}$ must be strictly increasing, contradicting the original supposition.
\end{proof}

Given Lemmas \ref{lem:lowest_history_of_effort} and \ref{lem:increasing_thresholds}, it suffices to optimize among increasing sequence of thresholds $\{ K_j \}_{j \ge 1}$ corresponding to the effort on the previous approach when a new approach is brainstormed. To pin down the optimal $K_n^*$, note that the choice of $K_n^*$ does not affect the payoff if breakthrough arrives before $K_{n-1}^*$ effort has been exerted on the first $n$ approaches, or after the $n+2$nd approach is brainstormed (since at that point, all approaches will have $K_{n+1}^*$ historical effort). 
    
Similar to the proof of Lemma \ref{lem:threshold_decrease_after}, let $k$ denote the optimal action path and $t_n$ denote the time that approach $n$ is brainstormed, and consider the following pair of interchanges:
        \begin{align*}
        k'_{i}(t) &:= \begin{cases}
         0 & t \in \left[t_{n+1} - \epsilon, t_{n+1} - \frac{n+1}{n}\epsilon + K^*_{n} \right), \\
         \frac{1}{n+1} & t \in \left[t_{n+1} - \frac{n+1}{n}\epsilon + K_{n}^*, t_{n+1} + K_{n+1}^*\right), \\
         k_{i}(t) & t < t_{n+1} - \epsilon \textnormal{ or } t\ge t_{n+1} + K^*_{n+1},
        \end{cases} & \forall i \le n, \\
        k'_{n+1}(t) &:= \begin{cases}
         1 & t \in [t_{n+1} - \epsilon, t_{n+1} - \frac{n+1}{n}\epsilon + K^*_{n}), \\
         \frac{1}{n+1} & t \in [t_{n+1} - \frac{n+1}{n}\epsilon + K_{n}^*, t_{n+1} + K_{n+1}^*), \\
         k_{n+1}(t) & t < t_{n+1} - \epsilon \textnormal{ or } t\ge t_{n+1} + K^*_{n+1},
        \end{cases}, \\
        k'_j(t) &:= k_j(t) & \forall j > n+1.
    \end{align*}
    and similarly: 
    \begin{align*}
        k''_{i}(t) &:= \begin{cases}
         \frac{1}{n} & t \in [t_{n+1}, t_{n+1} + \epsilon) , \\
         0 & t \in [t_{n+1} + \epsilon, t_{n+1} + \frac{n+1}{n}\epsilon + K_{n}^*), \\
         k_{i}(t) & t < t_{n+1}  \textnormal{ or } t\ge t_{n+1} + K^*_n + \frac{n+1}{n}\epsilon,
        \end{cases} & i \le n, \\
        k''_{n+1}(t) &:= \begin{cases}
         0 & t \in [t_{n+1}, t_{n+1} + \epsilon) , \\
         1 & t \in [t_{n+1} + \epsilon, t_{n+1} + \frac{n+1}{n}\epsilon + K_{n}^*), \\
         k_{n+1}(t) & t < t_{n+1}  \textnormal{ or } t\ge t_{n+1} + K^*_n + \frac{n+1}{n}\epsilon,
        \end{cases}, \\
        k''_j(t) &:= k_j(t) & j > n+1.
    \end{align*}
    See Figure \ref{fig:int_appendix} for a depiction of the interchanges. Intuitively, these two interchanges pin down the optimality of the \textit{time} $t_{n+1}$ that the $i+1$st approach is brainstormed, assuming the other brainstorming times are fixed and the policy adjusts accordingly to satisfy Lemma \ref{lem:lowest_history_of_effort}. By the same calculation as in Lemma \ref{lem:threshold_decrease_after}, the first-order condition for optimality is 
    \begin{equation}
        \mathbb{E} \left[ \begin{pmatrix}
       \lambda_\theta \nu_\theta(K^*_{n}) (1 - e^{-rK^*_{n}}S_\theta(K_{n}^*))  \\
       -( r + \lambda_\theta \nu_\theta(K^*_{n}))\left(-c + \int_0^{K^*_{n}} e^{-rt}(\nu_0 \lambda_\theta e^{-\lambda_\theta t})  \dd{t} \right)   \dd{t} 
    \end{pmatrix} \mid t_{n+1} \right] = 0. \label{eqn:appendix_int_condition}
    \end{equation} 
    Note that the belief at time $t_{n+1}$ depends on $K^*_n$. In particular, the belief that $\theta = H$ at $t_{n+1}$ is 
    \begin{equation}
        \delta_{n+1}(H) = \frac{\delta_0 S_H(K^*_n)^n}{\delta_0 S_H(K^*_n)^n + (1-\delta_0) S_E(K^*_n)^n}. \label{eqn:appendix_int_belief}
    \end{equation}  
    Combining \eqref{eqn:appendix_int_belief} with \eqref{eqn:appendix_int_condition},
    \[\sum_\theta  \delta_{n+1}(\theta) \begin{pmatrix}
       \lambda_\theta \nu_\theta(K^*_{n}) (1 - e^{-rK^*_{n}}S_\theta(K_{n}^*))  \\
       -( r + \lambda_\theta \nu_\theta(K^*_{n}))\left(-c + \int_0^{K^*_{n}} e^{-rt}(\nu_0 \lambda_\theta e^{-\lambda_\theta t})  \dd{t} \right)   \dd{t} 
    \end{pmatrix}  = 0. \]
    Multiplying through by the denominator of \eqref{eqn:appendix_int_belief}, and using $\phi_\theta$ as defined in \eqref{eqn:phi}, 
    \[ \delta_0 S_H(K^*_n)^n \phi_H(K^*_n) + (1-\delta_0)S_E(K^*_n)^n \phi_E(K^*_n) = 0. \]
    By rewriting the equation as 
    \[ \delta_0  \phi_H(K^*_n) + (1-\delta_0) \frac{S_E(K^*_n)^n}{S_H(K^*_n)^n}\phi_E(K^*_n) = 0, \]
    the fact that $\lim_{K\to\infty} S_E(K)^n/S_H(K)^n = 1$, together with the observation that \eqref{eqn:learning_c_assum} guarantees that $\lim_{K \to \infty} \delta_0 \phi_H(K) + (1-\delta_0)\phi_E(K) < 0$ implies that the LHS is strictly negative in the limit as $K^*_n \to \infty$. 
    Note further that $\delta_0 \phi_H(0) + (1-\delta_0) \phi_E(0) = \delta_0(r+\lambda_H \nu_0)c + (1-\delta_0)(r + \lambda_E\nu_0)c > 0$, so there exists a smallest solution $K^*_n$, and the LHS of \eqref{eqn:def_tn} crosses zero from above at $K^*_n$. 

    It remains to argue that the optimal policy must move on to generate arm $n+1$ at the smallest solution. Suppose not. Then the policy continues working on arms $[N]$ instead of generating at $N+1$ for when all arms in $[N]$ have historical effort $K > K^*_n$. Since the LHS of \eqref{eqn:def_tn} crosses zero from above, this implies that there exists some $K$ where $(1-\delta_0)S_E(K)^n \phi_E(K) + \delta_0 S_H(K)^n \phi_H(K) < 0,$ which implies that by the interchange construction of $\phi$, the value of the problem would be improved by generating the $N+1$st approach rather than working on the previous arms, a contradiction of optimality. This concludes the proof of Theorem \ref{thm:learning}.   

    \subsubsection{Proof of Theorem \ref{thm:learning_impossible}}
    Note that the logic of Theorem \ref{thm:learning} through the proof of Lemma \ref{lem:increasing_thresholds} still holds even when $\lambda_H = 0$, up to the last part of the argument. That is, it still holds that if it is optimal to brainstorm the $n+1$st approach, then the optimality of the time it is generated implies that 
    \[ \delta_0 S_H(K^*_n)^n \phi_H(K^*_n) + (1-\delta_0)S_E(K^*_n)^n \phi_E(K^*_n) = 0. \]
    However, now note that when $\lambda_H = 0$, $\phi_H$ is constant and equal to $rc$, and $S_H(\cdot) = 1$. So the condition becomes 
    \[ \delta_0 rc  + (1-\delta_0)S_E(K^*_n)^n \phi_E(K^*_n) = 0. \]
    Since $S_E(\cdot) \le 1$ and $\phi_E$ is initially positive (by assumption \eqref{eqn:learning_c_assum}) and decreasing, for sufficiently large $n$, there does not exist a solution $K^*_n$, so there is a maximum $\bar{N}$ number of approaches that the agent generates. To see that there exists a solution for all $n < \bar{N}$, it suffices to inductively show that if a solution exists for $n$, it also exists for $n-1$. That is, by the inductive hypothesis, for $K^*_n$, 
    \[ \delta_0 rc  + (1-\delta_0)S_E(K^*_n)^n \phi_E(K^*_n) = 0. \]
    Then $\phi_E(K^*_n) < 0$. Since $S_E(K^*_n) < 1$, this implies that 
    \[\delta_0 rc  + (1-\delta_0)S_E(K^*_n)^{n-1} \phi_E(K^*_n) < 0.\]
    So the intermediate value theorem implies that there exists a solution for $K^*_{n-1} \in (0, K^*_n)$, which completes the proof.
    
\subsection{Omitted Proofs in Section \ref{sec:continuum}}
As before, define the survival function: 
\[ S_\theta(x,t) = \exp \left( - \nu_0 x (1 - e^{-\lambda_\theta t/x}) \right).\]
I will provide the following  convenience Lemma for calculations later, which shows that by dividing through by the expectation of survival, one gets conditional expectations:
\begin{lemma}\label{lem:conditional_expectation}
    Let $\varphi(\lambda_\theta,x,t)$ be an arbitrary function of $\lambda_\theta,x,t$. Then \[ \frac{\mathbb{E}_\theta[\varphi(\lambda_\theta,x,t)S_\theta(x,t)] }{\mathbb{E}_\theta[S_\theta(x,t)]} = \mathbb{E}_\theta [\varphi(\lambda_\theta,x,t) \mid x,t ]. \]
\end{lemma}
\begin{proof}
    Note that $\mathbb{E}_\theta[S_\theta(x,t)]$ is the probability state $(x,t)$ is reached. Then 
    \begin{align*}
        \frac{\mathbb{E}_\theta[\varphi(\lambda_\theta,x,t)S_\theta(x,t)] }{\mathbb{E}_\theta[S_\theta(x,t)] } &= \frac{\sum_\theta p(\theta) \varphi(\lambda_\theta, x,t) S_\theta(x,t)}{\mathbb{E}_\theta[S_\theta(x,t)]}  \\
        &= \sum_\theta \frac{p(\theta)S_\theta(x,t) }{\sum_{\theta'}p(\theta')S_{\theta'}(x,t)}\varphi(\lambda_\theta,x,t) \\
        &= \sum_\theta p(\theta \mid x,t) \varphi(\lambda_\theta,x,t) = \mathbb{E}[\varphi(\lambda_\theta,x,t) | x,t].
    \end{align*}
\end{proof}

The proof of Proposition \ref{prop:continuum_single_comp} will be useful in the proof of Proposition \ref{prop:continuum_single}, so I will provide it first.

\begin{proof}[Proof of Proposition \ref{prop:continuum_single_comp}]
    Define the analogue of of $\phi$ as in the proof of Theorem \ref{thm:learning}, and the left-hand side of \eqref{eqn:constant_depth}:
    \begin{equation}
        \tilde{\phi}_\theta(d) := r \nu_0 \left(1 - e^{-\lambda_\theta d} - \lambda_\theta d e^{-\lambda_\theta d} \right) - rc - c \nu_0\lambda_\theta e^{-\lambda_\theta d}. \label{eqn:def_continuum_phi}
    \end{equation}
    It is not hard to verify that $\tilde{\phi}_\theta(d)$ is strictly increasing in $d$. Further, 
    \begin{align*}
        \tilde{\phi}_\theta(0) &= - rc - c\nu_0 \lambda_\theta < 0, \\
        \lim_{d\to\infty} \tilde{\phi}_\theta(d) &= r \nu_0 - rc > 0.
    \end{align*}
    Denote $d_\theta$ as the unique zero of $\tilde{\phi}_\theta$ on $[0, \infty)$. I now show that $d_\theta$ is decreasing in $\lambda_\theta$; note that by the implicit function theorem, it suffices to show that $\partial \tilde{\phi}_\theta / \partial \lambda_\theta$ is increasing at $d_\theta$. To see this, note that 
    \[ \frac{\partial \tilde{\phi}_\theta }{\partial \lambda_\theta} = \left( r \nu_0 \lambda_\theta d_\theta^2 + c \nu_0 \lambda_\theta d_\theta - c \nu_0 \right)e^{-\lambda_\theta d_\theta}. \]
    Applying the fact that $e^{-\lambda_\theta d} \ge 1 - \lambda_\theta d$ for any $d$,
    \begin{align*}
        0 = \tilde{\phi}_\theta(d_\theta) &= r\nu_0 \left(1 - e^{-\lambda_\theta d_\theta} - \lambda_\theta d e^{-\lambda_\theta d_\theta} \right) - rc - c \nu_0\lambda_\theta e^{-\lambda_\theta d_\theta} \\
        &\le r \nu_0 \left( \lambda_\theta^2 d_\theta^2  \right) - rc - c\nu_0 \lambda_\theta (1 - \lambda_\theta d_\theta) \\
        &= \lambda_\theta \left( r \nu_0 \lambda_\theta d_\theta^2 + c \nu_0 \lambda_\theta d_\theta - c \nu_0 - \frac{rc}{\lambda_\theta} \right).\\
        \implies 0 < \frac{rc}{\lambda_\theta} &\le r \nu_0 \lambda_\theta d_\theta^2 + c \nu_0 \lambda_\theta d_\theta - c \nu_0 
    \end{align*}
    Therefore $\frac{\partial \tilde{\phi}_\theta }{\partial \lambda_\theta}  \ge 0$, so $d_\theta$ is decreasing in $\lambda_\theta$ by the implicit function theorem. 

    Lastly, it is straightforward to check that 
    \[ \frac{\partial \tilde{\phi}(d) }{\partial c} = -r - \nu_0 \lambda_\theta e^{-\lambda_\theta d} < 0, \]
    so by the implicit function theorem, $d_\theta$ is increasing in $c.$
\end{proof}

\begin{proof}[Proof of Proposition \ref{prop:continuum_single}]
    The problem \eqref{eqn:continuum_single} can be rewritten as 
    \begin{gather*}
        \max_{x} \int_0^\infty \left( e^{-rt}- \int_0^t e^{-rs}c \dot{x} \dd{s} \right) \dd{F(x,t)}.
    \end{gather*}
    For now, drop the nonnegativity constraints. Integrating by parts, this becomes
    \begin{gather*}
        \max_{x} \int_0^\infty e^{-rt} \left( r F(x,t) - (1-F(x,t)) c\dot{x} \right) \dd{t}.
    \end{gather*}
    The rest of the solution is straightforward application of the Euler-Lagrange condition. Note the Lagrangian objective is concave in $x$ because $F$ is concave in $x$. The partial derivative of the integrand in $x$ is 
    \[ e^{-rt} \left(r\frac{\partial F(x,t) }{\partial x }+ \frac{\partial F(x,t)}{\partial x} c \dot{x} \right).  \]
    The partial derivative of $\dot{x}$ is 
    \[ e^{-rt} \left( - (1-F(x,t))c  \right). \]
    The total derivative of this partial derivative in $t$ is 
    \begin{align*}  -r e^{-rt}\left(  - (1-F(x,t))c  \right) + e^{-rt}\left( c \left(\frac{\partial F(x,t)}{\partial t} + \frac{\partial F(x,t)}{\partial x} \dot{x} \right) \right).
    \end{align*}
    Setting this equal to the partial derivative in $x$, 
    \[ r  \frac{\partial F(x,t)}{\partial x}   =r (1- F(x,t) )c + c \frac{\partial F(x,t)}{\partial t}.  \]
    Since the Lagrangian objective is concave in $x$, the Euler-Lagrange condition is necessary and sufficient for an optimal solution in the unconstrained problem, so it suffices to check that the Euler-Lagrange equation's solution is admissible. Fix a $t>0$, and rewrite the equation as 
    \[r \left( \frac{\partial F(x,t)}{\partial x} - (1- F(x,t) )c \right) - c \frac{\partial F(x,t)}{\partial t} = 0. \]
    Invoking the forms of the derivatives and linearity of expectation, this becomes 
    \[ \mathbb{E}_\theta \left[ \left(r \nu_0 \left(1 - e^{-\lambda_\theta t/x} - \frac{\lambda_\theta t}{x} e^{-\lambda_\theta t/x} \right) -r c - c \nu_0\lambda_\theta e^{-\lambda_\theta t/x}\right)S_\theta(x,t) \right] = 0.\]
    Using the definition \eqref{eqn:def_continuum_phi}, the Euler-Lagrange equation can be rewritten:
    \begin{equation}
        \delta_0 S_H(x,t) \tilde{\phi}_H(t/x) + (1-\delta_0)S_E(x,t) \tilde{\phi}_E(t/x) = 0, \label{eqn:cont_single_EL_rewrite}
    \end{equation} 
    which is increasing in $t/x$ fixing $t$, since $S_\theta(x,t) = \exp\left( - \nu_0 t \frac{1 - e^{-\lambda_\theta t/x}}{t/x}\right)$, and $(1 - e^{-\lambda_\theta d})/d$ is strictly decreasing in $d$. Note that the Euler-Lagrange equation also can be rewritten as
    \[ \left(\frac{\delta_0}{1-\delta_0} \frac{S_H(x,t)}{S_E(x,t)} \right)\tilde{\phi}_H(t/x) +  \tilde{\phi}_E(t/x) = 0. \]
    Since the left-hand side is negative for $t/x = 0$ (fixing $t$) and strictly positive in the limit as $t/x\to\infty$ fixing $t$, there exists a unique solution $t/x^*$ to the Euler-Lagrange equation for any fixed $t$, which also thus implies a unique solution $x^*$. 
    Further, since the parenthesized part is positive and $\tilde{\phi}_\theta$ is monotone, $t/x$ must always lie in $[d_E, d_H]$.

    To verify the feasibility constraints on $x^*$, it suffices to check that $\dot{x}^* \ge 0$ and 
    \[ \frac{\dd{}}{\dd{t}}\left(\frac{t}{x^*}\right) = \frac{x^*- t\dot{x}^*}{(x^*)^2} \ge 0.  \]
    To do this, I will invoke the implicit function theorem. First, I show that $t/x^*$ is weakly increasing in time $t$. To see this, rewrite the Euler-Lagrange equation condition in $d = t/x$ and $t$: 
    \[ \delta_0 \hat{S}_H(d,t) \tilde{\phi}_H(d) + (1-\delta_0)\hat{S}_E(d,t) \tilde{\phi}_E(d) = 0, \]
    where $\hat{S}_\theta(d,t) = \exp\left(- \nu_0 t \frac{1 - e^{-\lambda_\theta d}}{d} \right)$. Since I already showed that the left-hand side expression is increasing in $d$, it suffices to show that it is decreasing in $t$ at the optimal $d^*$. The partial derivative of the LHS with respect to $t$ is 
    \[ -\delta_0 \nu_0 \left(\frac{1 - e^{-\lambda_H d}}{d} \right) \hat{S}_H(d,t)\tilde{\phi}_H(d)  -(1-\delta_0)\nu_0 \left( \frac{1 - e^{-\lambda_E d}}{d}\right) \hat{S}_E(d,t)\tilde{\phi}_E(d). \]
    Evaluating this at the $d$ that solves the Euler-Lagrange condition, the $1/d$ parts of both terms cancel, so this simplifies to 
    \[ \frac{\nu_0}{d}e^{-\lambda_E d}\left(\delta_0 e^{(\lambda_E-\lambda_H) d}\hat{S}_H(d,t)\tilde{\phi}_H(d) +(1-\delta_0)\hat{S}_E(d,t)\tilde{\phi}_E(d) \right) . \]
    Applying the Euler-Lagrange condition again, 
    \[ \frac{\nu_0}{d}e^{-\lambda_E d}\left(\delta_0 (e^{(\lambda_E-\lambda_H) d} - 1)\hat{S}_H(d,t)\tilde{\phi}_H(d) \right) < 0, \]
    where inequality follows because $\tilde{\phi}_H(d) < 0$ since $d \in [d_E, d_H]$ and $\tilde{\phi}_H$ is increasing and zero at $d_H$.
    Hence, the implicit function theorem implies that $t/x^*$ is increasing, which immediately implies that $\dot{x}^* \le x/t$. 
    A similar argument shows that $\dot{x}^* \ge 0$, so the solution to the Euler-Lagrange equation is admissible.
\end{proof}

\subsubsection{Proof of Theorem \ref{thm:microfound}}
Before proving the theorem, I show one intermediate lemma:
\begin{lemma}\label{lem:increasing_depth}
    For any feasible $x$, $x/t$ must be weakly decreasing. 
\end{lemma}
\begin{proof}
    Taking the derivative, 
    \[ \frac{\dd{}}{\dd{t}}\left( \frac{x}{t} \right)  = \frac{\dot{x} t  - x}{t^2} \le 0,\]
    since $\dot{x} \le x/t$.
\end{proof}

\begin{proof}[Proof of Theorem \ref{thm:microfound}]
    Proceed point by point. 
    
    \textbf{Point 1:} Suppose $\{\sigma^n \}_n$ converges in strategy to $x$. First, I show the breakthrough time converges in distribution to that induced by $x$. To show this,
        I will instead construct a bracket $\ubar{F}^n(t;\sigma^n) \le F^n(t;\sigma^n) \le \bar{F}^n(t;\sigma^n)$, and show that both $\ubar{F}^n(t;\sigma^n) \to F(x(t),t)$ and $\bar{F}^n(t;\sigma^n) \to F(x(t),t)$ (and so the squeeze theorem implies the convergence in distribution). Note that for each $n$, the number of arms brainstormed by time $t$ is some $j(n)+1$, where $j(n)$ is such that 
        \[ j(n) K^n_{j(n)} \le t \le (j(n)+1)K^n_{j(n)+1}. \]
        Then 
        \[ \hat{N}^n(t) = \frac{j(n)+1}{n}, \]
        so by convergence in strategy, $\frac{j(n) + 1}{n} \to x(t)$ (which also implies that $j(n)/ n \to x(t)$). 
        
        Consider 
        \begin{align*} \ubar{F}^n(t;\sigma^n) &= 1 -  \sum_\theta \delta_\theta \left(1 - \frac{\nu^*_0}{n} + \frac{\nu^*_0}{n} e^{-\lambda_\theta^*n t /j(n)} \right)^{j(n)} , \\
        \bar{F}^n(t;\sigma^n) &= 1 -  \sum_\theta \delta_\theta \left(1 - \frac{\nu^*_0}{n} + \frac{\nu^*_0}{n} e^{-\lambda_\theta^*n t/(j(n) + 1)} \right)^{j(n)+1},
        \end{align*} 
        where $\delta_\theta$ is understood to mean $\delta_H = \delta_0^*$, $\delta_E = 1 - \delta_0^*$. To verify that these values bracket $F^n(t;\sigma^n)$, note that Lemma \ref{lem:lowest_history_of_effort} implies that in the interval $[j(n)K^n_{j(n)}, (j(n)+1)K^n_{j(n)})$, the strategy $\sigma^n$ only works on arm $j(n) + 1$, so if $t$ lies in this interval, 
        \[ F^n(t;\sigma^n) = 1 - \sum_\theta \delta_\theta \left(1 - \frac{\nu^*_0}{n} + \frac{\nu^*_0}{n} e^{-\lambda_\theta^*n K^n_{j(n)}} \right)^{j(n)}\left( 1 - \frac{\nu^*_0}{n} + \frac{\nu^*_0}{n} e^{-\lambda_\theta^*n \left(t - j(n)K^n_{j(n)}\right)}\right) . \]
        To see that this lies between $\ubar{F}^n(t;\sigma^n)$ and $\bar{F}^n(t;\sigma^n)$, it suffices to show that the expression
        \[ \varphi(Z) =  \left(1 - \frac{\nu^*_0}{n} + \frac{\nu^*_0}{n} e^{-\lambda_\theta^*n Z } \right)^{j(n)}\left( 1 - \frac{\nu^*_0}{n} + \frac{\nu^*_0}{n} e^{-\lambda_\theta^*n \left(t - j(n)Z\right)}\right)\]
        is monotone on $Z \in [ t/(j(n)+1), t/j(n)]$; since $Z = t/j(n)$ gives $\ubar{F}^n(t;\sigma^n)$, setting $Z = t/(j(n)+1)$ gives $\bar{F}^n(t;\sigma^n)$ and $F^n(t;\sigma^n)$ corresponds to $Z = K^n_{j(n)} \in [t/(j(n)+1),t/j(n)]$ where the inclusion follows because $t \in [j(n)K^n_{j(n)}, (j(n)+1)K^n_{j(n)}) $. To show $\varphi(Z)$ is monotone, it suffices to show that $\log \varphi(Z)$ is decreasing:
        \[ \log \varphi(Z) = j(n) \log  \left(1 - \frac{\nu^*_0}{n} + \frac{\nu^*_0}{n} e^{-\lambda_\theta^*n Z } \right) +  \log \left(1 - \frac{\nu^*_0}{n} + \frac{\nu^*_0}{n} e^{-\lambda_\theta^*n \left(t - j(n)Z\right)} \right), \]
        so 
        \begin{align*} \frac{\dd{}}{\dd{Z}} \log \varphi (Z) &=  j(n) \frac{- \lambda_\theta^*n e^{-\lambda_\theta^*n Z }}{1 - \frac{\nu^*_0}{n} + \frac{\nu^*_0}{n} e^{-\lambda_\theta^*n Z } } + \frac{ \lambda_\theta^* n j(n) e^{-\lambda_\theta^*n (t - j(n)Z) }}{1 - \frac{\nu^*_0}{n} + \frac{\nu^*_0}{n} e^{-\lambda_\theta^*n \left(t - j(n)Z\right)}}\\
        &= j(n)\lambda_\theta^*n\left( - \frac{ e^{-\lambda_\theta^*n Z }}{1 - \frac{\nu^*_0}{n} + \frac{\nu^*_0}{n} e^{-\lambda_\theta^*n Z } }+ \frac{ e^{-\lambda_\theta^*n (t - j(n)Z) }}{1 - \frac{\nu^*_0}{n} + \frac{\nu^*_0}{n} e^{-\lambda_\theta^*n \left(t - j(n)Z\right)}} \right) < 0.
        \end{align*}
        To see the inequality, note that the function $z \mapsto e^{-\lambda_\theta^* n z} / (1 - \frac{\nu^*_0}{n} + \frac{\nu^*_0}{n} e^{-\lambda_\theta^*n z} )$ is decreasing in $z$, and by assumption, $t - j(n)Z \le t - j(n) t / (j(n)+1) = t/(j(n)+1) \le Z$. Hence $\varphi(Z)$ is decreasing on $Z \in [t/(j(n)+1), t/j(n)]$, and so $F^n(t;\sigma^n)$ is bracketed by $\ubar{F}^n(t;\sigma^n)$ and $\bar{F}^n(t;\sigma^n)$ when  $t \in [j(n)K^n_{j(n)}, (j(n)+1)K^n_{j(n)}) $.
        
        In the other case, when $t \in [ (j(n)+1)K^n_{j(n)}), (j(n)+1)K^n_{j(n)+1})$, strategy $\sigma^n$ splits effort evenly on $j(n) + 1$ arms, so the state has exactly $t / (j(n)+1)$ effort on all $j(n)+1$ arms; so $F^n(t;\sigma^n) = \bar{F}^n(t;\sigma^n)$ in this case, so the bracket holds (since $F^n$ coincides with $\bar{F}^n$).

        It remains to verify that the bracket converges to $F(x(t), t)$ as $n\to \infty$.  A useful identity is that if $A_n \to A$, and $B_n \to B$ and $A, B \in (0,\infty)$, then 
        \begin{equation} \lim_{n \to \infty} \left(1 - \frac{A_n}{n} \right)^{n B_n} = e^{-AB}. \label{eqn:limiting_identity}\end{equation}
        Then straightforward limit-taking gives
        \begin{align*}
            \lim_{n \to \infty}\ubar{F}^n(t;\sigma^n)  &= 1 - \sum_\theta \delta_\theta \lim_{n \to \infty} \left(1 - \frac{\nu^*_0}{n} + \frac{\nu^*_0}{n} e^{-\lambda_\theta^*n t /j(n)} \right)^{j(n)} \\
            &= 1 - \sum_\theta \delta_\theta \lim_{n \to \infty} \left(1 - \frac{\nu^*_0}{n}\left( 1 - e^{-\lambda_\theta^* \frac{t}{[j(n)/n]}}  \right) \right)^{n \cdot [j(n)/n]} \\
            &= 1 - \sum_\theta \delta_\theta \exp\left(- \nu_0 \left(1 - e^{-\lambda_\theta t/x(t)} \right) x(t) \right) = F(x(t),t),
        \end{align*}
        since $j(n)/n \to x(t)$ by convergence in strategy, and the last line uses identity \eqref{eqn:limiting_identity}. Similarly, 
        \begin{align*}
            \lim_{n \to \infty}\bar{F}^n(t;\sigma^n)  &= 1 - \sum_\theta \delta_\theta \lim_{n \to \infty} \left(1 - \frac{\nu^*_0}{n} + \frac{\nu^*_0}{n} e^{-\lambda_\theta^*n t /[j(n) + 1]} \right)^{j(n)+1} \\
            &= 1 - \sum_\theta \delta_\theta \lim_{n \to \infty} \left(1 - \frac{\nu^*_0}{n}\left( 1 - e^{-\lambda_\theta^* \frac{t}{[(j(n) + 1)/n]}}  \right) \right)^{n \cdot [(j(n)+1)/n]} \\
            &= 1 - \sum_\theta \delta_\theta \exp\left(- \nu_0 \left(1 - e^{-\lambda_\theta t/x(t)} \right) x(t) \right) = F(x(t),t),
        \end{align*}
        which follows exactly the same steps as before, since $(j(n) + 1) / n \to x(t)$. Hence the squeeze theorem implies that $F^n(t;\sigma^n) \to F(x(t),t)$, and so the breakthrough time converges in distribution to that induced by $x$; since $t$ was chosen arbitrarily, uniform convergence of $\hat{N}^n$ to $x$ implies that the breakthrough distribution function also uniformly converges.
        
        Finally, having shown that the breakthrough time converges in distribution, consider the payoffs:
        \[ \Pi^n(\hat{\sigma}^n) = \mathbb{E}_\tau\left[\left.e^{-r\tau} - \sum_{j=1}^\infty e^{-rjK_j^n} \frac{c^*}{n} \mathbb{I}[jK^n_j \le \tau] \right| \hat{\sigma}^n\right]. \]
        Since the breakthrough time $\tau$ already converges in distribution, the $e^{-r\tau}$ part converges similarly. It suffices to show that the cost term converges appropriately. To see this, note that for any $\tau$, 
        \[ \lim_{n \to \infty} \sum_{j=1}^\infty e^{-rjK_j^n} \frac{c^*}{n} \mathbb{I}[jK^n_j \le \tau] = \lim_{n \to \infty} \int_0^\tau  e^{-r t} c^*  \dd{\hat{N}^n(t)}  =  \int_0^\tau e^{-rt} c^* \dot{x}(t) \dd{t}. \]
        The first equality follows by definition of the weighted counting measure $\hat{N}(t)$ induces, and the second follows because convergence in strategy implies that the Stieltjes counting measure induced by $\hat{N}^n$ weak converges to the Stieltjes measure induced by $x$ on $\mathbb{R}$ (since $x$ is continuous), so the integration of any continuous bounded measurable function (namely $e^{-rt} c^*$) under $\dd\hat{N}^n$ converges to its integral under $\dd{x}$. 
    
    \textbf{Point 2:} Take a feasible $x$. I construct the strategies $\{\sigma^n\}_{n \ge 1}$ such that $\sigma^n$ converges in strategy to $x$. Since the constructed strategies satisfy Lemmas \ref{lem:lowest_history_of_effort} and \ref{lem:increasing_thresholds}, it suffices to construct for each $n$, the effort threshold sequence $\{K^n_j\}$.  I will need the notion of the generalized inverse of $x$; since  $x(0) = 0$ and $x$ is increasing, the generalized inverse of $x$ is defined as:
        \[ x^{-1}(x') := \inf_t \{ x(t) \ge x'\},\]
        where $x^{-1}(x') = \infty$ if $x(t) < x'$ for all $t$. Further, since $x$ is continuous $x(x^{-1}(x')) = x'$, and since $x$ is weakly increasing, $x^{-1}(x') \le t$ implies that $x' \le x(t)$.
        Construct $K^n_j$ as follows:
        \begin{equation} K^n_j := \frac{1}{j}x^{-1}\left( \frac{j}{n}\right). \label{eqn:thm3_Kconstruct}\end{equation}
        To confirm that the induced strategy satisfies Lemma \ref{lem:increasing_thresholds}, it suffices to check that consecutive $K^n_j$ are weakly increasing (and induction implies the rest). Suppose, for sake of contradiction, that 
        \[ K^n_{j+1} = \frac{x^{-1}((j+1)/n)}{j+1} < \frac{x^{-1}(j/n)}{j} = K^n_j.\]
        Since $x$ is weakly increasing, let $\bar{t} = x^{-1}((j+1)/n) > x^{-1}(j/n) = \ubar{t}$. Then the contradiction supposition implies that 
        \[ \bar{t} < \frac{j+1}{j} \ubar{t}.\]
        Then 
        \[ \frac{x(\bar{t})}{\bar{t}} > \frac{x(\bar{t})}{\ubar{t}} \left(\frac{j}{j+1}\right) = \frac{(j+1)/n}{\ubar{t}} \left(\frac{j}{j+1}\right) = \frac{x(\ubar{t})}{\ubar{t}}. \]
        But this violates feasibility: since $\dot{x} \le x/t$, Lemma \ref{lem:increasing_depth} implies that $x/t$ must be decreasing for any feasible $x$. Therefore the construction \eqref{eqn:thm3_Kconstruct} has weakly increasing thresholds.
        
        Let $\hat{\sigma}^n$ denote the corresponding strategy induced by the thresholds in \eqref{eqn:thm3_Kconstruct}, and $\hat{N}^n$ the corresponding normalized number of arms as in \eqref{def:nna}. 
        I show that this construction converges in strategy to $x$. Pointwise convergence at $t = 0$ trivially holds, since $\hat{N}^n(0) = 1/n \to 0 = x(0)$ as $n \to \infty$. So fix a $t > 0$ and an $\epsilon > 0$. Since $x$ is continuous, there exists a $\Delta$ such that $|x(t + t') - x(t)| < \epsilon$ for any $t' \in [-\Delta, \Delta]$. Take some $n > 1/\epsilon$. Since $\epsilon > 1/n$, the mesh induced by $n$ is sufficiently fine that there is some $j$ such that 
        \[ x(t) - \epsilon <  \frac{j}{n} < x(t) \le \frac{j+1}{n} < x(t) + \epsilon. \]
        Let $\ubar{t} = x^{-1}(j/n)$, and $\bar{t} = x^{-1}((j+1)/n)$. By construction, $\hat{N}$ is such that 
        \begin{align*} \hat{N}^n(\ubar{t}) =& \frac{1}{n}\left(1 + \sum_{m=1}^\infty \mathbb{I}[mK^n_m < x^{-1}(j/n)] \right),\\
        &=  \frac{1}{n}\left(1 + \sum_{m=1}^\infty \mathbb{I}[x^{-1}(m/n) < x^{-1}(j/n)] \right) = \frac{j}{n} = x(\ubar{t}).
        \end{align*}
        since by construction $x(x^{-1}(j/n)) = j/n$. Similarly,  
        \begin{align*} \hat{N}^n(\bar{t}) =& \frac{1}{n}\left(1 + \sum_{m=1}^\infty \mathbb{I}[mK^n_m < x^{-1}((j+1)/n)] \right),\\
        &=  \frac{1}{n}\left(1 + \sum_{m=1}^\infty \mathbb{I}[x^{-1}(m/n) < x^{-1}((j+1)/n)] \right) = \frac{j+1}{n} = x(\bar{t}).
        \end{align*}
        Since $x(t) \in [j/n, (j+1)/n]$ by construction, and the fact that $\hat{N}^n(t)$ is increasing implies that $\hat{N}^n(t) \in [\hat{N}^n(\ubar{t}), \hat{N}^n(\bar{t})] = [j/n, (j+1)/n]$, it follows that $| x(t) - \hat{N}^n(t)| \le 1/n < \epsilon$. Since the choice of $n$ was independent of $t$, $\hat{N}^n$ uniformly converges to $x$ and so $\hat{\sigma}^n$ converges in strategy to $x$.
        
    \textbf{Point 3:} Fix a $t$.  From the discrete-arm characterization, the optimal strategy for $D^n$ is characterized by the thresholds $\{K^n_{\OPT,j} \}_j$, where $K^n_{\OPT,j}$ is the $K$ that solves
        \begin{equation}
        (1-\delta^*_0) S^n_E(K) ^j\phi^n_E(K)  + \delta^*_0 S^n_H(K)^j \phi^n_H(K) = 0, \label{eqn:discrete_threshold_condition}
    \end{equation}
    where
    \begin{align*} 
    S_\theta^n(K) &:=  1 - \frac{\nu^*_0}{n} + \frac{\nu^*_0}{n}e^{-\lambda^*_\theta n K}, & \nu_\theta^n(K) &:= \frac{ (\nu_0^*/n)e^{-\lambda_\theta^*n K} }{1 - \nu_0^*/n + (\nu_0^*/n)e^{-\lambda_\theta^*n K} },
    \end{align*}
    \[\phi_\theta^n(K) := \lambda_\theta^*n \nu^n_\theta(K)  - ( r + \lambda_\theta^*n \nu^n_\theta(K))  \left[ -\frac{c^*}{n} + \frac{\nu^*_0\lambda^*_\theta}{n\lambda^*_\theta + r}(1 - e^{-(r+\lambda_\theta^*n) K}) \right] - e^{-rK}S_\theta^n(K)\lambda_\theta^*n \nu_\theta^n(K). \]
    Note that as $n \to \infty$, $\phi_\theta^n(K) \to \lambda_\theta^*\nu_0^* (1 - e^{-rK})$. Further, for any sequence $y_n$ such that $y_n /n$ converges to some finite limit $Y$, $S^n_\theta(K)^{y_n} \to e^{-\nu_0 Y} > 0$. Since for all $K^n_{\OPT, j}$, the condition \eqref{eqn:discrete_threshold_condition} must hold, it follows that for any such $\{y_n\}$, $\lim_{n \to \infty} K^n_{\OPT, y_n} \to 0$, else the LHS of \eqref{eqn:discrete_threshold_condition} would converge to something nonzero, and the condition would fail for  sufficiently large $n$. 
    
    It will be useful to scale \eqref{eqn:discrete_threshold_condition} by $n$, so that 
    \[ (1-\delta^*_0) S^n_E(K) ^j\hat\phi^n_E(K)  + \delta^*_0 S^n_H(K)^j \hat\phi^n_H(K) = 0, \]
    where $S^n_\theta$ is as before, but 
    \[ \hat\phi_\theta^n(K) := n \phi_\theta^n(K). \]
    Note that there exists some $j^n$ such that 
    \[ j^n K^n_{\OPT,j^n} \le t < (j^n+1)K^n_{\OPT,j^n +1}. \]
    I first claim that $\lim\sup n K^n_{\OPT, j^n} > 0$. Suppose not. 
    Then along such an unbounded subsequence such that $j^n K^n_{\OPT, j^n}$ converges (which exists since $j^n K^n_{\OPT, j^n}$ is bounded in $[0,t]$), a Taylor expansion gives 
    \begin{align*} \lim _{n \to \infty} \left[S^n_\theta(K_{\OPT,j^n}^n) \right]^{j^n} &= \lim_{n \to \infty} \left[ 1 -  \frac{\nu_0}{n}(1 - e^{-\lambda_\theta n K^n_{\OPT,j^n}}) \right]^{j^n} \\
    &= \lim_{n \to \infty} \left[ 1 -  \frac{\nu_0}{j^n}( \lambda_\theta j^n K^n_{\OPT,j^n} + j^nK^n_{\OPT,j^n}o\left(nK^n_{\OPT,j^n}\right) \right]^{j^n}  \\
    &= e^{-\nu_0 \lambda_\theta t \lim_{n \to \infty}j^n K^n_{\OPT,j^n}} \ge e^{-\nu_0 \lambda_\theta t} .\end{align*}
    So $S_\theta^n(K^n_{\OPT,j^n})^{j^n}$ converges to a finite and positive limit along this subsequence. To finish the contradiction, I show that $\hat{\phi}^n_\theta(K^n_{\OPT,j^n})$ diverges along this subsequence; together, this implies that for some sufficiently large $n$ \eqref{eqn:discrete_threshold_condition} must be violated. To show that $\hat{\phi}^n_\theta(K^n_{\OPT,j^n})$ diverges, I rewrite it as 
    \begin{align*} \hat{\phi}^n_\theta(K^n_{\OPT, j^n}) =& \frac{\lambda_\theta^* \nu_0^* e^{-\lambda_\theta^*nK^n_{\OPT,j^n}}\left[ n(1 - e^{-rK^n_{\OPT,j^n}}S_\theta^n(K^n_{\OPT,j^n})) + \nu_0^* (1 - e^{-\lambda_\theta n K^n_{\OPT,j^n}}) \right] }{1 - (\nu_0^*/n) + (\nu_0^*/n)e^{-\lambda_\theta^*nK^n_{\OPT,j^n}}}\\
    &- \left( r + \frac{\lambda_\theta^* \nu_0^* e^{-\lambda_\theta^*nK^n_{\OPT,j^n}}}{1 - (\nu_0^*/n) + (\nu_0^*/n)e^{-\lambda_\theta^*nK^n_{\OPT,j^n}}}\right)\left( -c^* + \frac{\nu_0^*\lambda_\theta^* n}{\lambda_\theta^*n + r}(1 - e^{-(r + \lambda_\theta^*n)K^n_{\OPT,j^n}} \right)\end{align*}
    Taking the limit as $n \to \infty$, note that 
    $S_\theta^n(K^n_{\OPT,j^n}) \to 1$ and 
    \[ \frac{\lambda_\theta^* \nu_0^* e^{-\lambda_\theta^*nK^n_{\OPT,j^n}}}{1 - (\nu_0^*/n) + (\nu_0^*/n)e^{-\lambda_\theta^*nK^n_{\OPT,j^n}}} \to \lambda_\theta^* \nu_0^*\]
    so
    \[ \lim_{n\to\infty}\hat\phi^n_\theta(K^n_{\OPT,j^n}) = \lambda_\theta^*\nu_0^* \lim_{n\to\infty} n(1 - e^{-rK^n_{\OPT,j^n}}) - (r + \lambda_\theta^* \nu_0^*)(-c^*). \]
    A Taylor series expansion of $n (1 - e^{-rK^n_{\OPT,j^n}})$ gives 
    \[ n (1 - e^{-rK^n_{\OPT,j^n}}) = r n K^n_{\OPT,j^n} + o(n(K^n_{\OPT,j^n})^2).\]
    By the contradiction supposition, $nK^n_{\OPT,j^n} \to 0$, so 
    \[ \lim_{n\to\infty}\hat\phi^n_\theta(K^n_{\OPT,j^n}) =  - (r + \lambda_\theta^* \nu_0^*)(-c^*). \]
   But this implies that the LHS of \eqref{eqn:discrete_threshold_condition} in the $n \to \infty$ limit approaches something strictly positive, which implies that \eqref{eqn:discrete_threshold_condition} must be violated for sufficiently large $n$.
   
    Consider the nonnegative sequence  $\{ j^n/n \}_n$. This must be bounded, otherwise since the sequence $j^n K^n_{\OPT,j^n} \le  t$ $\forall n$, which implies $n K^n_{\OPT,j^n} = (j^n K^n_{\OPT,j^n})/(j^n/n)\to 0$, which was just showed to be impossible. Hence, the sequence $\{ j^n/n \}_n$ must be bounded.

    Since $\{j^n/n\}_n$ is bounded, Bolzano-Weierstrass immediately implies that there exists a convergent subsequence, in which case let $L$ be its limit. 
    
    As showed prior, $\lim n K^n_{\OPT,j^n} = (j^n K^n_{\OPT,j^n})/(j^n/n) \in (0, t/L]$. Let $\lim\inf n K^n_{\OPT, j^n} = \ubar{t}/L$ for some $\ubar{t} \le t$, and focus on the subsequence yielding this $\liminf$. Then along this sequence as $n \to \infty$, 
     \begin{align*} \lim _{n \to \infty} \left[S^n_\theta(K_{\OPT,j^n}^n) \right]^{j^n} &= \lim_{n \to \infty} \left[ 1 -  \frac{\nu_0}{n}(1 - e^{-\lambda_\theta n K^n_{\OPT,j^n}}) \right]^{n (j^n/n)}, \\ 
    &= \exp\left( - \nu_0 L \left(1 - \exp\left(- \frac{\lambda_\theta \ubar{t}}{L} \right) \right) \right) = S_\theta\left( L, \ubar{t}\right).\end{align*}
    Further, the limit of $\hat{\phi}_\theta$ gives 
    \begin{align*} \lim_{n\to\infty} \hat\phi_\theta(K^n_{\OPT,j^n}) =& \lambda_\theta^*\nu_0^* e^{-\lambda \ubar{t}/L}\left[ \lim_{n \to \infty}\left( r n K^n_{OPT,j^n} + rnK^n_\OPT o(K^n_{\OPT,j^n})\right) + \nu_0 (1 - e^{-\lambda_\theta^* \ubar{t}/L}) \right] \\&- (r + \lambda_\theta^* \nu_0^* e^{-\lambda t/L} )(-c^* + \nu_0^*(1 - e^{-\lambda^*_\theta \ubar{t}/L})), \\
    =& \lambda_\theta^*\nu_0^* e^{-\lambda \ubar{t}/L}\left( r \frac{t}{L}+ \nu_0 (1 - e^{-\lambda_\theta^* \ubar{t}/L})\right) \\&- (r + \lambda_\theta^* \nu_0^* e^{-\lambda_\theta^* \ubar{t}/L} )(-c^* + \nu_0^*(1 - e^{-\lambda^*_\theta \ubar{t}/L})), \\
    =& - r \nu_0 \left( 1 - e^{-\lambda_\theta^* \ubar{t}/L}  - \frac{\lambda_\theta^* \ubar{t} }{L}e^{-\lambda_\theta \ubar{t}/L}\right) + rc + c\lambda_\theta^*\nu_0^* e^{-\lambda_\theta^*\ubar{t}/L}.
    \end{align*}
    Therefore, the product
    \[ \lim_{n\to\infty} \left[S^n_\theta(K_{\OPT,j^n}^n) \right]^{j^n}\hat\phi_\theta(K^n_{\OPT,j^n}) =  - \left[r \frac{\partial F_\theta(L,\ubar{t})}{\partial x} - rc  (1 - F_\theta(L,\ubar{t})) - c\frac{\partial F_\theta (L,\ubar{t})}{\partial \ubar{t}} \right], \]
    which implies that the limit of the LHS of \eqref{eqn:discrete_threshold_condition} precisely becomes the negative of the Euler-Lagrange expression for the continuum solution at $t$, which implies that $L = x^*(\ubar{t})$. By the same argument, take $\lim \sup n K^n_{\OPT, j^n+1} = \bar{t}/L$ for some $\bar{t} \ge t$, since $(j^n+1) K^n_{\OPT, j^n + 1} > t$ and $(j^n+1) /n \to L$ because $j^n /n \to L$. Taking a subsequence achieving the $\limsup$ implies that on this subsequence, since $(j^n+1)/n \to L$ as well, the same calculations as above imply that $L = x^*(\bar{t})$. As $x^*$ is weakly increasing, $L = x^*(\ubar{t}) \le x^*(t) \le x^*(\bar{t}) = L$, so $x^*(t) = L$. Because $L$ was an arbitrarily chosen accumulation point, \textit{any} convergent subsequence of $j^n/n \to x^*(t)$; note that $\hat{N}^n(t) = (j^n + 1) / n $, so since $1/n \to 0$, it immediately follows that along any convergent subsequence $j^n/n$, $\hat{N}^n(t) \to x^*(t)$. Thus, the $\lim\inf$ and $\lim\sup$ of $\hat{N}^n(t)$ both converge to $x^*(t)$, so the sequence $\hat{N}^n(t) \to x^*(t)$. Note that the convergence is uniform in $t$, because the convergence depended only on the fineness of the mesh induced by $j K^n_{\OPT, j}$ and not the choice of $t$, so $\sigma^n_\OPT$ converges in strategy to $x^*$.
\end{proof}

\subsection{Omitted Proofs in Section \ref{sec:agency}}
In the proofs of results, a few derivatives will be useful to characterize. To make notation cleaner, define the recurring terms: 
\begin{align}
    H_{t,\theta}(x,t) &= \nu_0\lambda_\theta e^{-\lambda_\theta t/x} \label{eqn:ht}, \\
    H_{x,\theta}(x,t) &= \nu_0 \left(1 - e^{-\lambda_\theta t/x} - \frac{\lambda_\theta t}{x} e^{-\lambda_\theta t/x} \right) \label{eqn:hx}
\end{align}
The relevant derivatives are:
\begin{align}
    F_t(x,t) &=\mathbb{E}_\theta\left[ H_{t,\theta}(x,t)S_\theta(x,t)\right], \label{eqn:dFdt} \\
    F_x(x,t) &= \mathbb{E}_\theta\left[ H_{x,\theta}(x,t) S_\theta(x,t)\right] .\label{eqn:dFdx} \\
    F_{xx}(x,t) &= \mathbb{E}_\theta\left[  \left(- H_{t,\theta}(x,t)\frac{\lambda_\theta t^2}{x^3} - H_{x,\theta}(x,t)^2 \right)S_\theta(x,t)  \right] \le 0\label{eqn:d2Fdx2} \\
    F_{tt}(x,t) &= \mathbb{E}_\theta \left[ - \left( \frac{\lambda_\theta}{x}  H_{t,\theta}(x,t)+ H_{t,\theta}(x,t)^2\right) S_\theta(x,t) \right] \le 0 \label{eqn:d2Ft2} \\
    F_{xt}(x,t) &= \mathbb{E}_\theta\left[ H_{t,\theta}(x,t)\left( \lambda_\theta \frac{t}{x^2} -  H_{x,\theta}(x,t)\right) S_\theta(x,t) \right]\label{eqn:d2Fdxdt}
\end{align}
Note that $H_{t,\theta}$ and $H_{x,\theta}$ are the hazard rates in $t, x$ respectively.

\begin{proof}[Proof of Proposition \ref{prop:static_share}]
By the same Euler-Lagrange calculation as before, define the condition
\[ H_t(x,\alpha) := r \left(\frac{\partial F(x,t)}{\partial x}\alpha - (1-F(x,t))c\right) - c \frac{\partial F(x,t)}{\partial t} . \]
Then the Euler-Lagrange calculation from before implies that the best-response $x(t;\alpha)$ satisfies:
\[ H_t(x(t;\alpha),\alpha) = 0. \]
To derive the comparative static, note that since $H$ is affine in $\alpha$, 
\begin{align*} H_t(x(t;\alpha),\alpha) &= \alpha \frac{\partial H_t(x(t;\alpha),\alpha)}{\partial \alpha} - rc(1-F(x(t;\alpha),t)) - c \frac{\partial F(x(t;\alpha),t)}{\partial t} = 0 \\
\implies \frac{\partial H_t(x(t;\alpha),\alpha)}{\partial \alpha} &=\frac{1}{\alpha}\left(rc(1-F(x(t;\alpha),t)) + c\frac{\partial F(x(t;\alpha),t)}{\partial t} \right) >0,\end{align*}
as it is straightforward from \eqref{eqn:dFdt} that $\partial F(x,t)/\partial t > 0$ for any $x,t$. Note that this also implies that 
\[ \frac{\partial F(x(t;\alpha),t)}{\partial x} = \frac{1}{r}\frac{\partial H_t(x(t;\alpha),\alpha)}{\partial \alpha} > 0 \]
Now, by straightforward computation, the partial derivative $\partial H_t/\partial x$ is
\begin{align*}
&\mathbb{E}_\theta\left[ r \alpha \nu_0 \left( - \frac{\lambda_\theta ^2 t^2}{x^3}\right)e^{-\lambda_\theta t/x} S_\theta(x,t)- c \nu_0 \frac{\lambda_\theta^2 t}{x^2}e^{-\lambda_\theta t/x}S_\theta(x,t)   \right.\\
    &\left.- \nu_0 \left(1 - e^{-\lambda_\theta t/x} - \frac{\lambda_\theta t}{x} e^{-\lambda_\theta t/x} \right)\left[r\alpha \nu_0\left(1 - e^{-\lambda_\theta t/x} - \frac{\lambda_\theta t}{x}e^{-\lambda_\theta t/x} \right) - rc - c\lambda_\theta \nu_0 e^{-\lambda_\theta t/x} \right] S_\theta(x,t) \right]
\end{align*}
It is clear that the  terms in the first line are all negative. The second line is slightly more ambiguous, since the middle bracketed expression may be positive. Note that since $H(x,\alpha) = 0$, 
\begin{align*}
    &\mathbb{E}_\theta\left[ \left(r\alpha \nu_0\left(1 - e^{-\lambda_\theta t/x} - \frac{\lambda_\theta t}{x}e^{-\lambda_\theta t/x} \right) - rc - c\lambda_\theta \nu_0 e^{-\lambda_\theta t/x} \right) S_\theta(x,t) \right] = 0.
\end{align*}
It is straightforward to check that the expression
\[ \phi_\theta(x) = r\alpha \nu_0\left(1 - e^{-\lambda_\theta t/x} - \frac{\lambda_\theta t}{x}e^{-\lambda_\theta t/x} \right) - rc - c\lambda_\theta \nu_0 e^{-\lambda_\theta t/x}\]
is decreasing in $x$, and that for $x_E$ such that $\phi_E(x_E) = 0$, $\phi_H(x_E) < 0$; thus, at the best response $x(t;\alpha)$, it must be that $\phi_E(x(t;\alpha))S_E(x(t;\alpha)) > 0 > \phi_H(x(t;\alpha))S_H(x(t;\alpha))$. Further, since 
\[ 1 - e^{-\lambda_\theta t/x} - \frac{\lambda_\theta t}{x}e^{-\lambda_\theta t/x} \]
is increasing in $\lambda_\theta$, it follows that at $x(t;\alpha)$,
\[ \mathbb{E}_\theta\left[\left(1 - e^{-\lambda_\theta t/x} - \frac{\lambda_\theta t}{x} e^{-\lambda_\theta t/x} \right) \phi_\theta (x) S_\theta(x,t)\right] > 0, \]
which implies that 
\[ \mathbb{E}_\theta\left[- \nu_0 \left(1 - e^{-\lambda_\theta t/x(t;\alpha)} - \frac{\lambda_\theta t}{x(t;\alpha)} e^{-\lambda_\theta t/x(t;\alpha)} \right)\phi_\theta(x(t;\alpha)) S_\theta(x(t;\alpha),t) \right] < 0. \]
Thus $\partial H_t(x(t;\alpha),\alpha)/ \partial x < 0$, so $\partial x(t;\alpha)/\partial\alpha > 0$ by the implicit function theorem. The principal's objective then is 
\[ \max_\alpha (1-\alpha)\int_0^\infty e^{-rt} \dd F(x(t;\alpha),t) = \max_\alpha (1-\alpha)r \int_0^\infty e^{-rt} F(x(t;\alpha),t)\dd{t} \]
The principal gets zero profit from setting $\alpha = 1$, 
so the principal always chooses some $\alpha < 1$; since the social planner solution $x_{FB} = x(\cdot;1)$, the comparative static immediately follows.
\end{proof}

\subsubsection{Proof of Theorem \ref{thm:dynamic_contract}}
\begin{proof} 
Write $\Phi=F(1-I)$, and denote the agent's and principal's payoffs by $U_A$ and $U_P$. The derivative calculation in the proof of Proposition \ref{prop:static_share}, applied to $G(x,a)=rF_x(x,t)(a-I(x,t))$ at $a=I(x,t)$, gives $G_x<0$ and hence $I_x>0$. Since $\Phi(0+,t)=0$, a positive pointwise maximizer must have $\Phi(x^*(t),t)>0$. Thus $I(x^*(t),t)\in(0,1)$, and \eqref{eqn:dyn_a_alpha} defines a feasible share $\alpha^*$ satisfying $\alpha^*-\dot\alpha^*/r=I(x^*(t),t)$.

I now show that the agent's unique global best response to $\alpha^*$ is $x^*$. Define the cost potentials
\[
 K(y,t):=c\int_0^y S(z,t)\dd{z},\qquad C(y,t):=rK(y,t)-K_t(y,t).
\]
Every admissible path satisfies $y(t)\leq L_y t$ for some finite $L_y$, because $y(t)/t$ is nonincreasing. Consequently the boundary terms below vanish: $0\leq K(y(t),t)\leq cL_y t$ and $|K_t(y(t),t)|\leq c\nu_0\lambda_E L_y t$. Integrating the agent's expected bonus and expected exploration cost by parts gives
\[
 U_A[y;\alpha^*]=\int_0^\infty e^{-rt}\bigl[rI(x^*(t),t)F(y(t),t)-C(y(t),t)\bigr]\dd{t}.
\]
The derivative of the bracketed expression in $y$ is
\[
 rF_x(y,t)\bigl[I(x^*(t),t)-I(y,t)\bigr].
\]
Since $I_x>0$, each running term has its unique global maximum at $y=x^*(t)$. The admissible path $x^*$ attains all these maxima simultaneously, proving that it is the unique global best response, including against deviations with binding control constraints.

Finally, it remains to show optimality within the interior-response class for the principal. Let $\tilde\alpha$ be any competing measurable contract with an interior global response $y$. On each compact subinterval of $(0,\infty)$, apply the maximum principle with control $v=t\dot y/y\in(0,1)$.\footnote{The problem is normal: an abnormal interior Hamiltonian would force its costate to vanish identically.} The switching condition gives $p=e^{-rt}(cS-\tilde\alpha F_x)$, so $\tilde\alpha$ has a locally absolutely continuous representative. The costate equation then yields
\begin{equation}
 \frac{\dd{}}{\dd{t}}\bigl(e^{-rt}\tilde\alpha(t)\bigr)
 =-r e^{-rt}I(y(t),t)\quad\text{a.e.}
 \label{eqn:dyn_a_EL}
\end{equation}
Boundedness of $\tilde\alpha$ therefore implies
\begin{equation}
 \tilde\alpha(t)=\int_t^\infty re^{-r(s-t)}I(y(s),s)\dd{s}.
 \label{eqn:dyn_a_EL_diffeq}
\end{equation}
In particular, $\int_0^\infty re^{-rt}I(y(t),t)\dd{t}=\tilde\alpha(0+)\leq1$. Integration by parts gives
\[
 U_P[\tilde\alpha]
 =\int_0^\infty re^{-rt}\Phi(y(t),t)\dd{t}
 \leq\int_0^\infty re^{-rt}\Phi(x^*(t),t)\dd{t}
 =U_P[\alpha^*].
\]
This proves optimality within the interior-response class.

The pointwise first-order condition $\Phi_x(x^*(t),t)=0$ can be written as
\begin{equation}
 rF_x-c\bigl(r(1-F)+F_t\bigr)
 +\frac{F}{F_x}\left[
 \frac{cF_{xx}}{F_x}\bigl(r(1-F)+F_t\bigr)+crF_x-cF_{xt}
 \right]=0,
 \label{eqn:dyn_a_lawx}
\end{equation}
where all terms are evaluated at $(x^*(t),t)$.
\end{proof}

\begin{proof}[Proof of Proposition \ref{prop:no_commit}]
I will construct and verify the candidate equilibrium. The principal can condition their actions on their history of shares offered, but it is payoff-relevant only through the principal's conjecture on the exploration state $\hat{x}$, which is correct in equilibrium. So strategies for the principal map conjectured states $(\hat{x}, t) \in \mathbb{R}_+^2 $ to shares $\alpha \in [0,1]$. The agent, on the other hand, observes the state and the share offered by the principal at that state, and thus their strategy maps triples $(\alpha, x, t) \in [0,1]\times\mathbb{R}_+^2$ to an effort allocation choice of $\dot{x} \in [0,x/t]$. 

In order to define the equilibrium, I first introduce some notation. 
As a minor abuse of notation, for the scope of this proof, let
\[ \alpha(d) := \frac{rc + c\nu_0 \lambda e^{-\lambda d}}{r \nu_0 (1 - e^{-\lambda d} - \lambda d e^{-\lambda d})}, \]
which is the share of the profit that corresponds to the agent choosing depth $d$ according to the incentive compatibility constraint from the agent's Euler-Lagrange equation, and let $d(\alpha)$ be defined as the solution to 
\begin{equation} r \alpha \nu_0 \left(1 - e^{-\lambda d} - \lambda d e^{-\lambda d} \right) - rc - c \nu_0\lambda e^{-\lambda d} = 0, \label{eqn:static_a_EL} \end{equation} 
which exists for  $\alpha > c/\nu_0$. 
Note that $d(\alpha)$ is constructed as the inverse of $\alpha(d)$ on the domain $d > 0$, as it is straightforward to verify that $\alpha$ is strictly decreasing on $d > 0$. 

Let the value associated with a depth be 
\[ W(d) := \frac{\nu_0(1-e^{-\lambda d})}{rd + \nu_0(1-e^{-\lambda d})}.  \]
Intuitively, if the agent takes fixed depth $d$, this is the expected discount factor for how long it takes the solution to arrive. As a useful identity, for the hazard rates $H_t$, $H_x$ defined before and $d = t/x$, note that 
\begin{equation} \left( H_t + \frac{H_x}{d} \right)(1 - W(d)) = r W(d).
\label{eqn:mpe_vp_identity}
\end{equation}

Define $\alpha_\nc$ as the maximizer:
\begin{gather*} \max_{\alpha \in [0,1]} (1-\alpha) W(d(\alpha)).
\end{gather*} 
It is straightforward to check that the value is $0$ at $\alpha = 1$ and approaches zero in the limit when $\alpha \to c/\nu_0$ since $d \to \infty$, and it is easily verified that the objective of the above maximization is concave. I will use $d_\nc$ to denote $d(\alpha_\nc)$. It will be useful to define
\[ V(d):= (1 - \alpha(d)) W(d), \]
and note that concavity of $(1-\alpha) W(d(\alpha))$ and the fact that $d$ is monotonic in $\alpha$ implies that the transformed objective $V$ is quasiconcave in $d$.

Consider the following candidate equilibrium. The principal's strategy is:
\[ \sigma_P(\hat{x},t) = \begin{cases} 
0 & t/\hat{x} < \hat{d}(\hat{x}) \textnormal{ and } t>0, \\
\alpha_\nc & t = 0, \\
\min[\alpha(t/\hat{x}),1] & \textnormal{otherwise}.
\end{cases}\]where 
\begin{equation} \hat{d}(x) := \arg\max_{z \ge 0} \left \{ \int_0^{{x}z} e^{-rs} F_t({x},s)  \dd s + e^{-r{x}z}S(x, xz) V(z) \right\} .\footnote{It is straightforward to verify that (1) the objective of the maximization is strictly quasiconcave in $z$, so the maximizer is unique, and (2) the maximizer is less than $d_\nc$. Explicit calculations can be found in the supplemental appendix.}\label{eqn:hatd}\end{equation}
Note that when $t/\hat{x} = d_\nc$, the definition of $d_\nc$ implies that the principal offers $\alpha_\nc$. 
The principal's conjectured agent responses (changes in the state) in these three cases are then $\dot{x} = 0$ in the first case and $\dot{x} = \hat{x}/t$ in the second.
The agent's strategy is 
\[ \sigma_A(\alpha, x, t) = \begin{cases}
    x/t & \alpha \ge \alpha(t/x)  \textnormal{ and } t > 0, \\
    1/d(\alpha) & t = 0, \\
    0 & \textnormal{else.}
\end{cases}\]

To show that this is an MPE, first observe that the principal's conjectured agent response is correct on-path. I now verify that the strategies are mutually best-responding via an explicit value function construction. For the principal, define the candidate value function
\begin{align*}
     V_P(\hat{x},t)&:= \begin{cases}
         V(d_\nc) & t = 0, \\
         V(d)& t \ge \hat x\hat{d}(\hat x), \\
         \frac{1}{S(\hat{x},t)}\int_t^{\hat{x}\hat{d}(\hat{x})} e^{-r(s-t)} F_t(\hat{x},s)  \dd s + e^{-r(\hat{x}\hat{d}(\hat{x})-t)}\frac{S(\hat{x}, \hat{x} \hat{d}(\hat{x}))}{S(\hat{x},t)} V(\hat{d}(\hat{x}))& t < \hat x\hat{d}(\hat x),
     \end{cases} 
\end{align*}
It is straightforward to confirm value matching at the boundary. Further, smooth pasting also holds at the boundary by the definition of $\hat{d}$ using the envelope theorem. So the candidate value function is $\mathcal{C}^1$, the transversality condition holds since $V_P$ is bounded above by $1$, and so the standard deterministic verification theorem\footnote{See \cite{bcd1997}, Chapter 3, Theorem 2.12.} applies. So it suffices to verify that this candidate value function satisfies the following HJB (suppressing the $\hat{x},t$ arguments for cleanliness, where $\sigma_A$ is also understood to depend on the state):
\begin{equation}
rV_P = \max_{\alpha\in[0,1]}\left\{\frac{\partial V_P}{\partial t}+ \frac{\partial V_P}{\partial \hat{x}} \,\sigma_A(\alpha)+\big(H_t+H_x \sigma_A(\alpha)\big)\big(1-\alpha-V_P\big)\right\}.\footnote{Technically, the verification theorem results apply to the rearranged control problem where the principal directly chooses the agent's action, and incurs cost according to the appropriate $\alpha$ necessary to induce that action. This rearrangement is equivalent to this HJB.}
\label{eq:mpe_principal_hjb}
\end{equation}
I verify the HJB region-by-region first when $t> 0$ and $t/\hat{x} \ge \hat{d}(\hat{x})$, then for $t > 0$ and $t/\hat{x} < \hat{d}(\hat{x})$, and finally at the initial condition $t = 0$. In the following verification, I will write $\hat{x}$ as $x$ for simplicity, where it is understood here that the state is the principal's conjectured breadth. 

\textbf{Region 1:} Fix a state in the region where $t/x \ge \hat{d}(x)$. Note that the (candidate) value function is only a function of the state $x, t$ via the depth $d$. First consider $\alpha \ge \alpha(d)$. For such an $\alpha$, note that the agent response is $\sigma_A(\alpha) = x/t = 1/d$, and hence maintains the same depth; so on this region, 
\[ \frac{\partial V_P}{\partial t} + \frac{\partial V}{\partial x} \sigma_A = 0. \]
Hence the HJB maximand consists of just the remaining terms
\[ \left( H_t + \frac{H_x}{d} \right)(1 - \alpha - V_p). \]
The maximand is clearly decreasing in $\alpha$, so the only candidate optimal $\alpha$ on this region is $\alpha(d)$. 
At $\alpha = \alpha(d)$, the maximand becomes 
\[ \left( H_t + \frac{H_x}{d} \right)(1 - \alpha - V_p) = \left( H_t + \frac{H_x}{d} \right)(1 - \alpha(d) )(1 - W(d)) = r (1-\alpha(d)) W(d) = r V_P, \]
where the second equality used \eqref{eqn:mpe_vp_identity}. Thus, the HJB holds with equality at $\alpha(d)$.

To finish this case, it suffices to show that it cannot be optimal to set $\alpha < \alpha(d)$. 
For $\alpha < \alpha(d)$, the agent response is $\sigma_A(\alpha) = 0$, so the HJB maximand is $ \frac{\partial V_P}{\partial t} + H_t (1 - \alpha - V_P) $, which is decreasing in $\alpha$. So the only other candidate optimum is $\alpha = 0$, when the maximum attains value
\begin{align*} \frac{\partial V_P}{\partial t} + H_t (1  - V_P) &= V'(d)\left( \frac{1}{x} \right) + H_t(1-V_P) \le rV_P.
\end{align*}
where the first inequality follows because for all $d \ge \hat{d}(x)$, since the objective defining $\hat{d}$ in \eqref{eqn:hatd} is quasiconcave, the objective is decreasing at $d$; taking the derivative of the objective in \eqref{eqn:hatd} and setting it to be negative immediately gives the inequality.\footnote{For the details of the algebraic rearranging, see the Supplemental Appendix.}
Hence it cannot be optimal to set $\alpha < \alpha(d)$, so the candidate $V_P$ satisfies the HJB in this region. 

\textbf{Region 2:} First, by taking the partial derivative of $V_P$ in $t$, 
\begin{align*} \frac{\partial V_P}{\partial t} =& \frac{1}{S}\left( - F_t(x,t) \right) + r V_P + H_t V_P, \\
=&- H_t + r V_P + H_t V_P, \\
\implies r V_P=&\frac{\partial V_P}{\partial t} + H_t (1 - V_P) .
\end{align*}
 Thus, the value of the HJB maximand attains $r V_P$ when $\alpha = 0$, since $\sigma_A(0) = 0$. It then suffices to show that this is the maximal value that can be attained for any $\alpha > 0$. Clearly, for $\alpha < \alpha(d)$, $\sigma_A(\alpha) = 0$ still, and thus results in the maximand being $\frac{\partial V_P}{\partial t} + H_t (1 - \alpha - V_P)  < \frac{\partial V_P}{\partial t} + H_t (1 - V_P)  = r V_P.$  For $\alpha \ge \alpha(d)$, note that the HJB maximand
\[ \frac{\partial V_P}{\partial t} + \frac{\partial V_P}{\partial x} \frac{1}{d}  + (H_t + H_x/d) (1 - \alpha - V_P) \]
is decreasing in $\alpha$, so the only other candidate maximum is $\alpha = \alpha(d)$. However, using the above simplification, at $\alpha(d)$ the maximand becomes
\[ r V_P - \alpha(d) H_t  + \frac{\partial V_P}{\partial x}\frac{1}{d} + \frac{H_x}{d}(1 - \alpha(d) - V_P).\]
The partial derivative $\partial V_P/\partial x$ simplifies using the envelope condition due to the definition of $\hat{d}$, so with some rearranging,
\begin{align*} &\frac{\partial V_P}{\partial x}+H_x(1-V_P)  \\
& \quad =\frac{1}{S(x,t)} \left[ r\int_t^{\hat{d}(x) x} e^{-r(s-t)} F_x(x,s)\dd{s} + e^{-r(\hat{d}(x) x-t)} S(x,\hat{d}(x) x) \,\alpha(\hat{d}(x)) \,\nu_0 \left(1-e^{-\lambda \hat{d}(x)}\right) \right]. \end{align*}
To show that $\alpha(d)$ cannot be optimal, it suffices to show that the HJB maximand is less than $r V_P$, which is equivalent to $ \partial V_P/\partial x + H_x(1- V_P) \le \alpha(d) (H_x + H_t d )$, or that 
\begin{align*} &\frac{1}{S(x,t)} \left[ r\int_t^{\hat{d}(x) x} e^{-r(s-t)} F_x(x,s)\dd{s} + e^{-r(\hat{d}(x) x-t)} S(x,\hat{d}(x) x) \,\alpha(\hat{d}(x)) \,\nu_0 \left(1-e^{-\lambda \hat{d}(x)}\right) \right] \\
& \qquad \le \alpha(d) \left( H_td  + H_x\right). \end{align*}
Multiplying both sides of the inequality by $S$ gives 
\[ r\int_t^{\hat{d}(x) x} e^{-r(s-t)} F_x(x,s)\dd{s} + e^{-r(\hat{d}(x) x-t)} S(x,\hat{d}(x) x) \,\alpha(\hat{d}(x)) \,\nu_0 \left(1-e^{-\lambda \hat{d}(x)}\right)  \le \alpha(d) S(x,t) \nu_0 (1 - e^{-\lambda d}).\]
To show that this inequality holds, consider the ``gap'' quantity, 
\begin{align*} G(x,d) :=& \alpha(d) S(x,dx) \nu_0 (1-e^{-\lambda d}) - e^{-r(\hat{d} x-dx)} S(x,\hat{d} x) \,\alpha(\hat{d}) \,\nu_0 \left(1-e^{-\lambda \hat{d}}\right)\\ 
& - r\int_{dx}^{\hat{d} x} e^{-r(s-dx)} F_x(x,s)\,ds \end{align*}
defined for all $d \le \hat{d}$. It suffices to show that the gap is nonnegative for all $x$, $d \le \hat{d}$. By taking the first line terms as the endpoints from integration by parts and changing the integration variable to be over depth variable $z$, one can rewrite $G(x,d)$ as 
\begin{equation} \int_d^{\hat{d}}  e^{-rx(z-d)} S(x,xz) \left[ -\frac{\dd{}}{\dd d}[\alpha(d) \nu_0(1-e^{-\lambda d}) ]_{d=z} + x\left( (r + \nu_0\lambda e^{-\lambda z})\alpha(z) \nu_0 (1-e^{-\lambda z}) - r H_x(z)   \right)  \right] \dd z. \label{eqn:mpe_gap_form} \end{equation}
Observe that by definition of $W$, 
\[V'(z) (rz + \nu_0(1-e^{-\lambda z})) = -\frac{\dd{}}{\dd d}[\alpha(d) \nu_0(1-e^{-\lambda d}) ]_{d=z} + H_t - (r + H_t)V(z), \]
and that straightforward algebraic rearrangement gives that 
\begin{align*} H_t(z) (1 - V(z)) - rV(z) =\frac{(r + H_t(z))  \alpha(z) \nu_0 (1-e^{-\lambda z}) - rH_x(z)}{rz + \nu_0(1-e^{-\lambda z})}. \end{align*}
Thus, whenever $H_t(z) (1 - V(z)) - rV(z) \ge 0$, the second term in the bracket of \eqref{eqn:mpe_gap_form} is positive, and since $\alpha$ and $(1-e^{-\lambda d})$ are both decreasing in $d$, the entire bracket is positive. In the other case, using these two identities to plug into the expression for $G$, 
\begin{align*}
    G(x,d) =& \int_d^{\hat{d}}  e^{-rx(z-d)} S(x,xz) \left[ V'(z) + x (H_t(z) (1 - V(z)) - rV(z))\right](rz + \nu_0(1-e^{-\lambda z}))  \dd z \\
    & + \int_d^{\hat{d}}  e^{-rx(z-d)} S(x,xz) \left[ r V(z) - H_t(z)(1- V(z)) \right]\dd{z}.
\end{align*}
The first term is positive because for all $z \le \hat{d}$, the definition of $\hat{d}$ implies that at $z$ the derivative of the maximizing expression for $\hat{d}$ is positive, so the bracketed part is positive. The second term is positive by the case assumption. Hence $G(x,d) \ge 0$, which implies that the HJB maximand is maximized for $\alpha = 0$. 

\textbf{Initial Condition}: Lastly, observe that any initial depth $d_0$ that is maintained forever on-path yields payoff $V(d_0)$, and by definition of $d_\nc$, $d_\nc$ maximizes that payoff. Further, since the option of delaying reaching $d_\nc$ is worse than immediately starting at $d_\nc$ since $\hat{d}(x) \le d_\nc$ for all $x$, it follows that the best initial condition that the principal can obtain is precisely $d_\nc$. 

To finish the equilibrium construction, it remains to verify that the agent's policy is also a best response. Once again, it will be easiest to construct the agent's explicit value function at each $(x,t)$. Define
\[ V_A(x,t) := \begin{cases}
    \alpha(d_\nc)W(d_\nc) - \frac{c}{rd_\nc+\nu_0(1-e^{-\lambda d_\nc})} & t = 0, \\
    \alpha(d)W(d) - \frac{c}{rd+\nu_0(1-e^{-\lambda d})} & t \ge x\hat{d}(x), \\
    e^{-r(x\hat{d}(x) - t)}\frac{S(x, x\hat{d}(x))}{S(x,t)}\left(\alpha(\hat{d}(x))W(\hat{d}(x))- \frac{c}{r\hat{d}(x)+\nu_0(1-e^{-\lambda \hat{d}(x)})} \right) &  0 < t < x\hat{d}(x).
\end{cases}\]
The HJB for the agent is then to select
\[ r V_A = \max_{u \in [0, x/t]} \left\{- cu + (H_t + u H_x)(\alpha - V_A) + \frac{\partial V_A}{\partial x}\sigma_A + \frac{\partial V_A}{\partial t} \right\}.\]
Note that by assumption, the agent's path is unobservable, so the state evolution partial derivative in $V_A$ evolves according to the principal's conjectured equilibrium agent strategy which is $\sigma_A$.
I once again verify the value function in each region; standard verification theorem results apply \footnote{Note the value function for the agent is continuous and locally Lipschitz, but not differentiable at the boundary. However, standard verification methods still apply because $\hat{d}$ decreases in $x$, and both $x$ and $d$ are monotone, so any admissible trajectory can cross regions at most once; since the value function is verified at each region, continuous at the boundary, and satisfies transversality (as the value function is bounded), a solution to the HJB is sufficient for optimality. See \cite{fr1975}, Chapter IV Theorem 7.1.}  Note that the objective is linear in the control $u$, with coefficient 
\[ - c + H_x \left(\alpha - \alpha(d)W(d) +  \frac{c}{rd+\nu_0(1-e^{-\lambda d})} \right).\]
I claim this coefficient is zero at $\alpha = \alpha(d)$; linearity and the fact that $H_x > 0$ immediately implies the bang-bang policy switching from $u = 0$ to $u = x/t$ at $\alpha(d)$ is optimal. At $\alpha= \alpha(d)$,
\begin{align*}& -c + H_x \alpha(d) (1 - W(d)) +  \frac{cH_x}{rd+\nu_0(1-e^{-\lambda d})}   \\
& \qquad = - c + (r + H_t)c \left( \frac{d}{rd + \nu_0 (1-e^{-\lambda d})} \right) + H_x \frac{c}{rd+\nu_0(1-e^{-\lambda d})}   \\
&\qquad = -c  \frac{ \nu_0 (1 - e^{-\lambda d} - \lambda de^{-\lambda d})   }{rd + \nu_0(1-e^{-\lambda d})} + \frac{cH_x}{rd+\nu_0(1-e^{-\lambda d})} = 0.
\end{align*}
Hence, the agent's policy to choose $u = x/t$ when $\alpha \ge \alpha(d)$ and $0$ otherwise is optimal for the agent. So it remains to verify the value function solves the HJB at the optimal policy. Before doing so, it will be useful to note the identity that
\begin{align}
    \alpha(d)W(d) - \frac{c}{rd + \nu_0(1-e^{-\lambda d})} &= c \left( \frac{r+H_t}{rH_x} \right) \frac{H_x + d H_t}{rd + H_x + dH_t} - \frac{c}{rd + H_x + dH_t} \notag  \\
    &=  \left(\frac{c}{rH_x(rd + H_x + dH_t)} \right)\left( rH_x + H_t H_x + rd H_t + dH_t^2 - r H_x\right) \notag \\
    &= \left(\frac{c}{rH_x(rd + H_x + dH_t)} \right)H_t \left( H_x + rd  + dH_t\right) \notag \\
    &= \frac{c H_t }{r H_x} \label{eqn:mpe_va_identity}
\end{align}

\textbf{Region 1}: When $t \ge x\hat{d}(x)$, the principal offers $\alpha(d)$ and the optimal agent response is $\sigma_A = x/t = 1/d$, so since $V_A$ only depends on $x, t$ via $d$, it follows that the continuation partial derivative terms sum to zero. Hence, to verify the HJB when $\alpha = \alpha(d)$, recall that by the optimality argument the coefficient of $u$ is zero, so 
\begin{align*}
&-\frac{c}{d} + \left(H_t + \frac{H_x}{d}\right) (\alpha(d) - V_A) + \frac{\partial V_A}{\partial x}\sigma_A + \frac{\partial V_A}{\partial t} \\
&\quad =  H_t (\alpha(d) - V_A) = H_t \left( c\frac{r + H_t}{rH_x} - \frac{cH_t}{rH_x} \right) = \frac{rc H_t }{r H_x} = r V_A.
\end{align*}
So $V_A$ solves the HJB in this region.

\textbf{Region 2}: In this region, the offered share is 0 and so the optimal response $u = 0$, so maximand part of the HJB becomes 
\[ - H_t V_A + \frac{\partial V_A}{\partial t} = - H_t V_A + r V_A + H_t V_A = r V_A. \]

Finally, at the initial condition, the principal offers $\alpha_\nc$ forever, so by construction $d_\nc$ is the optimal agent depth.
\end{proof}

\begin{proof}[Proof of Proposition \ref{prop:dec_share}]
Let $\alpha$ be an optimal contract with an interior response $x=x_C$, and write $\Phi=F(1-I)$. The interior-response necessity and global verification arguments in the proof of Theorem \ref{thm:dynamic_contract}, neither of which requires pointwise maximization by the principal, give
\[
 \alpha(t)=\int_t^\infty re^{-r(s-t)}I(x(s),s)\dd{s},
 \qquad
 U_P[\alpha]=\int_0^\infty re^{-rt}\Phi(x(t),t)\dd{t}.
\]
Moreover, any admissible path with a feasible reconstructed share is that share's unique global agent response.

I first check that the share bounds do not obstruct local path variations. Since $I_x>0$,
\[
 \Phi_x=(1-I)F_x-FI_x.
\]
A compactly supported decrease of $\log x$ preserves interior admissibility and lowers the reconstructed share.\footnote{Writing $w=\log x$, interiority gives $0<\dot w<1/t$. On any compact positive-time interval choose a nonnegative bump $\eta$ with $|\dot\eta|\leq\min\{\dot w,1/t-\dot w\}$; then $w\pm\varepsilon\eta$ remain feasible for $0<\varepsilon<1$.}
If $\Phi_x<0$ on an interval, such a decrease would strictly raise profit. Thus $\Phi_x\geq0$ along the optimal path, which forces $I(x(t),t)<1$. The reconstruction then gives $0<\alpha(t)<1$ and $0<\alpha(0+)<1$. Compactly supported path variations change the share only on a bounded initial time interval, where its bounds are uniformly slack. Small variations of either sign are therefore feasible and globally implemented. Optimality consequently requires
\[
 \Phi_x(x(t),t)=0 \qquad (t>0),
\]
since otherwise a local variation of the appropriate sign would increase profit. Multiplying this condition by $r$ and expanding $I$ gives \eqref{eqn:dyn_a_lawx}. Note that this establishes the necessary state equation without assuming feasibility of a pointwise maximizing candidate.

Throughout the proof let $d := t/x$, $z := \lambda d$, and
\[v(z) := 1 - e^{-z} - ze^{-z},\qquad \kappa(z) := \nu_0\bigl(1 - e^{-z}\bigr),\qquad y(x,t) := \frac{F(x,t)}{x\bigl(1-F(x,t)\bigr)} .\]
Both hazard rates depend on the state only through the depth: $H_t = \nu_0\lambda e^{-z}$ and
$H_x = \nu_0 v(z)$. Define
\[\Psi(z) := rH_x(z) - rc - cH_t(z),\qquad\beta(z) := \frac{rc}{\nu_0}\,\frac{z^2e^{-z}}{v(z)^2}+ c\lambda\,\frac{z^2e^{-2z}}{v(z)^2}+ c\lambda\,\frac{ze^{-z}}{v(z)} .\]
Dividing the law of motion \eqref{eqn:dyn_a_lawx} through by $1-F$ and substituting the forms \eqref{eqn:dFdt}-\eqref{eqn:d2Fdxdt} gives
\begin{equation} rH_x - rc - cH_t + \frac{F}{x(1-F)}\left[  - \frac{H_t}{H_x^2} \frac{\lambda t^2}{x^2}rc - \left(\frac{H_t^2}{H_x^2} \frac{\lambda t^2}{x^2}  + \frac{H_t}{ H_x} \frac{\lambda t}{x} \right)c \right] = 0.  \label{eqn:rewritten_law_x}\end{equation}
Using the notational shorthands, the law becomes
\begin{equation}\label{eqn:prop7_law}
  \Psi\bigl(z(t)\bigr) = y\bigl(x(t),t\bigr)\,\beta\bigl(z(t)\bigr) \implies \frac{\Psi(z(t))}{\beta(z(t))} = y(x(t),t).
\end{equation}
The proof proceeds by showing that $\Psi/\beta$ is strictly increasing on the induced trajectory and that $y$ is strictly
increasing along every admissible path.

\begin{lemma}\label{lem:beta}
$\beta$ is strictly positive and strictly decreasing on $(0,\infty)$, with $\beta(z)\to 0$ as
$z\to\infty$.
\end{lemma}

\begin{proof}
Positivity follows by inspection, since $v(z) > 0$ because $e^z > 1+z$. It suffices to
show that each of
\[g_1(z) := \frac{z^2e^{-z}}{v(z)^2},\qquad g_2(z) := \frac{z^2e^{-2z}}{v(z)^2},\qquad g_3(z) := \frac{ze^{-z}}{v(z)}\]
is strictly decreasing; the limits are immediate since $v(z)\to1$. Note $v'(z) = ze^{-z}$. Hence
\[ \frac{d}{dz}\log g_1(z) = \frac{2}{z} - 1 - \frac{2ze^{-z}}{v(z)}, \]
which is negative if and only if $(2-z)\,v(z) < 2z^2e^{-z}$, that is, after multiplying by $e^{z}$,
if and only if
\[  (2-z)\bigl(e^{z}-1-z\bigr) < 2z^2 .\]
Using a Taylor expansion on the left-hand side, the coefficient of $z^m$ in $(2-z)(e^{z}-1-z) = (2-z)\sum_{m\ge2} z^m/m!$ equals $1$ for
$m=2$ and $2/m! - 1/(m-1)! = (2-m)/m! \le 0$ for $m\ge3$. Therefore
$(2-z)(e^{z}-1-z)\le z^2 < 2z^2$ for every $z>0$, and $g_1$ is strictly decreasing.

Since $\frac{d}{dz}\log g_2(z) = \frac{d}{dz}\log g_1(z) - 1$, $g_2$ is strictly decreasing as well.

Finally, $\frac{d}{dz}\log g_3(z) = \frac{1}{z} - 1 - \frac{ze^{-z}}{v(z)}$, which is negative if and
only if $(1-z)(e^{z}-1-z) < z^2$. Invoking a Taylor bound again, the coefficient of $z^m$ in $(1-z)(e^{z}-1-z)$ equals $1/2$ for
$m=2$ and $(1-m)/m! < 0$ for $m\ge3$, so the left side is at most $z^2/2 < z^2$.
\end{proof}

\begin{lemma}\label{lem:Theta}
$\Psi$ is strictly increasing on $(0,\infty)$ with a unique root $z_*$, and $\Psi/\beta$ is
strictly increasing on $(z_*,\infty)$ with $\Psi/\beta \to \infty$ as $z\to\infty$.
\end{lemma}

\begin{proof}
$H_x'(z) = \nu_0 ze^{-z} > 0$ and $H_t'(z) = -\nu_0\lambda e^{-z} < 0$, so $\Psi' > 0$. Also
$\Psi(0^+) = -rc - c\nu_0\lambda < 0$ and $\Psi(z)\to r(\nu_0-c) > 0$, using $c<\nu_0$, so
$z_*$ exists and is unique. On $(z_*,\infty)$ the function $\Psi$ is strictly positive and strictly
increasing while $\beta$ is strictly positive and strictly decreasing by
Lemma~\ref{lem:beta}, so $\Psi/\beta$ is strictly increasing there, and the limit goes to
$r(\nu_0-c)/0^+ = \infty$.
\end{proof}

\begin{lemma}\label{lem:y}
Along any admissible path, $t\mapsto y(x(t),t)$ is strictly increasing. Moreover, if $z$ is bounded
away from zero and $x$ is unbounded, then $y\to\infty$.
\end{lemma}

\begin{proof}
Parametrize by $(x,z)$, so that $y = (e^{x\kappa(z)}-1)/x$. Then
\[
  \frac{\partial y}{\partial x}\bigg|_{z}
  = \frac{x\kappa(z)e^{x\kappa(z)} - \bigl(e^{x\kappa(z)}-1\bigr)}{x^{2}} > 0 ,
  \qquad
  \frac{\partial y}{\partial z}\bigg|_{x} = \kappa'(z)\,e^{x\kappa(z)} > 0 ,
\]
the first because $we^{w} > e^{w}-1$ for $w>0$ and the second because $\kappa' > 0$.

The admissibility constraint $\dot x \in [0, x/t]$ gives $\dot x \ge 0$ and
$\dot z = \lambda\dot d = \lambda(1-d\dot x)/x \ge 0$, and the two cannot vanish simultaneously: if
$\dot x = 0$ then $\dot z = \lambda/x > 0$. Hence
$\dot y = y_x\dot x + y_z\dot z > 0$.

For the second claim, if $z(t) \ge z_0 > 0$ then $\kappa(z)\ge\kappa(z_0) > 0$, so
$y \ge \bigl(e^{x\kappa(z_0)}-1\bigr)/x \to \infty$ as $x\to\infty$.
\end{proof}

Note that Lemma~\ref{lem:y} uses only admissibility, not optimality. I now establish the proposition.
First, I show that depth is strictly increasing. Since $y>0$ and $\beta>0$, \eqref{eqn:prop7_law} forces
$\Psi(z(t))>0$, so $z(t) > z_*$ for all $t>0$. By Lemma~\ref{lem:Theta}, $\Psi/\beta$ is strictly
increasing on $(z_*,\infty)$, and by Lemma~\ref{lem:y} the right side of \eqref{eqn:prop7_law} is
strictly increasing in $t$. Hence $z$, and therefore the depth $t/x^{C}$, is strictly increasing in
$t$.

Second, I show that the depth diverges in the asymptotic limit. The path $x^{C}$ is nondecreasing, so it either converges to a finite limit (and is bounded) or
diverges. If $x^{C}$ is bounded then $z = \lambda t/x^{C}\to\infty$ directly. If $x^{C}$ is unbounded
then, since $z$ is increasing and hence bounded below by $z(t_0)>0$, Lemma~\ref{lem:y} gives
$y\to\infty$, so $\Psi(z)/\beta(z)\to\infty$ by \eqref{eqn:prop7_law}, and since $\Psi/\beta$ is finite at every
finite $z$, this implies $z\to\infty$. In either case, $t/x^{C}\to\infty$.

By Proposition \ref{prop:no_commit} the no-commitment equilibrium path satisfies $x^{NC}(t) = t/d^{NC}$ exactly, so
\[
  \frac{x^{C}(t)}{x^{NC}(t)} = \frac{d^{NC}}{t/x^{C}(t)} \longrightarrow 0 ,
\]
which is the third claim of the proposition.

For the first two claims in the proposition, define the distortion as $\Delta(t) := - \frac{y(x(t),t) \beta(z(t))}{H_x(z(t))} < 0$. From \eqref{eqn:prop7_law}, 
\[\Delta(t) = \frac{c\bigl(r + H_t(z(t))\bigr)}{H_x(z(t))} - r \;\in\; (-r, 0) .\]
Since $H_t$ is strictly decreasing and $H_x$ strictly increasing in $z$, $\Delta$ is strictly
decreasing in $z$, hence strictly decreasing in $t$ by Step 1. By Step 2 and
$H_t(z)\to0$, $H_x(z)\to\nu_0$,
\[  \Delta(t) \;\longrightarrow\; \frac{rc}{\nu_0} - r .\]
Using the definition of the distortion and substituting $\bigl((1-F)rc + cF_t\bigr)/F_x = r+\Delta$ into \eqref{eqn:dyn_a_alpha},
\[\alpha(t) = e^{rt}\int_{t}^{\infty}e^{-rs}\bigl(r + \Delta(s)\bigr)\dd{s} = 1 + \int_{0}^{\infty}e^{-rs}\,\Delta\bigl(t+s\bigr)\,\dd{s} .\]
The integral converges absolutely because $|\Delta| < r$ as established earlier, and $\alpha(t)\in(0,1)$, so the
share is feasible. For each fixed $u$ the integrand is strictly decreasing in $t$ by Step 3, so
$\alpha$ is strictly decreasing. Since $\Delta(t) \to rc/\nu_0 - r$ and by bounded convergence,
\[\alpha(t) \;\longrightarrow\; 1 + \frac1r\Bigl(\frac{rc}{\nu_0} - r\Bigr) = \frac{c}{\nu_0} ,\]
which is the second claim of the proposition.

Finally, the law \eqref{eqn:dyn_a_lawx} written in terms of the distortion as $c(r+H_t(z))/H_x(z) = r + \Delta$ has the LHS strictly
decreasing in $z$. The first best solution solves this equation setting $\Delta \equiv 0$, so
$\Psi(z^{FB}) = 0$ and $z^{FB} \equiv z_*$; the first-best path is the ray
$x^{FB}(t) = \lambda t/z_*$. As argued in the first paragraph of the proof, $z^{C}(t) > z_* = z^{FB}$ for all $t>0$, and therefore
\[x^{C}(t) = \frac{\lambda t}{z^{C}(t)} < \frac{\lambda t}{z_*} = x^{FB}(t) \qquad\text{for all } t>0,\]
which is the first claim of the proposition.
\end{proof}

\begin{proof}[Proof of Proposition \ref{prop:extensive}]
For this proof write
\[
 \begin{aligned}
 S(z)&=\mathbb E_\theta[e^{-\lambda_\theta z}],&
 F(z)&=1-S(z),& f(z)&=F'(z),\\
 h(z)&=\frac{f(z)}{S(z)},&
 b(z)&=\frac{\gamma}{h(z)},& \Phi(z)&=F(z)(1-b(z)).
 \end{aligned}
\]
The mixture hazard satisfies $h'(z)=-\operatorname{Var}(\lambda_\theta\mid\text{no success after effort }z)\leq0$. Thus $b$ is nondecreasing and $0\leq b(z)\leq\gamma/\lambda_H<1$.

First consider any measurable contract $\alpha$ and a global agent response $K$, with $k=\dot K\in[0,1]$. Define
\[
 \begin{aligned}
 p(t)&=\int_t^\infty e^{-rs}\bigl[\alpha(s)f'(K(s))+\gamma f(K(s))\bigr]k(s)\dd{s},\\
 B(t)&=e^{-rt}\bigl[\alpha(t)f(K(t))-\gamma S(K(t))\bigr]+p(t).
 \end{aligned}
\]
The first variation of the agent's payoff with respect to effort is $\int B(t)\delta k(t)\dd{t}$. Feasible reductions of effort therefore imply a.e. $B(t)k(t)\geq0$. Let $H=F/f$, so $H'=1-Ff'/f^2>0$. Integration by parts gives
\[
 U_P[\alpha;K]
 =\int_0^\infty re^{-rt}\Phi(K(t))\dd{t}
   -\int_0^\infty B(t)H'(K(t))k(t)\dd{t}
 \leq\int_0^\infty re^{-rt}\Phi(K(t))\dd{t}.
\]
Here $H(K(0))=0$ and $H(K(t))p(t)=O(e^{-rt})$, since $f\geq\lambda_H S$ and $|p(t)|\leq(\lambda_E^2+\gamma\lambda_E)e^{-rt}S(K(t))/r$.

The upper bound reduces to choosing an effort cutoff. Define
\[
 J(T)=\int_0^T e^{-rz}\Phi'(z)\dd{z},\qquad T\in[0,\infty].
\]
For any feasible $K$, integration by parts and the inverse effort time $\tau(z)$ yield
\[
 \int_0^\infty re^{-rt}\Phi(K(t))\dd{t}
 =\int_0^\infty e^{-rz}\Phi'(z)D(z)\dd{z},
 \qquad D(z)=e^{-r(\tau(z)-z)},
\]
where $D=0$ for effort levels never reached. Because $k\leq1$, $\tau(z)-z$ is nondecreasing, so $D$ is nonincreasing and lies in $[0,1]$. Representing $D$ as a mixture of cutoff indicators shows that the last integral is a convex combination of the values $J(T)$, and hence is at most $\max_TJ(T)$. A maximizer $T^*\in[0,\infty]$ exists because $J$ is continuous and has a finite limit at infinity. Moreover, $\Phi'(0)=f(0)-\gamma>0$, so $T^*>0$.

This bound is attained by $K_T(t)=\min\{t,T\}$ and the contract
\begin{equation}
 \alpha_T(t)=\int_t^\infty re^{-r(s-t)}b(K_T(s))\dd{s}.
 \label{eqn:alpha_extensive}
\end{equation}
Indeed, this share is feasible and satisfies $\alpha_T-\dot\alpha_T/r=b(K_T)$. With $G(z)=\gamma\int_0^zS(u)\dd{u}$, the agent's payoff from any feasible deviation $Y$ is
\[
 U_A[Y;\alpha_T]
 =r\int_0^\infty e^{-rt}\bigl[b(K_T(t))F(Y(t))-G(Y(t))\bigr]\dd{t}.
\]
The derivative of the bracketed term in $Y$ is $f(Y)[b(K_T)-b(Y)]$, so $Y=K_T$ maximizes it pointwise and is a global best response. Integration by parts then gives $U_P[\alpha_T;K_T]=J(T)$. Taking $T=T^*$ proves optimality among all feasible contracts and global agent responses, including boundary responses.

Finally, since $b(K_T)$ is nondecreasing,
\[
 \dot\alpha_T(t)=r\bigl[\alpha_T(t)-b(K_T(t))\bigr]\geq0.
\]
For finite $T$, the share is constant at $b(T)$ after stopping. If $\lambda_H < \lambda_E$ and $\gamma>0$, it is strictly increasing before $T$.

The conclusion also applies to any optimal $(\alpha,K)$, up to the choice of a payoff-equivalent implementing contract. Optimality implies $Bk=0$ a.e. in the payoff bound above. Consequently
\[
 \frac{\dd{}}{\dd{t}}\left(\frac{p(t)}{f(K(t))}\right)
 =e^{-rt}b'(K(t))k(t),\qquad
 \frac{p(t)}{f(K(t))}\longrightarrow0.
\]
Using $B=0$ on active dates and integrating yields
\[
 \alpha(t)=b(K(t))+\int_t^\infty e^{-r(s-t)}b'(K(s))k(s)\dd{s}
 =\int_t^\infty re^{-r(s-t)}b(K(s))\dd{s}
 \quad\text{a.e. on }\{k>0\}.
\]
This gives a nondecreasing, locally absolutely continuous contract. The same pointwise verification implements $K$ under this contract, and its on-path payments are unchanged.
\end{proof}

\begin{proof}[Proof of Proposition \ref{prop:asymptotic_decrease}]
Throughout the proof again $d := t/x$, and let $\pi_t$ denote the posterior belief over $\theta$ at time $t$.

\begin{lemma}\label{lem:collapse}
Suppose that along some path as $t\to\infty$ that $d\to\bar d\in(0,\infty)$. Then on this path, $x(t)\to\infty$,
$x(1-F)\to 0$, and
\[\pi_t \to_d \delta_{H},\qquad x\,\mathrm{Var}^{\pi_t}\bigl(H_{x,\theta}\bigr)\to 0,\qquad x\,\bigl|\mathrm{Cov}^{\pi_t}\bigl(H_{x,\theta},H_{t,\theta}\bigr)\bigr|\to 0 .\]
\end{lemma}
\begin{proof}
Since $x = t/d$ and $\bar d<\infty$, $x\to\infty$. Because $\lambda_\theta \ge \lambda_H > 0$,
$\nu_0 x\bigl(1-e^{-\lambda_\theta d}\bigr) \ge \nu_0 x\bigl(1-e^{-\lambda_H d}\bigr)$, so
$x(1-F) \le x\exp\bigl(-\nu_0 x(1-e^{-\lambda_H d})\bigr)\to 0$. The posterior log-likelihood ratio:
\[ \log \frac{\pi_t(H)}{\pi_t(E)} = - \nu_0 x (e^{-\lambda_H d} - e^{-\lambda_E d}) \to - \infty, \]
hence $\pi_t \to_d \delta_H$. 
For the variance/covariance limits, note that the posterior probability of $E$ is
\[ \frac{e^{-\nu_0 x (1-e^{-\lambda_E d})}}{e^{-\nu_0 x (1-e^{-\lambda_E d})} + e^{-\nu_0 x (1-e^{-\lambda_H d})}} = \frac{e^{-\nu_0 x (e^{-\lambda_Hd} -e^{-\lambda_E d})}}{ e^{-\nu_0 x (e^{-\lambda_Hd} -e^{-\lambda_E d})}+ 1 } = O\bigl(\exp(-\nu_0 x (e^{-\lambda_H \bar{d}}-e^{-\lambda_E \bar{d}}))\bigr),\] 
and since $H_{x,\theta}, H_{t,\theta}$ can bounded above by some $M$, 
$\mathrm{Var}^{\pi_t}(H_{x,\theta})\le \pi_t(\theta\ne H)\cdot M^2 = O\bigl(\exp(-\nu_0 x (e^{-\lambda_H \bar{d}}-e^{-\lambda_E \bar{d}}))\bigr)$,
with the same bound for $|\mathrm{Cov}^{\pi_t}(A^x_\theta,A^t_\theta)|$ by Cauchy-Schwarz.
Multiplying by $x$ preserves the limit $0$.
\end{proof}

\begin{lemma}\label{lem:decomp}
Let
\[\mathcal B := \frac{ F c}{(\mathbb{E}^{\pi_t}[H_{x,\theta}])^2}\left[ \biggl(x\frac{F_{xx}}{1-F} + x(\mathbb{E}^{\pi_t}[H_{x,\theta}])^2\biggr)r + x\frac{F_{xx}}{1-F}\,\mathbb{E}^{\pi_t}[H_{t,\theta}] - x\frac{F_{xt}}{1-F}\,\mathbb{E}^{\pi_t}[H_{x,\theta}]\right].\]
Suppose that on some path of $x$, $t \to \infty$ that $d \to \bar{d} < \infty$. Then $\mathcal{B}$ converges to a finite limit strictly less than zero. 
\end{lemma}
\begin{proof}
Take the limit in pieces. By the supposition, $d \to \bar{d}$, so by Lemma \ref{lem:collapse} $x \to \infty$, which implies $F \to 1$, and since $\lambda_H > 0$, $\mathbb{E}^{\pi_t} [ H_{x,\theta}] \to \nu_0 (1 - e^{-\lambda_H \bar{d}} - \lambda_H \bar{d} e^{-\lambda_H \bar{d}})$. Therefore,
\[ \frac{Fc}{\mathbb{E}^{\pi_t}[H_{x,\theta}]^2 } \to \frac{c}{\nu_0^2 (1 - e^{-\lambda_H \bar{d}} - \lambda_H \bar{d} e^{-\lambda_H \bar{d}})^2}.\]
Using the identities from \eqref{eqn:dFdt} - \eqref{eqn:d2Fdxdt},
\begin{align*}
    x\frac{F_{xx}}{1-F} +  x(\mathbb{E}^{\pi_t}[H_{x,\theta}])^2 &= \mathbb{E}^{\pi_t}\left[ - H_{t,\theta}\frac{\lambda_\theta t^2}{x^2} - xH_{x,\theta}^2  \right] + x(\mathbb{E}^{\pi_t}[H_{x,\theta}])^2, \\
    &= - \mathbb{E}^{\pi_t}\left[ H_{t,\theta}(x,t)\frac{\lambda_\theta t^2}{x^2}   \right] - x\mathrm{Var}^{\pi_t}(H_{x,\theta}),
\end{align*}
Note that in the limit as $t \to \infty$, using Lemma \ref{lem:collapse}, the variance term vanishes, so
\[  x\frac{F_{xx}}{1-F} +  x(\mathbb{E}^{\pi_t}[H_{x,\theta}])^2  \to - \nu_0 \lambda_H e^{-\lambda_H \bar{d}} \lambda_H \bar{d}^2. \]
Finally, the last term can be rewritten as  
\begin{align*}
     &x\frac{F_{xx}}{1-F}\mathbb{E}^{\pi_t}[H_{t,\theta}] - x\frac{F_{xt}}{1-F}\mathbb{E}^{\pi_t}[H_{x,\theta}] \\
     &= -  \mathbb{E}^{\pi_t}\left[ H_{t,\theta}\frac{\lambda_\theta t^2}{x^2}   \right] \mathbb{E}^{\pi_t}[H_{t,\theta}] - x\mathrm{Var}^{\pi_t}(H_{x,\theta})\mathbb{E}^{\pi_t}[H_{t,\theta}]   - \mathbb{E}^{\pi_t}\left[ H_{t,\theta}\frac{\lambda_\theta t}{x} \right]\mathbb{E}^{\pi_t}\left[H_{x,\theta} \right] + x\mathrm{Cov}^{\pi_t}(H_{x,\theta},H_{t,\theta})
\end{align*}
In the limit, the covariance and variance terms vanish by Lemma \ref{lem:collapse}, so 
\[ x\frac{F_{xx}}{1-F}\mathbb{E}^{\pi_t}[H_{t,\theta}] - x\frac{F_{xt}}{1-F}\mathbb{E}^{\pi_t}[H_{x,\theta}] \to -\nu_0 \left(\lambda_H e^{-\lambda_H \bar{d}}\right)^2 \lambda_H \bar{d}^2 - \lambda_H^2 \nu_0 e^{-\lambda_H \bar{d}} \bar{d} \nu_0\left(1 - e^{-\lambda_H\bar{d}}-\lambda_H\bar{d}e^{-\lambda_H \bar{d}}\right). \]
Putting these three pieces together, the limit of $\mathcal{B}$ is negative and finite.
\end{proof}

Now, I present the proof of the proposition. First, I show that $t / x \to \infty$. Since along any admissible path the depth $d(t) = t/x(t)$ is weakly increasing, it converges to some (possibly infinite) $\bar d \in (0,\infty]$. Suppose, for contradiction, that $\bar d<\infty$, and write
$z_H := \lambda_H\bar d > 0$, which is strictly positive precisely because $\lambda_H>0$.

Multiplying the law \eqref{eqn:dyn_a_lawx} through by $x$ and using the definition of $\mathcal B$,
\begin{equation}\label{eq:xlaw}
  0 = rx(1-F)\mathbb{E}^{\pi_t}[H_{x,\theta}] - rcx(1-F) - cx(1-F)\mathbb{E}^{\pi_t}[H_{t,\theta}] + \mathcal B .
\end{equation}
By Lemma~\ref{lem:collapse}, $x(1-F)\to0$ while $\mathbb{E}^{\pi_t}[H_{x,\theta}]$ and $\mathbb{E}^{\pi_t}[H_{t,\theta}]$ are bounded, so the first three
terms vanish. So  the right side of \eqref{eq:xlaw} converges to the limit of $\mathcal{B}$, which by Lemma~\ref{lem:decomp} is a strictly negative
constant, contradicting the fact that it is identically zero. Therefore $\bar d = \infty$.

 Now, as $d\to\infty$, $H_{t,\theta} = \nu_0\lambda_\theta e^{-\lambda_\theta d}\to0$
and $H_{x,\theta} \to \nu_0$ uniformly over $\theta$, so $\mathbb{E}^{\pi_t}[H_{t,\theta}]\to0$ and $\mathbb{E}^{\pi_t}[H_{x,\theta}]\to\nu_0$ regardless of the posterior. Consequently the
incentive cost term
\[I(t) := \frac{rc + c\mathbb{E}^{\pi_t}[H_{t,\theta}]}{r\mathbb{E}^{\pi_t}[H_{x,\theta}]} \longrightarrow \frac{c}{\nu_0} .\]
Fix $T>0$ such that $d(t)\ge d(T)>0$ for $t\ge T$; on $[T,\infty)$ the quantity $h_t$ is bounded
above and $h_x$ is bounded below away from zero, so $I$ is bounded there. Since \eqref{eqn:dyn_a_alpha} can be written
$\alpha(t) = \int_0^\infty re^{-rs}I(s+t)\,ds$, dominated convergence gives
$\lim_{t\to\infty}\alpha(t) = \frac{c}{\nu_0}$.
To show convergence from above, observe that for every finite $t$ and $d$, $\mathbb{E}^{\pi_t}[H_{t,\theta}]>0$ and
$\mathbb{E}^{\pi_t}[H_{x,\theta}] < \nu_0$
strictly, so
\[ I(t) = \frac{c\bigl(r+\mathbb{E}^{\pi_t}[H_{t,\theta}]\bigr)}{r\mathbb{E}^{\pi_t}[H_{x,\theta}]} >\frac{cr}{r\nu_0} = \frac{c}{\nu_0} \qquad \forall t>0 .\]
Therefore for all finite $t$, $\alpha(t) > \int_0^\infty re^{-rs} (c/\nu_0) = c/\nu_0$ . Together with
Step 2, $\alpha(t)\to c/\nu_0$ from above.
\end{proof}
\newpage 
\section{Supplemental Material}
\subsection{Equivalence of \eqref{defn:taustar} and \eqref{defn:taustar_appendix}}

\begin{proof}
    Start from \eqref{defn:taustar_appendix}:
    \[ 1+   \frac{c(r+\lambda)}{\lambda(1-\nu_0)} =  \frac{\lambda}{r+\lambda}e^{-rK^*} +  \left(\frac{r}{r+\lambda} - \frac{cr}{\nu_0\lambda}\right)e^{\lambda K^*}. \] 
    Isolating the cost, 
    \[ c \left( \frac{r+\lambda}{\lambda (1-\nu_0) } + \frac{r}{\nu_0\lambda e^{-\lambda K^*}} \right) = \frac{\lambda}{\lambda+ r} e^{-rK^*} + \frac{r}{\lambda + r}e^{\lambda K^*}- 1 \]
    Simplifying the term multiplying $c$, 
    \[ c \left( \frac{(r+\lambda)\nu_0 e^{-\lambda K^*}}{\lambda (1-\nu_0) \nu_0 e^{-\lambda K^*}} + \frac{r(1-\nu_0)}{(1-\nu_0)\nu_0\lambda e^{-\lambda K^*}} \right) = \frac{\lambda}{\lambda+ r} e^{-rK^*} + \frac{r}{\lambda + r}e^{\lambda K^*}- 1, \]
    which becomes 
    \[ c \left( \frac{r (1- \nu_0 + \nu_0 e^{-\lambda K^*}) + \lambda \nu_0 e^{-\lambda K^*} }{\lambda (1-\nu_0) \nu_0 e^{-\lambda K^*}}  \right) = \frac{\lambda}{\lambda+ r} e^{-rK^*} + \frac{r}{\lambda + r}e^{\lambda K^*}- 1 . \]
    Using $\nu(K) = \nu_0 e^{-\lambda K}/ (1 - \nu_0 + \nu_0 e^{-\lambda K})$,
    \[ c \left( \frac{r + \lambda \nu(K^*) }{\lambda (1-\nu_0) \nu(K^*)} \right) = \frac{\lambda}{\lambda+ r} e^{-rK^*} + \frac{r}{\lambda + r}e^{\lambda K^*}- 1 . \]
    Multiplying through by $1-\nu_0$, 
    \[ c \left( \frac{r + \lambda \nu(K^*) }{\lambda \nu(K^*)} \right) = (1-\nu_0)\frac{\lambda}{\lambda+ r} e^{-rK^*} + (1-\nu_0)\frac{r}{\lambda + r}e^{\lambda K^*}- 1 + \nu_0 . \]
    Unifying the right-hand side,
    \[ c \left( \frac{r + \lambda \nu(K^*) }{\lambda \nu(K^*)} \right) = \frac{(1-\nu_0) \lambda e^{-rK^*} + (1-\nu_0) r e^{\lambda K^*} - (\lambda + r) + \nu_0(\lambda + r) }{\lambda + r}. \]
    Since $\nu(K) = \nu_0 e^{-\lambda K}/ (1 - \nu_0 + \nu_0 e^{-\lambda K})$ equivalently implies that $\nu_0 \frac{1-\nu(K)}{\nu(K)} = (1-\nu_0)e^{\lambda K}$,
    \[ c \left( \frac{r + \lambda \nu(K^*) }{\lambda \nu(K^*)} \right) = \frac{(1-\nu_0) \lambda e^{-rK^*} + r\nu_0 \frac{1-\nu(K^*)}{\nu(K^*)}- (\lambda + r) + \nu_0(\lambda + r) }{\lambda + r}. \]
    Simplifying,
    \[ c \left( \frac{r + \lambda \nu(K^*) }{\lambda \nu(K^*)} \right) = \frac{(1-\nu_0) \lambda e^{-rK^*} + r\nu_0 \frac{1}{\nu(K^*)}- (\lambda + r) + \nu_0 \lambda  }{\lambda + r}. \]
    Multiplying through by $\nu(K^*)$,
    \begin{align*} c \left( \frac{r + \lambda \nu(K^*) }{\lambda} \right) &= \frac{(1-\nu_0) \lambda e^{-rK^*}\nu(K^*) + r \nu_0 -(\lambda + r) \nu(K^*) + \nu_0 \lambda \nu(K^*) }{\lambda + r}, \\
     &=  \frac{(1-\nu_0) \lambda e^{-rK^*}\nu(K^*) + \nu_0(r + \lambda \nu(K^*)) }{\lambda + r} - \nu(K^*).\end{align*}
    Multiplying through by $\lambda / (r + \lambda \nu(K^*))$,
    \[ c =  \frac{(1-\nu_0) \lambda e^{-rK^*}\lambda p (K^*) }{(\lambda + r)(\lambda \nu(K^*) + r)} + \nu_0 \frac{\lambda}{\lambda + r} - \frac{\lambda \nu(K^*)}{\lambda \nu(K^*) + r}.\]
    Rearranging,
    \[ \frac{\lambda \nu(K^*)}{\lambda \nu(K^*) + r}\left(1 - (1-\nu_0) \frac{\lambda}{\lambda + r} e^{-rK^*} \right) =   \nu_0 \frac{\lambda}{\lambda + r} - c .\]
    Subtracting $\nu_0e^{-(r+\lambda) K^*} \lambda /(\lambda + r)$ from both sides, and pulling out the common factor $\frac{\lambda \nu(K^*)}{\lambda \nu(K^*) + r}$ on the left-hand side, 
    \[ \frac{\lambda \nu(K^*)}{\lambda \nu(K^*) + r}\left(1 - (1-\nu_0) \frac{\lambda}{\lambda+r} e^{-rK^*}  - \frac{\nu_0 \lambda e^{-(r+\lambda)K^*} (\lambda p^*(K) + r)}{\lambda p^*(K)(\lambda + r)}\right) =   \nu_0 \frac{\lambda}{\lambda + r}\left( 1 - e^{-(r+\lambda)K^*}\right) - c .\]
    Simplifying by using $ \nu_0  e^{-\lambda K^*} / \nu(K^*) = (1-\nu_0 + \nu_0e^{-\lambda K^*})$, the left-hand-side becomes
    \[\frac{\lambda \nu(K^*)}{\lambda \nu(K^*) + r}\left(1 - (1-\nu_0) \frac{\lambda}{\lambda + r} e^{-rK^*}  - \frac{ e^{-rK^*} (1 - \nu_0 + \nu_0 e^{-\lambda K^*}) (\lambda p^*(K) + r)}{\lambda + r}\right). \]
    Rearranging the left-hand side yields
    \[\frac{\lambda \nu(K^*)}{\lambda \nu(K^*) + r}\left(1 - \frac{e^{-rK^*} (1 - \nu_0 + \nu_0 e^{-\lambda K^*})}{\lambda + r}\left[ \frac{(1-\nu_0) \lambda }{1 - \nu_0 + \nu_0 e^{-\lambda K^*}} +  \lambda p^*(K) + r \right] \right). \]
    Simplifying the bracketed expression, 
    \[\frac{\lambda \nu(K^*)}{\lambda \nu(K^*) + r}\left(1 - \frac{e^{-rK^*} (1 - \nu_0 + \nu_0 e^{-\lambda K^*})}{\lambda + r}\left[ \frac{(1-\nu_0) \lambda  + \lambda \nu_0 e^{-\lambda K^*}}{1 - \nu_0 + \nu_0 e^{-\lambda K^*}} + r \right] \right), \]
    and the bracketed term becomes $\lambda + r$, so 
    \[ \frac{\lambda \nu(K^*)}{\lambda \nu(K^*) + r}\left(1 - e^{-rK^*} (1 - \nu_0 + \nu_0 e^{-\lambda K^*})\right) =   \nu_0 \frac{\lambda}{\lambda + r}\left( 1 - e^{-(r+\lambda)K^*}\right) - c .\]
    Dividing, 
    \[ \frac{\nu_0 \frac{\lambda}{\lambda + r}\left( 1 - e^{-(r+\lambda)K^*}\right) - c}{1 - e^{-rK^*} (1 - \nu_0 + \nu_0 e^{-\lambda K^*})} =  \frac{\lambda \nu(K^*)}{\lambda \nu(K^*) + r}  .\]
    Finally, noting that the LHS is the sum of an infinite geometric series gives \eqref{defn:taustar}.
\end{proof}
\subsection{Algebraic Derivations in Proposition \ref{prop:no_commit}}

\begin{claim}
Fix $x>0$. On the region where $0\leq \alpha(z)\leq 1$, 
\[ \int_0^{xz} e^{-rs}F_t(x,s)\,ds +e^{-rxz}S(x,xz)(1-\alpha(z))W(z) \]
is (1) quasiconcave in $z$, (2) maximized at some value of less than $ d_\nc$, (3) for any $z$ greater (resp. less) than the maximizer, 
\[ V'(z)  + x[H_t(z) (1 - V(z)) - rV(z)] \le (\ge) \, 0 . \]
The maximizer is decreasing in $x$. 
\end{claim}
\begin{proof}
Differentiating the expression in $z$ yields
\[e^{-rxz}S(x,xz)\left\{\frac{d}{dz}\left[(1-\alpha(z))W(z)\right]  +x\left[H_t(z)\left(1-(1-\alpha(z))W(z)\right)-r(1-\alpha(z))W(z)\right]\right\},\]
Because $ e^{-rxz}S(x,xz)>0$, the sign of the derivative is exactly the sign of the expression in braces.
Thus, it suffices to prove that this term is strictly single-crossing from above.

Letting $V(z) := (1-\alpha(z))W(z)$ as before, consider the term in curly braces, 
\[V'(z)+x\left[H_t(z)(1-V(z))-rV(z)\right].\]
I show first that whenever the bracketed part is positive, $V'(z)$ is also positive, so a zero can only occur when the bracketed term is negative and $V'(z)$ is positive. To see this, note the following identity, obtained by substituting the definition of $W$ and simplifying:
\begin{align}
H_t(z)(1-V(z))-rV(z) &=  \frac{H_t(z)(rz + \alpha(z) \nu_0(1-e^{-\lambda z})) - r (1-\alpha(z)) \nu_0(1-e^{-\lambda z})}{ rz + \nu_0(1-e^{-\lambda z})} \notag \\
&= \frac{(r + H_t(z))\alpha(z) \nu_0(1-e^{-\lambda z}) - r H_x(z)}{ rz + \nu_0(1-e^{-\lambda z})}.\notag
\end{align}
So 
\begin{align}
&H_t(z)(1-V(z))-rV(z) > 0 \notag \\
&\iff  \alpha(z)>\frac{rH_x(z)}{(r + H_t(z))\nu_0(1-e^{-\lambda z})} = \frac{r H_x(z)}{r H_x(z) + H_t(z) ( rz + \nu_0 (1-e^{-\lambda z}))}.
\label{ineq:first_threshold}
\end{align}
Next, differentiation of $V(z)$ gives
\begin{align}
V'(z)&=-(1-\alpha(z))\frac{rH_x(z)}{(rz+\nu_0(1-e^{-\lambda z})\bigr)^2}+\alpha(z)W(z)\left[\frac{\lambda H_t(z)}{r+H_t(z)}+\frac{\lambda zH_t(z)}{H_x(z)}\right], \notag \\
&= \frac{-rH_x(z) + \alpha(z) \left( r H_x(z) + H_t(z) (rz + \nu_0(1-e^{-\lambda z})  )\lambda \nu_0 (1-e^{-\lambda z})\left[\frac{rz + \nu_0(1-e^{-\lambda z})}{(r + H_t(z))H_x(z)}\right] \right)}{(rz+\nu_0(1-e^{-\lambda z})\bigr)^2}. \notag 
\end{align}
Thus, $V'(z)>0$ iff
\begin{align*}
\alpha(z)>\frac{rH_x(z)}{rH_x(z)+H_t(z)\left(rz+\nu_0(1-e^{-\lambda z})\right)\lambda \nu_0(1-e^{-\lambda z})\left[\frac{rz + \nu_0(1-e^{-\lambda z})}{(r + H_t(z))H_x(z)}\right]}.
\end{align*}
The first $\alpha$ threshold in \eqref{ineq:first_threshold} is weakly larger than the threshold because
\begin{align*} 
\lambda \nu_0(1-e^{-\lambda z})\left[\frac{rz + \nu_0(1-e^{-\lambda z})}{(r + H_t(z))H_x(z)}\right] &\ge \lambda \nu_0(1-e^{-\lambda z})\left[\frac{rz + \nu_0\lambda z e^{-\lambda z}}{(r + H_t(z))H_x(z)}\right],  \\
& = \lambda z \nu_0(1-e^{-\lambda z})\left[\frac{r + H_t(z)}{(r + H_t(z))H_x(z)}\right], \\
& = \frac{\lambda z (1 - e^{-\lambda z}) }{1 - e^{-\lambda z} - \lambda z e^{-\lambda z}} \ge 1,
\end{align*}
where the first inequality follows since $1 - e^{-\lambda z} \ge \lambda z e^{-\lambda z}$ and the second because $1 - e^{-\lambda z} - \lambda z e^{-\lambda z} \le \lambda z (1-e^{-\lambda z})$. Hence
\begin{equation} H_t(z)(1-V(z))-rV(z)>0 \quad\Longrightarrow\quad V'(z)>0.\label{eqn:appendix_prop6_claim} \end{equation}
Consequently, any zero of the derivative must lie in the region
\[H_t(z)(1-V(z))-rV(z)<0<V'(z).\]
On this intermediate region a zero at $z$ implies that
\[x=\frac{V'(z)}{rV(z)-H_t(z)(1-V(z))}.\]
To finish the argument, I show that the right-hand side expression is strictly decreasing on this
region. By rewriting, 
\[ V'(z) = \frac{-r(1 - \alpha(z))H_x(z)}{(rz + \nu_0(1-e^{-\lambda z}))^2} + \frac{\alpha(z) H_t(z) \lambda \nu_0 (1-e^{-\lambda z})}{(r + H_t(z))H_x(z)}\]
So the ratio is 
\[ \frac{V'(z)}{rV(z) - H_t(z) (1 - V(z)) } = \frac{\frac{-r(1 - \alpha(z))H_x(z)}{(rz + \nu_0(1-e^{-\lambda z}))^2} + \frac{\alpha(z) H_t(z) \lambda \nu_0 (1-e^{-\lambda z})}{(r + H_t(z))H_x(z)}}{\frac{r H_x(z) - (r + H_t(z))\alpha(z) \nu_0(1-e^{-\lambda z})  }{ rz + \nu_0(1-e^{-\lambda z})}}. \]
Multiplying the numerator and denominator by $\frac{rz + \nu_0 (1-e^{-\lambda z})}{\nu_0 (1-e^{-\lambda z})}$ gives
\[ \frac{V'(z)}{rV(z) - H_t(z) (1 - V(z)) } = \frac{-rV(z)\left( \frac{H_x(z)}{\nu_0(1-e^{-\lambda z})}\right) + \lambda \alpha(z) \frac{H_t(z) (rz + \nu_0 (1-e^{-\lambda z}))}{(r + H_t(z))H_x(z)}}{r\frac{H_x(z)}{\nu_0 (1-e^{-\lambda z})} - (r + H_t(z))\alpha(z) }. \]
Observe that $H_x(z)/(\nu_0 (1-e^{-\lambda z}))$ is always increasing, and $H_t$ and $\alpha$ are decreasing, so the denominator of this expression is increasing. It suffices to argue that the numerator is decreasing. Note that by the region assumption, $V'(z) > 0$, so the first term in the numerator is negative and decreasing. Further, $\alpha(z)$ is decreasing, so it is sufficient to show that the last fraction in the numerator is always decreasing. It will be easiest to show that its log derivative is negative. The log derivative of the last fraction gives
\begin{align*}
    \frac{\dd{}}{\dd{z}}\log\left( \frac{H_t(z) (rz + \nu_0 (1-e^{-\lambda z}))}{(r + H_t(z))H_x(z)}\right) &= -\lambda  + \frac{r + \nu_0\lambda e^{-\lambda z}}{rz + \nu_0(1-e^{-\lambda z})} + \frac{\nu_0 \lambda^2 e^{-\lambda z}}{r + \nu_0\lambda e^{-\lambda z}} - \frac{\lambda^2 z e^{-\lambda z}}{1-e^{-\lambda z} - \lambda ze^{-\lambda z}},  \\
    &= \frac{r + \nu_0\lambda e^{-\lambda z}}{rz + \nu_0(1-e^{-\lambda z})} - \frac{r\lambda }{r + \nu_0\lambda e^{-\lambda z}} - \frac{\lambda^2 z e^{-\lambda z}}{1-e^{-\lambda z} - \lambda ze^{-\lambda z}}.
\end{align*}
To see that this expression is negative, it will be useful to pull out a $1/\lambda$ and rewrite everything in terms of $t = \lambda z$ and $a = r/(\nu_0 \lambda)$:
\begin{align*}
    \frac{\dd{}}{\dd{z}}\log\left( \frac{H_t(z) (rz + \nu_0 (1-e^{-\lambda z}))}{(r + H_t(z))H_x(z)}\right) 
    &= \frac{1}{\lambda}\left( \frac{a + e^{-t}}{at + (1-e^{-t})} - \frac{a }{a + e^{-t}} - \frac{t e^{-t}}{1-e^{-t} - te^{-t}} \right) , \\
    &= \frac{1}{\lambda}\left( \frac{(a + e^{-t})^2(1 - e^{-t} - t) - a(1-e^{-t}-te^{-t})^2}{(at + (1-e^{-t}))(a + e^{-t}) (1 - e^{-t}-te^{-t})} \right) \le 0.
\end{align*}
The inequality follows because $1-e^{-t} \le t$. All together, this implies that there is a unique $z$ that maximizes the objective, and further, since it lies in the region where $V'(z) > 0$, quasiconcavity of $V$ immediately implies that $z < d_\nc$. Additionally, since $x$ is exactly set equal to the ratio $\frac{V'(z)}{rV(z)-H_t(z)(1-V(z))}$ for the maximizer, which is decreasing in $z$; hence, the maximizer is decreasing in $x$.
\end{proof}

\subsection{Robustness Proofs}
\subsubsection{General Model of Difficulty}
To prove Theorem \ref{thm:learning_general}, I will replicate the same proof structure as in the baseline model (see Section \ref{sec:learning_proof_sketch}). 

\begin{proof}[Proof of Theorem \ref{thm:learning_general}]
I first start with an analogue of Lemma \ref{lem:exp_bkth_cdf}. Let $S_\theta(K) := \mathbb{E}_\theta[ e^{-\lambda_n K} ] $ be the survival probability of an arm with $K$ effort given $\theta$ (and also the Laplace transform of $G_\theta$). 
    \begin{lemma}
    \label{lem:exp_bkth_cdf_general}
    Given fixed brainstorming times and conditioning on $s, \lambda_\theta$, the action path $k$ induces breakthrough time cumulative distribution function \eqref{eqn:bk_cdf}.
\end{lemma}
\begin{proof}
Let $\tau^k_\theta$ denote the stochastic breakthrough arrival time given path $k$, conditional on $s$ and $\theta$. 
The CDF of the $\tau^k_\theta$ for given realization of $\{ \lambda_n \}$ is 
\[ F(t\mid s,k,\{\lambda_n\}) = 1 - \exp\left( - \sum_{m=1}^\infty \lambda_m (\xi^k_m(t) - s_m)  \right) . \]
Taking the expectation over $\{\omega_n\}$, 
\begin{align*} F(t\mid s,\theta,k) &= 1 - \mathbb{E}\left[\left.\prod_{m=1}^\infty \exp\left( - \lambda_m (\xi^k_m(t) - s_m)  \right) \right| s, \theta \right], \\
&= 1 - \prod_{m=1}^\infty \mathbb{E}\left[\left.\exp\left( - \lambda_m (\xi^k_m(t) - s_m)  \right) \right| s, \theta \right].
\end{align*}
where the last line follows because $\lambda_m$ and $\lambda_n$ are conditionally independent given $\theta$, since $\theta$ is the only source of correlation across arms. The expectation
\begin{align*}
    \mathbb{E}\left[\left. e^{-\lambda_n (\xi^k_n(t) - s_n)}\right|s,\theta \right]  &= \int \frac{e^{-\lambda_n s_n} e^{-\lambda_n (\xi^k_n(t) - s_n)}}{\mathbb{E}[e^{-\lambda_m s_n} \mid \theta]} dG_\theta(\lambda_n)  = \frac{\mathbb{E}_\theta[e^{-\lambda_n \xi^k_n(t) }]}{\mathbb{E}_\theta[e^{-\lambda_n s_n }]}.
\end{align*}
Plugging this into the CDF expression yields the result.
\end{proof}

The next step is to establish the equalizing lemma. In the interest of not copying the entirety of the proof of Lemma \ref{lem:lowest_history_of_effort}, observe that the key step in Lemma \ref{lem:lowest_history_of_effort} where the functional form of $S_\theta$ was used was in establishing the direction of the interchanges on the survival functions; in particular, because $S_\theta$ is still convex, the exact same argument for Lemma \ref{lem:lowest_history_of_effort} still holds. 

Now, I establish Lemma \eqref{lem:increasing_thresholds} for the general model. As before, I will first prove the intermediate lemma:
\begin{lemma}\label{lem:threshold_decrease_after_general}
    Suppose the optimal threshold sequence satisfies $K^*_i \ge K^*_{i+1}$ for some $i\ge 1$. Then $K^*_{i+1} \ge K^*_{i+2}$.
\end{lemma}
\begin{proof}
The proof will proceed identically to Lemma \ref{lem:threshold_decrease_after}. Suppose, for sake of contradiction, that optimally $K^*_i \ge K^*_{i+1} < K^*_{i+2}$. Let $k$ be the original (supposedly) optimal action path, and define $k'$ and $k''$ exactly as in Lemma \ref{lem:threshold_decrease_after}. By the same logic, this pair of interchanges implies that 
    \[ \mathbb{E} \left[ \begin{pmatrix}
       \lambda_{i+1} (1 - e^{-rK^*_{i+1}}e^{-\lambda_{i+2}K_{i+1}^*})  \\
       -( r + \lambda_{i+1})\left(-c + \int_0^{K^*_{i+1}} e^{-rt}\lambda_{i+2} e^{-\lambda_{i+2} t}  \dd{t} \right)   \dd{t} 
    \end{pmatrix} \mid t_{i+2} \right] = 0.  \]
    Denote the expected arrival rate given $\theta$ and historical effort $K$ as 
    \[ \lambda_\theta(K) = \mathbb{E}_\theta[\lambda_{n} \mid K] = \frac{\int \lambda e^{-\lambda K} dG_\theta(\lambda)}{S_\theta(K)} .\]
    Using the law of iterated expectations and the fact that $\lambda_{i+1}$, $\lambda_{i+2}$ are conditionally independent given $\theta$, this becomes 
    \[ \mathbb{E}\left[\left. \lambda_{\theta}(K^*_{i+1}) (1 - e^{-rK^*_{i+1}}S_\theta(K^*_{i+1})) - (r+\lambda_\theta(K^*_{i+1})) \left( -c + \int_0^{K^*_{i+1}} e^{-rt}\lambda_{\theta}(t) S_\theta(t)  \dd{t} \right)  \right|  t_{i+2} \right] = 0\]
    Recalling the definition of $\hat{\phi}_\theta$ from \eqref{eqn:hatphi}, this is equivalent to 
    \[ \mathbb{E}\left[\left. \hat\phi_\theta(K^*_{i+1}) \right|  t_{i+2} \right] = 0 \]
    \begin{lemma}\label{lem:phi_decrease_general}
        The function $\hat\phi$ defined in \eqref{eqn:hatphi} is decreasing.
    \end{lemma}
    \begin{proof}
    Taking the derivative,
    \begin{align} \frac{\dd}{\dd K }\hat\phi_\theta(K ) =& \lambda_\theta'(K)\left[ 1 + c - \int_0^{K} e^{-rt}\lambda_{\theta}(t) S_\theta(t)  \dd{t}  - e^{-rK}S_\theta(K) \right]  \notag \\ 
    &- (r + \lambda_\theta(K)) e^{-rK}\lambda_\theta(K)S_\theta(K) + \left[ re^{-rK}S_\theta(K) +  e^{-rK}\lambda_\theta(K)S_\theta(K) \right]\lambda_\theta(K) , \notag \\
    =& \lambda_\theta'(K)\left[ 1 + c - \int_0^{K} e^{-rt}\lambda_{\theta}(t) S_\theta(t)  \dd{t}  - e^{-rK}S_\theta(K) \right],  \notag \\
    =& \lambda_\theta'(K)\left[ (1 - e^{-rK} S_\theta(\infty)) + c - \int_0^{K} e^{-rt}\lambda_{\theta}(t) S_\theta(t)  \dd{t}  - \int_K^\infty e^{-rK} \lambda_\theta(t) S_\theta(t) \dd{t} \right],  \notag\\
    =& \lambda_\theta'(K)\left[ (1 - e^{-rK} S_\theta(\infty)) + c - \int_0^{\infty} e^{-r \min(t,K)}\lambda_{\theta}(t) S_\theta(t)  \dd{t}  \right] , \notag\\
    =& \lambda_\theta'(K)\left[ (1 - e^{-rK}) S_\theta(\infty) + c + \int_0^{\infty} (1 - e^{-r \min(t,K)})\lambda_{\theta}(t) S_\theta(t)  \dd{t}  \right].  \notag
    \end{align}
    Note that the bracketed term is positive. Further, by Cauchy-Schwarz,
    \[ \int \left(\sqrt{e^{-\lambda K}}\right)^2 \dd{G_\theta(\lambda)} \int \left( \lambda \sqrt{e^{-\lambda K}} \right)^2 \dd{G_\theta(\lambda)} \ge \left(\int \lambda e^{-\lambda K} \dd{G_\theta(\lambda)} \right)^2. \]    
    Hence $\lambda_\theta$ is decreasing:
    \begin{align*}
        \lambda'_\theta(K) =&  \frac{- \int e^{-\lambda K} \dd{G_\theta(\lambda)} \int \lambda^2 e^{-\lambda K} \dd{G_\theta(\lambda)} + \int \lambda e^{-\lambda K}\dd G_\theta{(\lambda)} \int \lambda e^{-\lambda K}\dd G_\theta(\lambda) }{S_\theta(K)^2} \le 0,
    \end{align*}
    and the lemma follows.
    \end{proof}
    
    Let $K_E$ denote the root of $\hat\phi_E$ and respectively $K_H$ for $\hat\phi_H$. Since $\hat\phi_\theta$ is strictly decreasing from Lemma \ref{lem:phi_decrease_general}, this immediately implies that $\hat\phi_H(K_E) > 0$, and so $\hat\phi_H(K) > 0 > \hat\phi_E(K)$ for any $K \in (K_E, K_H)$. In consequence, since $\hat\phi_\theta$ is strictly decreasing, it follows that $K^*_{i+1} \in (K_E, K_H)$, and $\hat\phi_E(K^*_{i+1}) < 0 < \hat\phi_H(K^*_{i+1})$. Further, note that by Bayes' rule, if $q_i$ denotes the belief when approach $i$ is brainstormed,
    \[ q_{i+1}(H) = \frac{q_i(H) S_H(K^*_{i+1})}{q_i(E) S_E(K^*_{i+1}) + q_i(H) S_H(K^*_{i+1})}.\]
    Since $S_E(K) < S_H(K)$ for any $K> 0$ by FOSD, the subjective belief $q_{i}(H) > q_{i+1}(H)$ and $q_{i}(L) < q_{i+1}(L)$. Therefore, since $\hat\phi_E(K^*_{i+1}) < 0$ and $\hat\phi_H(K^*_{i+1})) > 0$, it follows that
    \[ \mathbb{E}[\hat\phi_\theta(K^*_{i+1}) \mid t_{i+1}] < \mathbb{E}[\hat\phi_\theta(K^*_{i+1}) \mid t_{i+2}] = 0.\]
    Rewriting out $\phi_\theta$,
    \begin{align*} &\mathbb{E}\left[ \begin{pmatrix} \lambda_\theta (K^*_{i+1})  \left( 1 + c - \int_0^{K^*_{i+1}}e^{-rt}\lambda_\theta(t)S_\theta(t) \dd{t} - e^{-rK^*_{i+1}}S_\theta(K^*_{i+1}) \right)  \\
     - r  \left( -c + \int_0^{K^*_{i+1}}e^{-rt}\lambda_\theta(t)S_\theta(t) \dd{t} \right) \end{pmatrix} \mid t_{i+1} \right] < 0. \end{align*}
    Since $\lambda_\theta$ is decreasing and the contradiction supposition assumed $K_i \ge K^*_{i+1}$, $\lambda_\theta(K^*_i) < \lambda_\theta(K^*_{i+1})$ so it follows that 
    \begin{align*} &\mathbb{E}\left[ \begin{pmatrix} \lambda_\theta (K^*_{i})  \left( 1 + c - \int_0^{K^*_{i+1}}e^{-rt}\lambda_\theta(t)S_\theta(t) \dd{t} - e^{-rK^*_{i+1}}S_\theta(K^*_{i+1}) \right)  \\
     - r  \left( -c + \int_0^{K^*_{i+1}}e^{-rt}\lambda_\theta(t)S_\theta(t) \dd{t} \right) \end{pmatrix} \mid t_{i+1} \right] < 0, \end{align*}
    or 
    \begin{align} &\mathbb{E}\left[ \begin{pmatrix} \lambda_\theta (K^*_{i})  \left( 1 - e^{-rK^*_{i+1}}S_\theta(K^*_{i+1}) \right)  \\
     - (r+\lambda_\theta (K^*_{i}))  \left( -c + \int_0^{K^*_{i+1}}e^{-rt}\lambda_\theta(t)S_\theta(t) \dd{t} \right) \end{pmatrix} \mid t_{i+1} \right] < 0. 
     \label{eqn:int_contradiction_general}
     \end{align}
    To finish the contradiction, consider a third and final interchange $\hat{k}$ \eqref{eqn:hat_inter}.The relevant interval now is from $t_{i+1}- \epsilon$ to $t_{i+1} + K^*_{i+1} = t_{i+2}$. To first order, the original path on the interval yielded payoff
    \[ \mathbb{E}\left[ \lambda_\theta (K^*_{i})\epsilon+e^{-r\epsilon}(1 - \lambda_\theta (K^*_{i+1})\epsilon)\left(-c + \int_0^{K^*_{i+1}} e^{-rt} \lambda_\theta(t)S_\theta(t)  \dd{t} \right)   \dd{t} \mid t_{i+1}-\epsilon \right] + O(\epsilon^2). \]
    The interchanged strategy $\hat{k}$ yields 
    \[ \mathbb{E}\left[ -c + \int_0^{K^*_{i+1}} e^{-rt}\lambda_\theta(t)S_\theta(t)   \dd{t} + e^{-rK^*_{i+1}}S_\theta( K^*_{i+1})\lambda_\theta (K^*_{i})\epsilon \mid t_{i+1}-\epsilon \right] + O(\epsilon^2). \]
    Optimality implies that the payoff difference is nonnegative; taking $\epsilon \to 0$, this implies that
    \begin{align*} \mathbb{E} \left[ \begin{pmatrix}
       \lambda_\theta(K^*_{i}) (1 - e^{-rK^*_{i+1}}S_\theta(K^*_{i+1}))  \\
       -( r + \lambda_\theta(K^*_{i}))\left(-c + \int_0^{K^*_{i+1}} e^{-rt}\lambda_\theta(t)S_\theta(t)  \dd{t} \right)   \dd{t} 
    \end{pmatrix} \mid t_{i+1} \right]\ge 0,
    \end{align*}
    which contradicts \eqref{eqn:int_contradiction_general}, completing the proof. 
\end{proof}

Now, I prove Lemma \ref{lem:increasing_thresholds} in the generalized model.

\begin{proof}[Proof of Lemma \ref{lem:increasing_thresholds} in Generalized Model] 
The proof follows the same structure as before. By induction, Lemma \ref{lem:threshold_decrease_after_general} implies that if the sequence of thresholds ever starts decreasing, it must decrease thereafter; as a result, Lemma \ref{lem:lowest_history_of_effort} implies that no previous approaches are revisited, and hence the solution must coincide with the solution where the agent is mechanically restricted from revisiting approaches from that point on. This alternative problem is easier to characterize; the subsequent sequence of stopping times must therefore solve 
    \begin{align}
        V(\delta)
        &= \max_{\{K_j\}_{j \ge 1}} \sum_\theta \delta(\theta) \left[ -c + \int_0^{K_{1}} e^{-rt}\lambda_\theta (t) S_\theta(t) \dd t + e^{-rK_{1}} S_\theta(K_{1}) V_{\theta}(\{K_j\}_{j \ge 2}) \right]. \label{eqn:alt_problem_no_recall_general}
    \end{align}  
    Let $f(\delta, K_1, \{K_j\}_{j \ge 2})$ denote the objective of the optimization in \eqref{eqn:alt_problem_no_recall_general}.
    \begin{lemma}\label{lem:supermodular_general}
        For the optimal $\{K_j\}_{j \ge 2}$, the maximizing $K_1$ of the objective $f$ is strictly increasing in $\delta(H)$.
    \end{lemma}
    \begin{proof}[Proof of Lemma \ref{lem:supermodular_general}]
    The cross derivative of the objective is 
    \begin{align}
        \frac{\partial^2 f}{\partial q \, \partial K_1} 
        =& - e^{-rK_1} S_E(K_1) \left[ \lambda_E (K_1)  - (r+\lambda_E (K_1))V_{E}(\{K_j\}_{j \ge 2}) \right] \notag \\
        & + e^{-rK_1}S_H(K_1) \left[ \lambda_H (K_1) - (r+\lambda_H (K_1))V_{H}(\{K_j\}_{j \ge 2}) \right].\label{eqn:cross_deriv_general}
    \end{align}
    I will show that this quantity is positive. Optimality of $K_1$ implies that 
    \begin{align*}
        0 = \frac{\partial f}{\partial K_1} =& \sum_{\theta} \delta(\theta)e^{-rK_1} \left[ \lambda_\theta(K_1)S_\theta(K_1) - rS_\theta(K_1) V_\theta(\{K_j\}_{j\ge 2}) - \lambda_\theta(K_1)S_\theta(K_1)V_\theta(\{K_j\}_{j\ge 2}) \right] \\
        =&\sum_{\theta} \delta(\theta) e^{-rK_1} S_\theta(K_1) \left[ \lambda_\theta (K_1) - (r + \lambda_\theta (K_1))V_\theta(\{K_j\}_{j\ge 2})  \right] .
    \end{align*}
    Note that $e^{-rK_1}$ has no $\theta$-dependence, so the first-order condition simplifies to:
    \begin{equation} 0 = \sum_{\theta} \delta(\theta) S_\theta(K_1) \left[ \lambda_\theta (K_1) - (r + \lambda_\theta (K_1))V_\theta(\{K_j\}_{j\ge 2})  \right].  \label{eqn:cross_deriv_foc_general} \end{equation}
    I consider two cases, depending on whether $K_1 < K_H$ or $K_1 \ge K_H$.

    \textbf{Case 1:} Suppose $K_1 < K_H$. The maximum that $V_H(\{K_j\}_{j \ge 2})$ can attain is given by the optimal policy when $H$ is known for sure, which is exactly the benchmark solution. Hence, 
    \[ V_H(\{K_j\}_{j \ge 2}) \le \frac{\lambda_H (K_H)}{\lambda_H(K_H) + r} < \frac{\lambda_H(K_1)}{\lambda_H(K_1) + r}  \]
    where the second inequality follows since $\lambda_\theta$ is strictly decreasing and $K_1 < K_H$ by supposition. Thus, the inequality implies that $\lambda_H (K_1) - (r + \lambda_H (K_1))V_H(\{K_j\}_{j\ge 2}) > 0$; the first-order condition \eqref{eqn:cross_deriv_foc_general} summing to zero immediately implies that $\lambda_E (K_1) - (r + \lambda_E (K_1))V_E(\{K_j\}_{j\ge 2}) < 0$, and so \eqref{eqn:cross_deriv_general} must be positive.

    \textbf{Case 2:} Suppose $K_1 \ge K_H$. Recall $\hat\phi_E$ as defined from \eqref{eqn:hatphi}. By definition, $\hat\phi_E(K_E) = 0$ so Lemma \ref{lem:phi_decrease} implies that $\hat\phi_E(K_1) < 0$, so  
    \begin{equation} \frac{\lambda_E (K_1) }{\lambda_E (K_1) + r} <  \frac{-c + \int_0^{K_1} \nu_0 \lambda_E(t) S_E(t) \dd{t} }{1 - e^{-rK_1}S_E(K_1)} = V_E(\{K_1, K_1,K_1,...\}). \label{eqn:vh_bound1_general}
    \end{equation}
    Further, since the Gittins index for $E$ is maximized at $K_E$, the maximand
    \[\frac{-c + \int_0^{K}\lambda_E(t) S_E(t) \dd{t} }{1 - e^{-rK}S_E(K)} \]
    has positive derivative at $K < K_E$ and negative derivative above $K > K_E$. Since $K_1 \ge K_H > K_E$ by the DRP assumption, it follows that 
    \begin{align*} V_E(\{K_H, K_H, K_H,...\}) &= \frac{-c + \int_0^{K_H}\lambda_E(t) S_\theta(t) \dd{t} }{1 - e^{-rK_H}S_E(K_H)} \\
    &\ge  \frac{-c + \int_0^{K_1}\lambda_E(t) S_E(t) \dd{t} }{1 - e^{-rK_1}S_E(K_1)} = V_E(\{K_1, K_1,K_1,...\}). \end{align*}
    Note that since 
    \[ V_H(\{K_H,K_H,K_H,...\}) \ge V_H(\{K_j\}_{j\ge 2}) \]
    and $\{K_j\}_{j \ge 2}$ optimizes \eqref{eqn:alt_problem_no_recall_general}, 
    \begin{equation}
        V_E(\{K_j\}_{j\ge 2}) \ge V_E(\{K_H,K_H,K_H,...\})  \label{eqn:vh_bound2_general}
    \end{equation} 
    else the maximization could be improved by switching the continuation strategy to $K_H, K_H, K_H, ...$. Combining \eqref{eqn:vh_bound1_general} and \eqref{eqn:vh_bound2_general}, this implies that 
    \[ V_E(\{K_j\}_{j\ge 2}) > \frac{\lambda_E (K_1) }{\lambda_E (K_1) + r}. \]
    Hence $\lambda_E (K_1) - (r + \lambda_E (K_1))V_E(\{K_j\}_{j\ge 2}) < 0$; the first-order condition \eqref{eqn:cross_deriv_foc} summing to zero immediately implies that $\lambda_H (K_1) - (r + \lambda_H (K_1))V_H(\{K_j\}_{j\ge 2}) > 0$, and so \eqref{eqn:cross_deriv} must be positive.

    Once again, the second order condition is straightforward to check, since $S_\theta$, $S_\theta \lambda_\theta$, and $e^{-rK_1}$ are all decreasing in $K_1$. So the implicit function theorem implies the result.
\end{proof}

By the same argument as the baseline model proof of Lemma \ref{lem:increasing_thresholds}, $V(\delta_{j}(H)) = V^*_{j}$ for all $j > i$, and if $K^*_1$ denotes the maximizer in $K_1$ of the problem \eqref{eqn:alt_problem_no_recall}, then $K^*_{j} = K^*_1(\delta_j(H))$ for all $j > i$.
    Lemma \ref{lem:supermodular_general} implies that because $\{ \delta_j(H) \}_{j \ge i}$ is strictly increasing, $\{K^*(\delta_j(H))\}_{j \ge i} $ must be strictly increasing, so $\{K^*_j\}_{j\ge i}$ must be strictly increasing, contradicting the original supposition.
\end{proof}

Given Lemmas \ref{lem:lowest_history_of_effort} and \ref{lem:increasing_thresholds}, it suffices to optimize among increasing sequence of thresholds $\{ K_j \}_{j \ge 1}$ corresponding to the effort on the previous approach when a new approach is brainstormed. To pin down the optimal $K_n^*$, note that the choice of $K_n^*$ does not affect the payoff if breakthrough arrives before $K_{n-1}^*$ effort has been exerted on the first $n$ approaches, or after the $n+2$nd approach is brainstormed (since at that point, all approaches will have $K_{n+1}^*$ historical effort). 
    
Similar to the proof of Lemma \ref{lem:threshold_decrease_after}, let $k$ denote the optimal action path and $t_n$ denote the time that approach $n$ is brainstormed, and consider the interchanges $k', k''$ as in the proof of Theorem \ref{thm:learning}.  By the same calculation as in Lemma \ref{lem:threshold_decrease_after_general}, the first-order condition is 
    \begin{equation}
        \mathbb{E} \left[ \begin{pmatrix}
       \lambda_\theta (K^*_{n}) (1 - e^{-rK^*_{n}}S_\theta(K_{n}^*))  \\
       -( r + \lambda_\theta (K^*_{n}))\left(-c + \int_0^{K^*_{n}} e^{-rt}\lambda_\theta(t)S_\theta(t)  \dd{t} \right)   \dd{t} 
    \end{pmatrix} \mid t_{n+1} \right] = 0. \label{eqn:appendix_int_condition_general}
    \end{equation} 
    Note that the belief at time $t_{n+1}$ depends on $K^*_n$: 
    \begin{equation}
        \delta_{n+1}(\theta) = \frac{ S_\theta(K^*_n)^n \delta_0(\theta) }{ \sum_\theta \delta_0(\theta) S_\theta(K^*_n)^n }. \label{eqn:appendix_int_belief_general}
    \end{equation}  
    Combining \eqref{eqn:appendix_int_belief_general} with \eqref{eqn:appendix_int_condition_general},
    \[\sum_\theta  \delta_{n+1}(\theta) \begin{pmatrix}
       \lambda_\theta (K^*_{n}) (1 - e^{-rK^*_{n}}S_\theta(K_{n}^*))  \\
       -( r + \lambda_\theta (K^*_{n}))\left(-c + \int_0^{K^*_{n}} e^{-rt}\lambda_\theta(t)S_\theta(t)  \dd{t} \right)   \dd{t} 
    \end{pmatrix}  = 0. \]
    Multiplying through by the denominator of \eqref{eqn:appendix_int_belief_general}, and using $\hat\phi_\theta$ as defined in \eqref{eqn:hatphi}, gives the characterization equation desired. The same argument as Theorem 1 establishes that the optimal policy brainstorms at the smallest such solution.
\end{proof}

\end{document}